\documentclass[a4paper,onecolumn,superscriptaddress,11pt,accepted=2020-04-27]{quantumarticle}
\pdfoutput=1
\usepackage[utf8]{inputenc}
\usepackage[english]{babel}
\usepackage[T1]{fontenc}
\usepackage{amsmath}
\usepackage{hyperref}
\usepackage[numbers,sort&compress]{natbib}

\usepackage{tikzit}
\input{zx.tikzdefs}

\tikzstyle{box}=[shape=rectangle, text height=1.5ex, text depth=0.25ex, yshift=0.5mm, fill=white, draw=black, minimum height=5mm, yshift=-0.5mm, minimum width=5mm, font={\small}]
\tikzstyle{Z dot}=[inner sep=0mm, minimum size=2mm, shape=circle, draw=black, fill={rgb,255: red,221; green,255; blue,221}]
\tikzstyle{Z phase dot}=[minimum size=5mm, font={\footnotesize\boldmath}, shape=rectangle, rounded corners=2mm, inner sep=0.2mm, outer sep=-2mm, scale=0.8, tikzit shape=circle, draw=black, fill={rgb,255: red,221; green,255; blue,221}, tikzit draw=blue]
\tikzstyle{X dot}=[Z dot, shape=circle, draw=black, fill={rgb,255: red,255; green,136; blue,136}]
\tikzstyle{X phase dot}=[Z phase dot, tikzit shape=circle, tikzit draw=blue, fill={rgb,255: red,255; green,136; blue,136}, font={\footnotesize\boldmath}]
\tikzstyle{hadamard}=[fill=yellow, draw=black, shape=rectangle, inner sep=0.6mm, minimum height=1.5mm, minimum width=1.5mm]
\tikzstyle{vertex}=[inner sep=0mm, minimum size=1mm, shape=circle, draw=black, fill=black]
\tikzstyle{vertex set}=[inner sep=0mm, minimum size=1mm, shape=circle, draw=black, fill=white, font={\footnotesize\boldmath}]

\tikzstyle{hadamard edge}=[-, dashed, dash pattern=on 2pt off 0.5pt, thick, draw={rgb,255: red,68; green,136; blue,255}]
\tikzstyle{brace edge}=[-, tikzit draw=blue, decorate, decoration={brace,amplitude=1mm,raise=-1mm}]
\tikzstyle{diredge}=[->]

\usepackage{xspace}
\usepackage{algorithm}
\usepackage{algorithmicx}
\usepackage[noend]{algpseudocode}

\newcommand{\symd}{\mathbin{\Delta}\xspace}
\newcommand{\Symdi}[1]{\underset{\scriptstyle #1}{\scalebox{1.5}{$\symd$}}\,}
\newcommand{\symdi}[1]{{\scalebox{1.5}{$\symd$}}_{#1}\,}

\makeatletter
\newcommand\etc{etc\@ifnextchar.{}{.\@}\xspace}

\makeatother

\newcommand{\odd}[2]{\textsf{Odd}_{#1}\left(#2\right)}
\newcommand{\codd}[2]{\textsf{Odd}_{#1}\left[#2\right]}

\newcommand{\CZ}{\ensuremath{\textrm{CZ}}\xspace}
\newcommand{\CX}{\ensuremath{\textrm{CNOT}}\xspace}

\newcommand{\CNOT}{\CX}

\newcommand{\UG}[1]{\ensuremath{G(#1)}\xspace}

\usepackage{bm}

\newcommand{\bra}[1]{\ensuremath{\left\langle #1 \right|}}
\newcommand{\ket}[1]{\ensuremath{\left|  #1 \right\rangle}}

\newcommand{\ketbra}[2]{\ensuremath{\ket{#1}\!\bra{#2}}}
\usepackage{amsmath,amsthm,amssymb}
\theoremstyle{definition}
\newtheorem{theorem}{Theorem}[section]
\newtheorem{corollary}[theorem]{Corollary}
\newtheorem{lemma}[theorem]{Lemma}
\newtheorem*{lemma*}{Lemma}
\newtheorem*{proposition*}{Proposition}
\newtheorem{proposition}[theorem]{Proposition}

\newtheorem{definition}[theorem]{Definition}
\newtheorem{example}[theorem]{Example}
\newtheorem{remark}[theorem]{Remark}

\usepackage[color,leftbars]{changebar}

\title{Graph-theoretic Simplification of Quantum Circuits with the ZX-calculus}

\author{Ross Duncan}
\affiliation{University of Strathclyde, 26 Richmond Street, Glasgow G1
1XH, UK}  
\affiliation{Cambridge Quantum Computing Ltd, 9a Bridge Street,
  Cambridge CB2 1UB, UK}
\email{ross.duncan@strath.ac.uk}
\orcid{0000-0001-6758-1573}

\author{Aleks Kissinger}
\affiliation{Department of Computer Science, University of Oxford}
\email{aleks.kissinger@cs.ox.ac.uk}
\orcid{0000-0002-6090-9684}

\author{Simon Perdrix}
\affiliation{CNRS LORIA, Inria-MOCQUA, Universit\'e de Lorraine, F 54000 Nancy,
France}
\email{simon.perdrix@loria.fr}
\orcid{0000-0002-1808-2409}

\author{John van de Wetering}
\affiliation{Institute for Computing and Information Sciences, Radboud University Nijmegen}
\email{john@vdwetering.name}
\orcid{0000-0002-5405-8959}

\begin{document}






\maketitle

\begin{abstract}
We present a completely new approach to quantum circuit optimisation, based on the ZX-calculus. We first interpret quantum circuits as ZX-diagrams, which provide a flexible, lower-level language for describing quantum computations graphically. Then, using the rules of the ZX-calculus, we give a simplification strategy for ZX-diagrams based on the two graph transformations of local complementation and pivoting and show that the resulting reduced diagram can be transformed back into a quantum circuit. While little is known about extracting circuits from arbitrary ZX-diagrams, we show that the underlying graph of our simplified ZX-diagram always has a graph-theoretic property called generalised flow, which in turn yields a deterministic circuit extraction procedure. For Clifford circuits, this extraction procedure yields a new normal form that is both asymptotically optimal in size and gives a new, smaller upper bound on gate depth for nearest-neighbour architectures. For Clifford+T and more general circuits, our technique enables us to to `see around' gates that obstruct the Clifford structure and produce smaller circuits than na\"ive `cut-and-resynthesise' methods.
\end{abstract}


\section{Introduction}

Quantum circuits provide a \textit{de facto} `assembly language' for quantum computation, in which computations are described as the composition of many simple unitary linear maps called \textit{quantum gates}. An important topic in the study of quantum circuits is \textit{quantum circuit optimisation}, whereby quantum circuits which realise a given computation are transformed into new circuits involving fewer or simpler gates.
While some strides have already been made in this area, the field is still relatively undeveloped. Most approaches to quantum circuit optimisation are based on only a handful of techniques: gate substitutions, computation of small (pseudo-)normal forms for special families of circuits~\cite{aaronsongottesman2004,markov2008optimal,CliffOpt}, optimisation of phase polynomials~\cite{amy2014polynomial,heyfron2018efficient}, or some combination thereof~\cite{abdessaied2014quantum,nam2018automated}.

This paper lays a theoretical foundation for a completely new quantum circuit optimisation technique based on the \zxcalculus~\cite{CD2}. A key point of departure of our technique is that we break the rigid structure of quantum circuits and perform reductions on a lower-level, string diagram-based representation of a quantum circuit called a \textit{\zxdiagram}. These diagrams are more flexible than circuits, in that they can be deformed arbitrarily, and are subject to a rich equational theory: the \zxcalculus.
The core rules of the \zxcalculus give a sound and complete~\cite{Backens1} theory for \textit{Clifford circuits}, a well-known class of circuits that can be efficiently classically simulated. More surprisingly, it was shown in 2018 that modest extensions to the \zxcalculus suffice to give completeness for families of circuits that are approximately universal~\cite{SimonCompleteness,ZXNormalForm} and exactly universal~\cite{HarnyAmarCompleteness,JPV-universal,euler-zx,carette2019completeness} for quantum computation.



Since \zxdiagrams provide a more flexible representation of a quantum computation than a circuit, we can derive simplifications of \zxdiagrams that have no quantum circuit analogue. 
However, this added flexibility comes at a price: while any quantum circuit can be interpreted as a \zxdiagram by decomposing gates into smaller pieces, the converse is not true. For a generic \zxdiagram corresponding to a unitary, there is no known general-purpose procedure for efficiently recovering a quantum circuit. Hence, an important part of our optimisation procedure is keeping enough information about the quantum circuit structure to get a circuit back at the end. Schematically:
\[ \scalebox{0.9}{\begin{tikzpicture}
	\begin{pgfonlayer}{nodelayer}
		\node [style=none] (0) at (-3.25, 2.75) {};
		\node [style=none] (1) at (-3.25, -2.5) {};
		\node [style=none] (2) at (-6.5, 0) {};
		\node [style=none] (3) at (0, 0) {};
		\node [style=none] (4) at (0, 3.5) {\zxdiagrams};
		\node [style=none] (5) at (0, 4.5) {};
		\node [style=none] (6) at (0, -3.5) {};
		\node [style=none] (7) at (-8, 0.5) {};
		\node [style=none] (8) at (7.75, 0.5) {};
		\node [style=none] (9) at (-3.25, 0) {Quantum circuits};
		\node [style=none] (11) at (-2.25, 1.5) {};
		\node [style=vertex] (12) at (5.25, 0) {};
		\node [style=none] (13) at (2.5, 2) {\color{blue!80!black}\it simplification};
		\node [style=none] (14) at (-2.25, -1.25) {};
		\node [style=none] (15) at (2.5, -1.75) {\color{blue!80!black}\it circuit extraction};
	\end{pgfonlayer}
	\begin{pgfonlayer}{edgelayer}
		\draw [in=180, out=90] (2.center) to (0.center);
		\draw [in=90, out=0] (0.center) to (3.center);
		\draw [in=0, out=-90] (3.center) to (1.center);
		\draw [in=270, out=180] (1.center) to (2.center);
		\draw [in=180, out=90] (7.center) to (5.center);
		\draw [in=90, out=0] (5.center) to (8.center);
		\draw [in=0, out=-90] (8.center) to (6.center);
		\draw [in=270, out=180] (6.center) to (7.center);
		\draw [style=diredge, very thick, draw={blue!80!black}, shorten >=1mm, in=135, out=0, looseness=0.75] (11.center) to (12);
		\draw [style=diredge, very thick, draw={blue!80!black}, shorten <=1mm, in=0, out=-135, looseness=0.75] (12) to (14.center);
	\end{pgfonlayer}
\end{tikzpicture}} \]

Until recently, this extraction step was poorly understood, and the only techniques for doing circuit-to-circuit translation with the \zxcalculus did so without departing from the overall structure of the original quantum circuit~\cite{FaganDuncan}, avoiding the extraction problem altogether. In this paper we adopt a more ambitious approach. First, building on prior work of two of the authors~\cite{DP1,DP2,DP3}, we use the rules of the \zxcalculus to derive a sound, terminating simplification procedure for \zxdiagrams. The key simplification steps involve the graph-theoretic transformations of \emph{local complementation} and \emph{pivoting}~\cite{kotzig,Bouchet87}, 
which allow certain generators to be deleted from \zxdiagrams one-by-one or in pairs, respectively. When applied to Clifford circuits, the diagram resulting from the simplification procedure is represented in the form of a \emph{graph-state with local Cliffords} (GS-LC)~\cite{hein2006entanglement}, a pseudo-normal form for representing Clifford circuits whose size is at most quadratic in the number of qubits. Hence, one side effect of our simplification procedure is that it gives a simple, graph-theoretic alternative to the normalisation of a ZX-diagram to GS-LC form proposed by Backens~\cite{Backens1}. For non-Clifford circuits, the simplified \zxdiagram represents a `skeleton' of the circuit we started with, consisting only of generators arising from non-Clifford phase gates and their nearest neighbours. Although this is no longer a canonical form, it can still be significantly smaller than the input circuit, especially when there is a relatively low proportion of non-Clifford gates. We then show that, even though this simplification breaks the circuit structure, it preserves a graph-theoretic invariant called \textit{focused gFlow}~\cite{GFlow,mhalla2011graph}, from which we can derive an efficient circuit extraction strategy. We demonstrate the full simplify-and-extract procedure by means of a running example, which is also available online as a Jupyter notebook.\footnote{\href{https://nbviewer.jupyter.org/github/Quantomatic/pyzx/blob/906f6db3/demos/example-gtsimp.ipynb}{\color{blue!80!black}\texttt{nbviewer.jupyter.org/github/Quantomatic/pyzx/blob/906f6db3/demos/example-gtsimp.ipynb}} Note this link is read-only. To run the notebook interactively, download the file and open it in Jupyter (with PyZX installed). The simplest way to do this is to clone or download the PyZX repostory from: \href{https://github.com/quantomatic/pyzx}{\color{blue!80!black}\texttt{github.com/quantomatic/pyzx}} and find the notebook in \texttt{demos}.}

In the case of Clifford circuits, this procedure will produce circuits comparable in size to those described by standard techniques~\cite{aaronsongottesman2004,nest2010clifford}. In the non-Clifford case, this can already find more simplifications than na\"ively `cutting' the circuit and simplifying purely-Clifford sections. More importantly, this paper establishes a new theoretical framework upon which to build powerful new optimisation techniques, e.g. for important tasks such as T-count reduction. Indeed an ancilla-free T-count reduction technique based on the framework of this paper has recently been introduced in Ref.~\cite{zxtcount} and implemented in a tool called PyZX~\cite{pyzx}. At the time of publication, this new technique matched or out-performed the state of the art on 72\% of the benchmark circuits tested, in some cases decreasing the T-count by as much as 50\%.\footnote{The `PyZX-only' T-counts from~\cite{zxtcount} were subsequently matched by parallel, independent work of Zhang and Chen~\cite{zhang2019tgates}.}



The paper is organised as follows. After giving a brief introduction to the \zxcalculus in Section~\ref{sec:zx}, we will introduce \emph{graph-like} \zxdiagrams in Section~\ref{sec:graph-like-zx}. 
This mild restriction to the family of all \zxdiagrams will enable us to speak more readily about graph-theoretic properties of the diagram, like the existence of a focused gFlow. 
In Section~\ref{sec:lcomppivot} we will show that both local complementation and pivoting preserve focused gFlow. 
Then in Section~\ref{sec:simp} we will derive a simplification routine using these graph-theoretic notions.
In Section~\ref{sec:clifford} we study the properties of Clifford diagrams resulting from this simplification routine, and show how these can be transformed into circuits.
Then in Section~\ref{sec:cliffordT} we show how general diagrams produced by this routine can be extracted into a circuit,
and give some experimental evidence that this method performs better
than only optimizing Clifford sub-circuits. Finally, in
Section~\ref{sec:conclusion} we conclude and discuss future work.

\bigskip

\noindent \textit{\textbf{Note:} The results in Proposition~\ref{prop:cliff-gate-depth} concerning Clifford gate count and depth were added to a version of this paper submitted for peer review in October 2019. Prior to us making this version public, Bravyi and Maslov independently reported similar results in Ref.~\cite{bravyi2020hadamard}.}

\section{The \zxcalculus and quantum circuits}\label{sec:zx}

We will provide a brief overview of the \zxcalculus. For an in-depth
reference see Ref.~\cite{CKbook}.

The \zxcalculus is a diagrammatic language similar to the familiar
quantum circuit notation.  A \emph{\zxdiagram} consists of \emph{wires} and \emph{spiders}.  Wires entering the diagram from the left are \emph{inputs}; wires exiting to
the right are \emph{outputs}.  Given two diagrams we can compose them
by joining the outputs of the first to the inputs of the second, or
form their tensor product by simply stacking the two diagrams.

Spiders are linear maps which can have any number of input or output
wires.  There are two varieties: $Z$ spiders depicted as green dots and $X$ spiders depicted as red dots.\footnote{If you are reading this
  document in monochrome or otherwise have difficulty distinguishing green and red, $Z$ spiders will appear lightly-shaded and $X$ spiders darkly-shaded.} Written explicitly in Dirac notation, these linear maps are:
\[\textrm{
\small
$\begin{tikzpicture}
	\begin{pgfonlayer}{nodelayer}
		\node [style=Z phase dot] (0) at (0, 0) {$\alpha$};
		\node [style=none] (1) at (1.25, 1) {};
		\node [style=none] (2) at (-1.25, 1) {};
		\node [style=none] (3) at (-1.25, -1) {};
		\node [style=none] (4) at (1.25, -1) {};
		\node [style=none] (5) at (1.25, 0.5) {};
		\node [style=none] (6) at (-1.25, 0.5) {};
		\node [style=none, rotate=90] (7) at (-1, -0.25) {...};
		\node [style=none, rotate=90] (8) at (1, -0.25) {...};
	\end{pgfonlayer}
	\begin{pgfonlayer}{edgelayer}
		\draw [in=-141, out=0, looseness=0.75] (3.center) to (0);
		\draw [in=180, out=-39, looseness=0.75] (0) to (4.center);
		\draw [in=180, out=22, looseness=0.75] (0) to (5.center);
		\draw [in=180, out=39, looseness=0.75] (0) to (1.center);
		\draw [in=0, out=158, looseness=0.75] (0) to (6.center);
		\draw [in=141, out=0, looseness=0.75] (2.center) to (0);
	\end{pgfonlayer}
\end{tikzpicture} \ := \ \ketbra{\textrm{$0$...$0$}}{\textrm{$0$...$0$}} +
e^{i \alpha} \ketbra{\textrm{$1$...$1$}}{\textrm{$1$...$1$}} \hfill
\qquad
\hfill \begin{tikzpicture}
	\begin{pgfonlayer}{nodelayer}
		\node [style=X phase dot] (0) at (0, 0) {$\alpha$};
		\node [style=none] (1) at (1.25, 1) {};
		\node [style=none] (2) at (-1.25, 1) {};
		\node [style=none] (3) at (-1.25, -1) {};
		\node [style=none] (4) at (1.25, -1) {};
		\node [style=none] (5) at (1.25, 0.5) {};
		\node [style=none] (6) at (-1.25, 0.5) {};
		\node [style=none, rotate=90] (7) at (-1, -0.25) {...};
		\node [style=none, rotate=90] (8) at (1, -0.25) {...};
	\end{pgfonlayer}
	\begin{pgfonlayer}{edgelayer}
		\draw [in=-141, out=0, looseness=0.75] (3.center) to (0);
		\draw [in=180, out=-39, looseness=0.75] (0) to (4.center);
		\draw [in=180, out=22, looseness=0.75] (0) to (5.center);
		\draw [in=180, out=39, looseness=0.75] (0) to (1.center);
		\draw [in=0, out=158, looseness=0.75] (0) to (6.center);
		\draw [in=141, out=0, looseness=0.75] (2.center) to (0);
	\end{pgfonlayer}
\end{tikzpicture} \ := \ \ketbra{\textrm{$+$...$+$}}{\textrm{$+$...$+$}} +
e^{i \alpha} \ketbra{\textrm{$-$...$-$}}{\textrm{$-$...$-$}}$
}
\]
Therefore a \zxdiagram with $m$ input wires and $n$ output wires represents a linear map $(\mathbb C^2)^{\otimes m} \to (\mathbb C^2)^{\otimes n}$ built from
spiders (and permutations of qubits) by composition and tensor product
of linear maps.  As a special case, diagrams with no inputs and $n$ outputs represent vectors in $(\mathbb C^2)^{\otimes n}$, i.e.
(unnormalised) $n$-qubit states.

\begin{example}\label{ex:basic-maps-and-states}
  We can immediately write down some simple state preparations and
  unitaries in the \zxcalculus:
  \[
  \begin{array}{rclcrcl}
  \begin{tikzpicture}
	\begin{pgfonlayer}{nodelayer}
		\node [style=Z dot] (0) at (-0.5, 0) {};
		\node [style=none] (1) at (0.5, 0) {};
	\end{pgfonlayer}
	\begin{pgfonlayer}{edgelayer}
		\draw (0) to (1.center);
	\end{pgfonlayer}
\end{tikzpicture} & = & \ket{0} + \ket{1} \ \propto \ket{+} &
  \qquad\qquad &
  \begin{tikzpicture}
	\begin{pgfonlayer}{nodelayer}
		\node [style=X dot] (0) at (-0.5, 0) {};
		\node [style=none] (1) at (0.5, 0) {};
	\end{pgfonlayer}
	\begin{pgfonlayer}{edgelayer}
		\draw (0) to (1.center);
	\end{pgfonlayer}
\end{tikzpicture} & = & \ket{+} + \ket{-} \ \propto \ket{0} \\
  &\quad& & & \quad \\
  \begin{tikzpicture}
	\begin{pgfonlayer}{nodelayer}
		\node [style=Z phase dot] (0) at (0, 0) {$\alpha$};
		\node [style=none] (1) at (-1.25, 0) {};
		\node [style=none] (2) at (1.25, 0) {};
	\end{pgfonlayer}
	\begin{pgfonlayer}{edgelayer}
		\draw (1.center) to (0);
		\draw (0) to (2.center);
	\end{pgfonlayer}
\end{tikzpicture} & = & \ketbra{0}{0} + e^{i \alpha} \ketbra{1}{1} =
  Z_\alpha &
  & 
  \begin{tikzpicture}
	\begin{pgfonlayer}{nodelayer}
		\node [style=X phase dot] (0) at (0, 0) {$\alpha$};
		\node [style=none] (1) at (-1, 0) {};
		\node [style=none] (2) at (1, 0) {};
	\end{pgfonlayer}
	\begin{pgfonlayer}{edgelayer}
		\draw (1.center) to (0);
		\draw (0) to (2.center);
	\end{pgfonlayer}
\end{tikzpicture}
 & = & \ketbra{+}{+} + e^{i \alpha} \ketbra{-}{-} = X_\alpha
  \end{array}
  \]
  In particular we have the Pauli matrices:
  \[
  \hfill
  \begin{tikzpicture}
	\begin{pgfonlayer}{nodelayer}
		\node [style=Z phase dot] (0) at (0, 0) {$\pi$};
		\node [style=none] (1) at (-1.25, 0) {};
		\node [style=none] (2) at (1.25, 0) {};
	\end{pgfonlayer}
	\begin{pgfonlayer}{edgelayer}
		\draw (1.center) to (0);
		\draw (0) to (2.center);
	\end{pgfonlayer}
\end{tikzpicture} = Z \qquad\qquad   \begin{tikzpicture}
	\begin{pgfonlayer}{nodelayer}
		\node [style=X phase dot] (0) at (0, 0) {$\pi$};
		\node [style=none] (1) at (-1.25, 0) {};
		\node [style=none] (2) at (1.25, 0) {};
	\end{pgfonlayer}
	\begin{pgfonlayer}{edgelayer}
		\draw (1.center) to (0);
		\draw (0) to (2.center);
	\end{pgfonlayer}
\end{tikzpicture} = X \qquad\qquad
  \hfill
  \]
\end{example}
It will be convenient to introduce a symbol -- a yellow square -- for
the Hadamard gate. This is defined by the equation:
\begin{equation}\label{eq:Hdef}
\hfill
\begin{tikzpicture}
	\begin{pgfonlayer}{nodelayer}
		\node [style=none] (0) at (-2.75, 0) {};
		\node [style=none] (1) at (-0.75, 0) {};
		\node [style=hadamard] (2) at (-1.75, 0) {};
		\node [style=none] (3) at (0, 0) {$=$};
		\node [style=none] (4) at (0.75, 0) {};
		\node [style=Z phase dot] (5) at (1.5, 0) {$\frac\pi2$};
		\node [style=X phase dot] (6) at (2.5, 0) {$\frac\pi2$};
		\node [style=Z phase dot] (7) at (3.5, 0) {$\frac\pi2$};
		\node [style=none] (8) at (4.25, 0) {};
	\end{pgfonlayer}
	\begin{pgfonlayer}{edgelayer}
		\draw (0.center) to (2);
		\draw (2) to (1.center);
		\draw (4.center) to (8.center);
	\end{pgfonlayer}
\end{tikzpicture}
\hfill
\end{equation}
We will often use an alternative notation to simplify the diagrams,
and replace a Hadamard between two spiders by a blue dashed edge, as
illustrated below.
\ctikzfig{blue-edge-def} 
Both the blue edge notation and the Hadamard box can always be
translated back into spiders when necessary. We will refer to the blue
edge as a \emph{Hadamard edge}.

Two diagrams are considered \emph{equal} when one can be deformed to
the other by moving the vertices around in the plane, bending,
unbending, crossing, and uncrossing wires, so long as the connectivity
and the order of the inputs and outputs is maintained.  
Furthermore, there is an additional set of equations that we call the \emph{rules} of the \zxcalculus; these are shown in
Figure~\ref{fig:zx-rules}.

\begin{remark}
We neglect (non-zero) scalar factors in the rules in Figure~\ref{fig:zx-rules}. That is, if we are able to prove two ZX-diagrams are equal using the rules of Figure~\ref{fig:zx-rules}, then their associated matrices satisfy $A = zB$ for $z \in \mathbb C\backslash\{0\}$. It is possible to give a presentation of the ZX-calculus that accounts for scalar factors (see e.g.~\cite{Backens:2015aa}), but for our purposes it will not be necessary. This is because the inputs and outputs of our simplification procedure are quantum circuits, which are unitary by construction. If $A$ and $B$ are unitary, $A = zB \implies |z| = 1$, so neglecting scalars will (at worst) produce and output that differs from the input by a global phase.
\end{remark}

\begin{figure}
\centering
\begin{tikzpicture}
	\begin{pgfonlayer}{nodelayer}
		\node [style=Z phase dot] (0) at (-1.5, -0.625) {$\beta$};
		\node [style=none] (1) at (-0.5, -0.375) {};
		\node [style=none] (2) at (-3.25, -0.375) {};
		\node [style=none] (3) at (-3.25, -1.625) {};
		\node [style=none] (4) at (-0.5, -1.625) {};
		\node [style=none, rotate=90] (7) at (-3.5, -1) {...};
		\node [style=none, rotate=90] (8) at (-0.5, -1) {...};
		\node [style=Z phase dot] (10) at (-3, 0.875) {$\alpha$};
		\node [style=none] (11) at (-1, 1.625) {};
		\node [style=none] (12) at (-4, 1.625) {};
		\node [style=none] (13) at (-1, 0.375) {};
		\node [style=none, rotate=90] (14) at (-1, 1) {...};
		\node [style=none] (16) at (-4, 0.375) {};
		\node [style=none, rotate=90] (17) at (-3.75, 1) {...};
		\node [style=none] (18) at (-0.5, 0.375) {};
		\node [style=none] (19) at (-0.5, 1.625) {};
		\node [style=none] (22) at (-4, -1.625) {};
		\node [style=none] (23) at (-4, -0.375) {};
		\node [style=none] (24) at (0.5, 0.25) {$=$};
		\node [style=none, rotate=90] (25) at (-2.35, 0.05) {...};
		\node [style=none] (26) at (1.5, 1.25) {};
		\node [style=none, rotate=90] (28) at (2, 0) {...};
		\node [style=none] (29) at (5, -1.25) {};
		\node [style=none, rotate=90] (30) at (4.75, 0) {...};
		\node [style=Z phase dot] (32) at (3.25, 0) {$\ \alpha\!+\!\beta\ $};
		\node [style=none] (33) at (5, 1.25) {};
		\node [style=none] (34) at (1.5, -1.25) {};
		\node [style=none] (35) at (0.5, 1) {\spiderrule};
		\node [style=Z phase dot] (36) at (2.75, -4.25) {$\ -\alpha\ $};
		\node [style=none] (37) at (5, -3.125) {};
		\node [style=none] (38) at (0.5, -4.25) {$=$};
		\node [style=none] (39) at (1.5, -4.25) {};
		\node [style=none] (40) at (5, -5.375) {};
		\node [style=X phase dot] (42) at (4.25, -5.375) {$\pi$};
		\node [style=X phase dot] (43) at (-2.75, -4.25) {$\pi$};
		\node [style=none] (44) at (-0.25, -3.125) {};
		\node [style=Z phase dot] (45) at (-1.75, -4.25) {$\alpha$};
		\node [style=none, rotate=90] (46) at (-0.5, -4.25) {...};
		\node [style=none, rotate=90] (48) at (4.25, -4.25) {...};
		\node [style=none] (50) at (-3.5, -4.25) {};
		\node [style=none] (51) at (-0.25, -5.375) {};
		\node [style=X phase dot] (52) at (4.25, -3.125) {$\pi$};
		\node [style=none] (53) at (0.5, -3.5) {\pirule};
		\node [style=X dot] (54) at (12, -3.375) {};
		\node [style=none, rotate=90] (55) at (12.75, -4.25) {...};
		\node [style=Z phase dot] (57) at (8.25, -4.25) {$\alpha$};
		\node [style=none] (58) at (9.5, -5.125) {};
		\node [style=none] (59) at (10.5, -4.25) {$=$};
		\node [style=none] (60) at (13, -5.125) {};
		\node [style=none] (62) at (13, -3.375) {};
		\node [style=none, rotate=90] (63) at (9.25, -4.25) {...};
		\node [style=X dot] (65) at (12, -5.125) {};
		\node [style=none] (66) at (10.5, -3.5) {\copyrule};
		\node [style=none] (67) at (9.5, -3.375) {};
		\node [style=X dot] (68) at (7.25, -4.25) {};
		\node [style=hadamard] (69) at (7.5, 1) {};
		\node [style=none, rotate=90] (72) at (7.5, 0) {...};
		\node [style=none] (73) at (10.75, 0) {$=$};
		\node [style=hadamard] (74) at (7.5, -1) {};
		\node [style=none] (77) at (7, 1) {};
		\node [style=none] (81) at (11.75, -1) {};
		\node [style=none] (82) at (7, -1) {};
		\node [style=none, rotate=90] (83) at (12, 0) {...};
		\node [style=none] (84) at (11.75, 1) {};
		\node [style=none] (86) at (10.75, 0.75) {\hadamardrule};
		\node [style=Z dot] (87) at (16.5, 0.75) {};
		\node [style=none] (88) at (15.75, 0.75) {};
		\node [style=none] (89) at (17.25, 0.75) {};
		\node [style=none] (90) at (19.25, 0.75) {};
		\node [style=none] (91) at (20.25, 0.75) {};
		\node [style=none] (92) at (18.25, 1.5) {\identityrule};
		\node [style=none] (93) at (18.25, 0.75) {$=$};
		\node [style=none] (94) at (18.25, -1.75) {$=$};
		\node [style=none] (95) at (18.25, -1) {\hcancelrule};
		\node [style=none] (96) at (20.25, -1.75) {};
		\node [style=none] (97) at (17.25, -1.75) {};
		\node [style=none] (98) at (15.75, -1.75) {};
		\node [style=none] (99) at (19.25, -1.75) {};
		\node [style=hadamard] (100) at (16.25, -1.75) {};
		\node [style=hadamard] (101) at (16.75, -1.75) {};
		\node [style=none] (102) at (18.25, -3.5) {\bialgrule};
		\node [style=X dot] (103) at (21.25, -5) {};
		\node [style=none] (104) at (22, -3.5) {};
		\node [style=X dot] (105) at (15.5, -4.25) {};
		\node [style=Z dot] (106) at (16.75, -4.25) {};
		\node [style=none] (107) at (19, -5) {};
		\node [style=none] (108) at (14.75, -3.5) {};
		\node [style=none] (109) at (19, -3.5) {};
		\node [style=Z dot] (110) at (19.75, -3.5) {};
		\node [style=none] (111) at (17.5, -3.5) {};
		\node [style=none] (112) at (18.25, -4.25) {$=$};
		\node [style=none] (113) at (14.75, -5) {};
		\node [style=X dot] (114) at (21.25, -3.5) {};
		\node [style=none] (115) at (22, -5) {};
		\node [style=Z dot] (116) at (19.75, -5) {};
		\node [style=none] (117) at (17.5, -5) {};
		\node [style=hadamard] (118) at (9.5, 1) {};
		\node [style=none, rotate=90] (120) at (9.5, 0) {...};
		\node [style=hadamard] (121) at (9.5, -1) {};
		\node [style=Z phase dot] (122) at (8.5, 0) {$\alpha$};
		\node [style=none] (123) at (10, 1) {};
		\node [style=none] (126) at (10, -1) {};
		\node [style=X phase dot] (127) at (13, 0) {$\alpha$};
		\node [style=none] (128) at (14.25, -1) {};
		\node [style=none, rotate=90] (129) at (14, 0) {...};
		\node [style=none] (130) at (14.25, 1) {};
	\end{pgfonlayer}
	\begin{pgfonlayer}{edgelayer}
		\draw [in=-141, out=0, looseness=0.75] (3.center) to (0);
		\draw [in=180, out=-39, looseness=0.75] (0) to (4.center);
		\draw [in=180, out=39, looseness=0.75] (0) to (1.center);
		\draw [in=180, out=0] (2.center) to (0);
		\draw [in=-141, out=0, looseness=0.75] (16.center) to (10);
		\draw [in=180, out=0] (10) to (13.center);
		\draw [in=180, out=39, looseness=0.75] (10) to (11.center);
		\draw [in=141, out=0, looseness=0.75] (12.center) to (10);
		\draw [bend left=45] (10) to (0);
		\draw (11.center) to (19.center);
		\draw (13.center) to (18.center);
		\draw (23.center) to (2.center);
		\draw (22.center) to (3.center);
		\draw [bend left=45] (0) to (10);
		\draw [in=-120, out=0] (34.center) to (32);
		\draw [in=180, out=-60] (32) to (29.center);
		\draw [in=180, out=60] (32) to (33.center);
		\draw [in=120, out=0] (26.center) to (32);
		\draw [in=180, out=-75, looseness=0.75] (45) to (51.center);
		\draw [in=180, out=75, looseness=0.75] (45) to (44.center);
		\draw [in=180, out=0, looseness=0.50] (43) to (45);
		\draw (50.center) to (43);
		\draw [in=-60, out=180, looseness=0.75] (42) to (36);
		\draw [in=0, out=180, looseness=0.75] (36) to (39.center);
		\draw [in=60, out=180, looseness=0.75] (52) to (36);
		\draw (37.center) to (52);
		\draw (40.center) to (42);
		\draw [in=180, out=-75, looseness=0.75] (57) to (58.center);
		\draw [in=180, out=75, looseness=0.75] (57) to (67.center);
		\draw [in=180, out=0, looseness=0.50] (68) to (57);
		\draw (62.center) to (54);
		\draw (60.center) to (65);
		\draw (77.center) to (69);
		\draw (82.center) to (74);
		\draw (88.center) to (89.center);
		\draw (90.center) to (91.center);
		\draw (98.center) to (97.center);
		\draw (99.center) to (96.center);
		\draw (116) to (114);
		\draw (103) to (110);
		\draw (110) to (114);
		\draw (116) to (103);
		\draw (115.center) to (103);
		\draw (107.center) to (116);
		\draw (110) to (109.center);
		\draw (114) to (104.center);
		\draw [bend right] (113.center) to (105);
		\draw [bend right] (105) to (108.center);
		\draw (105) to (106);
		\draw [bend right] (106) to (117.center);
		\draw [bend left] (106) to (111.center);
		\draw [in=-60, out=180, looseness=0.75] (121) to (122);
		\draw [in=75, out=180, looseness=0.75] (118) to (122);
		\draw (123.center) to (118);
		\draw (126.center) to (121);
		\draw [in=-120, out=0, looseness=0.75] (74) to (122);
		\draw [in=105, out=0, looseness=0.75] (69) to (122);
		\draw [in=180, out=-60, looseness=0.75] (127) to (128.center);
		\draw [in=180, out=75, looseness=0.75] (127) to (130.center);
		\draw [in=0, out=-120, looseness=0.75] (127) to (81.center);
		\draw [in=0, out=105, looseness=0.75] (127) to (84.center);
		\draw [in=105, out=-15, looseness=0.75] (10) to (0);
	\end{pgfonlayer}
\end{tikzpicture}
\caption{
A convenient presentation for the ZX-calculus. These rules hold for all $\alpha, \beta \in [0, 2 \pi)$, and due to $(\bm{h})$ and $(\bm{i2})$ all rules also hold with the colours interchanged. Note `...' should be read as `0 or more', hence the spiders on the LHS of \SpiderRule are connected by one or more wires.}
\label{fig:zx-rules}
\end{figure}


Let us derive two additional rules, known as the
\emph{antipode} rule and the \emph{$\pi$-copy} rule.

\begin{lemma}\label{lem:hopf-law}
  The antipode rule, 
  \ctikzfig{hopf-rule}
  is derivable in the \zxcalculus.
  \begin{proof}
    To derive the \HopfRule rule we take advantage of the freedom to
    deform the diagram as shown below.
    \ctikzfig{hopf-proof}
    The final equation is simply dropping the scalar.
  \end{proof}
\end{lemma}

\begin{lemma}\label{lem:pi-state-copy}
  The $\pi$-copy rule,
   \ctikzfig{picopy-rule}
   is derivable in the \zxcalculus.
  \begin{proof}
    To derive the \PiCopyRule rule we use the \SpiderRule rule first
    to split the $\pi$-labelled vertex, and the \PiRule and \CopyRule
    rules do the bulk of the work as shown below:
    \ctikzfig{picopy-proof}    
  \end{proof}  
\end{lemma}


\begin{remark}\label{rem:completeness}
  The \zxcalculus is \emph{universal} in the sense that any linear map can be expressed as a \zxdiagram. Furthermore, when restricted to \textit{Clifford \zxdiagrams}, i.e. diagrams whose phases are all multiples of $\pi/2$, the version we present in Figure~\ref{fig:zx-rules} is \emph{complete}. That is, for any two Clifford \zxdiagrams that describe the same linear map, there exists a series of rewrites transforming one into the other. Recent extensions to the calculus have been introduced which are complete for the larger \textit{Clifford+T} family of \zxdiagrams \cite{SimonCompleteness}, where phases are multiples of $\pi/4$, and for \textit{all} \zxdiagrams~\cite{HarnyAmarCompleteness,JPV-universal,euler-zx}.
\end{remark}

Quantum circuits can be translated into \zxdiagrams in a straightforward manner. 
We will take as our starting point circuits constructed
from the following universal set of gates:
\[
\CNOT \ :=\
\left(\begin{matrix}
  1 & 0 & 0 & 0 \\
  0 & 1 & 0 & 0 \\
  0 & 0 & 0 & 1 \\
  0 & 0 & 1 & 0 \\
\end{matrix}\right)
\qquad\qquad
Z_\alpha \ :=\
\left(\begin{matrix}
  1 & 0 \\
  0 & e^{i \alpha}
\end{matrix}\right)
\qquad\qquad
H \ :=\ \frac{1}{\sqrt{2}}
\left(\begin{matrix}
  1 & 1 \\
  1 & -1
\end{matrix}\right)
\]
We fix this gate
set because they admit a convenient representation in terms of
spiders: 
\begin{align}\label{eq:zx-gates}
\CNOT & = \begin{tikzpicture}
	\begin{pgfonlayer}{nodelayer}
		\node [style=Z dot] (0) at (0, 0.5) {};
		\node [style=none] (1) at (1.25, 0.5) {};
		\node [style=X dot] (2) at (0, -0.5) {};
		\node [style=none] (3) at (-1.25, 0.5) {};
		\node [style=none] (4) at (1.25, -0.5) {};
		\node [style=none] (5) at (-1.25, -0.5) {};
	\end{pgfonlayer}
	\begin{pgfonlayer}{edgelayer}
		\draw (0) to (2);
		\draw (0) to (3.center);
		\draw (0) to (1.center);
		\draw (5.center) to (2);
		\draw (2) to (4.center);
	\end{pgfonlayer}
\end{tikzpicture} &
Z_\alpha & = \begin{tikzpicture}
	\begin{pgfonlayer}{nodelayer}
		\node [style=Z phase dot] (0) at (0, 0) {$\alpha$};
		\node [style=none] (1) at (-1.25, 0) {};
		\node [style=none] (2) at (1.25, 0) {};
	\end{pgfonlayer}
	\begin{pgfonlayer}{edgelayer}
		\draw (1.center) to (0);
		\draw (0) to (2.center);
	\end{pgfonlayer}
\end{tikzpicture} &
H & = \begin{tikzpicture}
	\begin{pgfonlayer}{nodelayer}
		\node [style=none] (0) at (-0.25, 0) {};
		\node [style=none] (1) at (1.75, 0) {};
		\node [style=hadamard] (2) at (0.75, 0) {};
	\end{pgfonlayer}
	\begin{pgfonlayer}{edgelayer}
		\draw (0.center) to (2);
		\draw (2) to (1.center);
	\end{pgfonlayer}
\end{tikzpicture}
\end{align}
For the \CNOT gate, the green spider is the first (i.e. control) qubit and the red spider is the second (i.e. target) qubit. Other common gates can easily be expressed in terms of these gates. In particular, $S := Z_{\frac\pi2}$, $T := Z_{\frac\pi4}$ and:
\begin{align*}
X_\alpha & = \begin{tikzpicture}
	\begin{pgfonlayer}{nodelayer}
		\node [style=X phase dot] (0) at (0, 0) {$\alpha$};
		\node [style=none] (1) at (-1, 0) {};
		\node [style=none] (2) at (1, 0) {};
		\node [style=Z phase dot] (3) at (-4.25, 0) {$\alpha$};
		\node [style=none] (4) at (-5.5, 0) {};
		\node [style=none] (5) at (-3, 0) {};
		\node [style=none] (6) at (-2, 0) {$=$};
		\node [style=hadamard] (7) at (-3.5, 0) {};
		\node [style=hadamard] (8) at (-5, 0) {};
	\end{pgfonlayer}
	\begin{pgfonlayer}{edgelayer}
		\draw (1.center) to (0);
		\draw (0) to (2.center);
		\draw (4.center) to (3);
		\draw (3) to (5.center);
	\end{pgfonlayer}
\end{tikzpicture} &
\CZ & = \begin{tikzpicture}
	\begin{pgfonlayer}{nodelayer}
		\node [style=Z dot] (0) at (-2.75, 0.5) {};
		\node [style=none] (1) at (-1.5, 0.5) {};
		\node [style=X dot] (2) at (-2.75, -0.5) {};
		\node [style=none] (3) at (-4, 0.5) {};
		\node [style=none] (4) at (-1.5, -0.5) {};
		\node [style=none] (5) at (-4, -0.5) {};
		\node [style=hadamard] (6) at (-3.5, -0.5) {};
		\node [style=hadamard] (7) at (-2, -0.5) {};
		\node [style=none] (8) at (-0.75, 0) {$=$};
		\node [style=Z dot] (9) at (1, 0.5) {};
		\node [style=none] (10) at (2, 0.5) {};
		\node [style=Z dot] (11) at (1, -0.5) {};
		\node [style=none] (12) at (0, 0.5) {};
		\node [style=none] (13) at (2, -0.5) {};
		\node [style=none] (14) at (0, -0.5) {};
		\node [style=hadamard] (15) at (1, 0) {};
		\node [style=none] (16) at (2.75, 0) {$=$};
		\node [style=Z dot] (17) at (4.5, 0.5) {};
		\node [style=none] (18) at (5.5, 0.5) {};
		\node [style=Z dot] (19) at (4.5, -0.5) {};
		\node [style=none] (20) at (3.5, 0.5) {};
		\node [style=none] (21) at (5.5, -0.5) {};
		\node [style=none] (22) at (3.5, -0.5) {};
	\end{pgfonlayer}
	\begin{pgfonlayer}{edgelayer}
		\draw (0) to (2);
		\draw (0) to (3.center);
		\draw (0) to (1.center);
		\draw (5.center) to (2);
		\draw (2) to (4.center);
		\draw (9) to (11);
		\draw (9) to (12.center);
		\draw (9) to (10.center);
		\draw (14.center) to (11);
		\draw (11) to (13.center);
		\draw [style=hadamard edge] (17) to (19);
		\draw (17) to (20.center);
		\draw (17) to (18.center);
		\draw (22.center) to (19);
		\draw (19) to (21.center);
	\end{pgfonlayer}
\end{tikzpicture}
\end{align*}

\begin{remark}
  Note that the directions of the wires in the depictions of the \CNOT and \CZ gates are irrelevant, hence we can draw them vertically without ambiguity. This undirectedness of wires is a general property of \zxdiagrams, and from hence forth we will ignore wire directions without further comment. We will also freely draw wires entering or exiting the diagram from arbitrary directions if the interpretation (i.e. as an input or an output) is either irrelevant or clear from context. 
\end{remark}




\noindent For our purposes, we will define quantum circuits as a special case of \zxdiagrams.

\begin{definition}\label{def:circuit}
  A \emph{circuit} is a \zxdiagram generated by compositions and tensor products of the \zxdiagrams in equation~\eqref{eq:zx-gates}.
\end{definition}

Important subclasses of circuits are \textit{Clifford circuits}, sometimes called stabiliser circuits, which are obtained from compositions of only \CNOT, $H$, and $S$ gates. They are efficiently classically simulable, thanks to the \textit{Gottesman-Knill theorem}~\cite{aaronsongottesman2004}. A unitary is \textit{local Clifford} if it arises from a single-qubit Clifford circuit, i.e. a composition of $H$ and $S$ gates. The addition of $T$ gates yields \textit{Clifford+T circuits}, which are capable of approximating any $n$-qubit unitary to arbitrary precision, whereas the inclusion of $Z_\alpha$ gates for all $\alpha$ enables one to construct any unitary exactly~\cite{NielsenChuang}.

\section{Graph-like ZX-diagrams}\label{sec:graph-like-zx}

Before starting our simplification procedure, we transform \zxdiagrams into the following, convenient form. 



\begin{definition}\label{def:graph-form}
  A \zxdiagram is \emph{graph-like} when:
  \begin{enumerate}
    \item All spiders are Z-spiders.
    \item Z-spiders are only connected via Hadamard edges.
    \item There are no parallel Hadamard edges or self-loops.
    \item Every input or output is connected to a Z-spider and every Z-spider is connected to at most one input or output.
  \end{enumerate}
\end{definition}

A \zxdiagram is called a \textit{graph state} if it is graph-like, every spider is connected to an output, and there are no non-zero phases. This corresponds to the standard notion of graph state from the literature~\cite{hein2006entanglement}. Hence, graph-like \zxdiagrams are a generalisation of graph states to maps where we allow arbitrary phases and some \textit{interior} spiders, i.e. spiders not connected to an input or output.

\begin{lemma}\label{lem:all-zx-are-graph-like}
  Every \zxdiagram is equal to a graph-like \zxdiagram.
\end{lemma}
\begin{proof} Starting with an arbitrary \zxdiagram, we apply \HadamardRule to turn all red spiders into green spiders surrounded by Hadamard gates. We then remove excess Hadamards via \HCancelRule. Any non-Hadamard edge is then removed by fusing the adjacent spiders with \SpiderRule. Any parallel Hadamard edges or self-loops can be removed via the following 3 rules:
\begin{equation}\label{eq:parallel-edges-loops}
\begin{tikzpicture}
	\begin{pgfonlayer}{nodelayer}
		\node [style=Z phase dot] (0) at (-1.25, 0) {$\alpha$};
		\node [style=Z phase dot] (1) at (1.25, 0) {$\beta$};
		\node [style=none] (2) at (-2.5, 1) {};
		\node [style=none] (3) at (-2.5, -1) {};
		\node [style=none] (4) at (2.5, 1) {};
		\node [style=none] (5) at (2.5, -1) {};
		\node [style=none] (6) at (-2.5, 0) {...};
		\node [style=none] (7) at (2.5, 0) {...};
		\node [style=hadamard] (8) at (0, 0.75) {};
		\node [style=hadamard] (9) at (0, -0.75) {};
		\node [style=none] (10) at (3.5, 0) {=};
		\node [style=Z phase dot] (11) at (5.75, 0) {$\alpha$};
		\node [style=Z phase dot] (12) at (7.25, 0) {$\beta$};
		\node [style=none] (13) at (4.5, 1) {};
		\node [style=none] (14) at (4.5, -1) {};
		\node [style=none] (15) at (8.5, 1) {};
		\node [style=none] (16) at (8.5, -1) {};
		\node [style=none] (17) at (4.5, 0) {...};
		\node [style=none] (18) at (8.5, 0) {...};
	\end{pgfonlayer}
	\begin{pgfonlayer}{edgelayer}
		\draw (0) to (2.center);
		\draw (0) to (3.center);
		\draw (1) to (4.center);
		\draw (1) to (5.center);
		\draw [in=-180, out=60] (0) to (8);
		\draw [in=120, out=0] (8) to (1);
		\draw [in=0, out=-120] (1) to (9);
		\draw [in=-60, out=180] (9) to (0);
		\draw (11) to (13.center);
		\draw (11) to (14.center);
		\draw (12) to (15.center);
		\draw (12) to (16.center);
	\end{pgfonlayer}
\end{tikzpicture} \qquad
\begin{tikzpicture}
	\begin{pgfonlayer}{nodelayer}
		\node [style=Z phase dot] (0) at (-1.75, 0) {$\alpha$};
		\node [style=none] (1) at (-2.75, -0.75) {};
		\node [style=none] (2) at (-0.75, -0.75) {};
		\node [style=none] (3) at (-1.75, 1) {};
		\node [style=none] (4) at (-1.75, -0.75) {...};
		\node [style=none] (5) at (0, 0) {=};
		\node [style=none] (6) at (1.5, -0.75) {...};
		\node [style=none] (7) at (0.5, -0.75) {};
		\node [style=none] (8) at (2.5, -0.75) {};
		\node [style=Z phase dot] (9) at (1.5, 0) {$\alpha$};
	\end{pgfonlayer}
	\begin{pgfonlayer}{edgelayer}
		\draw (1.center) to (0);
		\draw (0) to (2.center);
		\draw [in=360, out=15, looseness=1.50] (0) to (3.center);
		\draw [in=540, out=165, looseness=1.50] (0) to (3.center);
		\draw (7.center) to (9);
		\draw (9) to (8.center);
	\end{pgfonlayer}
\end{tikzpicture} \qquad
\begin{tikzpicture}
	\begin{pgfonlayer}{nodelayer}
		\node [style=Z phase dot] (0) at (-1.75, 0) {$\alpha$};
		\node [style=none] (1) at (-2.75, -0.75) {};
		\node [style=none] (2) at (-0.75, -0.75) {};
		\node [style=none] (3) at (-1.75, 1) {};
		\node [style=none] (4) at (-1.75, -0.75) {...};
		\node [style=none] (5) at (0, 0) {=};
		\node [style=none] (6) at (1.75, -0.75) {...};
		\node [style=none] (7) at (0.75, -0.75) {};
		\node [style=none] (8) at (2.75, -0.75) {};
		\node [style=Z phase dot] (9) at (1.75, 0) {$\ \alpha + \pi\ $};
		\node [style=hadamard] (44) at (-1.75, 1) {};
	\end{pgfonlayer}
	\begin{pgfonlayer}{edgelayer}
		\draw (1.center) to (0);
		\draw (0) to (2.center);
		\draw [in=360, out=15, looseness=1.50] (0) to (3.center);
		\draw [in=540, out=165, looseness=1.50] (0) to (3.center);
		\draw (7.center) to (9);
		\draw (9) to (8.center);
	\end{pgfonlayer}
\end{tikzpicture}
\end{equation}
The first one of these follows by using \HopfRule:
\ctikzfig{double-had-edge}
The second one follows by applying \IdentityRule from right to left and then using \SpiderRule. For the last one we do:
\begin{equation*}
\begin{tikzpicture}
	\begin{pgfonlayer}{nodelayer}
		\node [style=Z phase dot] (10) at (-7.5, 0) {$\alpha$};
		\node [style=none] (11) at (-8.5, -0.75) {};
		\node [style=none] (12) at (-6.75, -0.75) {};
		\node [style=none] (13) at (-7.5, -0.75) {...};
		\node [style=hadamard] (14) at (-7.5, 1) {};
		\node [style=none] (23) at (-4.5, -0.75) {};
		\node [style=Z phase dot] (24) at (-4.5, 1) {$\frac\pi2$};
		\node [style=X phase dot] (25) at (-3.5, 2) {$\frac\pi2$};
		\node [style=Z phase dot] (26) at (-3.5, 0) {$\alpha$};
		\node [style=none] (27) at (-2.75, -0.75) {};
		\node [style=none] (28) at (-3.5, -0.75) {...};
		\node [style=Z phase dot] (29) at (-2.5, 1) {$\frac\pi2$};
		\node [style=none] (30) at (-5.5, 0) {=};
		\node [style=X phase dot] (31) at (0.75, 1.25) {$\frac\pi2$};
		\node [style=none] (32) at (0.75, -0.75) {...};
		\node [style=none] (34) at (1.5, -0.75) {};
		\node [style=none] (35) at (-0.25, -0.75) {};
		\node [style=Z phase dot] (36) at (0.75, 0) {$\ \alpha+\pi\ $};
		\node [style=none] (37) at (-1.25, 0) {=};
		\node [style=none] (38) at (5.75, -0.75) {};
		\node [style=none] (40) at (5, -0.75) {...};
		\node [style=Z phase dot] (41) at (5, 0) {$\ \alpha+\pi\ $};
		\node [style=none] (42) at (3, 0) {=};
		\node [style=none] (43) at (4, -0.75) {};
		\node [style=none] (44) at (-1.25, 1) {\spiderrule};
	\end{pgfonlayer}
	\begin{pgfonlayer}{edgelayer}
		\draw (11.center) to (10);
		\draw (10) to (12.center);
		\draw [bend right=90, looseness=1.75] (10) to (14);
		\draw [bend left=90, looseness=1.75] (10) to (14);
		\draw (23.center) to (26);
		\draw (26) to (27.center);
		\draw [in=-90, out=-180] (26) to (24);
		\draw [in=180, out=90] (24) to (25);
		\draw [in=90, out=0] (25) to (29);
		\draw [in=0, out=-90] (29) to (26);
		\draw (35.center) to (36);
		\draw (36) to (34.center);
		\draw [bend right=90, looseness=1.50] (36) to (31);
		\draw [bend left=90, looseness=1.50] (36) to (31);
		\draw (43.center) to (41);
		\draw (41) to (38.center);
	\end{pgfonlayer}
\end{tikzpicture}
\end{equation*}
Here, the first step is simply the definition of the Hadamard box, and in the last step we use the antipode rule \HopfRule and implicitly drop the scalar X-spider that we are left with.

At this point, the first 3 conditions are satisfied. To satisfy condition 4, we must deal with two special cases: (a) inputs/outputs not connected to any Z-spider, and (b) multiple inputs/outputs connected to the same Z-spider. For case (a), there are only two possibilities left: either an input and an output are directly connected (i.e. a `bare wire'), or they are connected to a Hadamard gate. These situations can both be removed by right-to-left applications of \IdentityRule and \HHRule as follows:
\ctikzfig{ident-graph-form-2}
For case (b), we can again use \IdentityRule and \HHRule to introduce `dummy' spiders until each input/output is connected to a single spider:
\ctikzfig{ident-graph-form}
Once this is done, the resulting ZX-diagram satisfies conditions 1-4.
\end{proof}





A useful feature of a graph-like \zxdiagram is that much of its structure is captured by its underlying \textit{open graph}.

\begin{definition}\label{def:open-graph}
  An \emph{open graph} is a triple $(G,I,O)$ where $G = (V,E)$ is an undirected graph, and $I \subseteq V$ is a set of \emph{inputs} and $O \subseteq V$ a set of \emph{outputs}. For a \zxdiagram $D$, the \emph{underlying open graph} $\UG D$ is the open graph whose vertices are the spiders of $D$, whose edges correspond to Hadamard edges, and whose sets $I$ and $O$ are the subsets of spiders connected to the inputs and outputs of $D$, respectively.
\end{definition}

See Figure~\ref{fig:underlying-graph} for an example of converting a circuit into a graph-like diagram with Lemma~\ref{lem:all-zx-are-graph-like} and how the associated underlying graph is found.


\begin{figure}[!htb]
	\centering
	\input{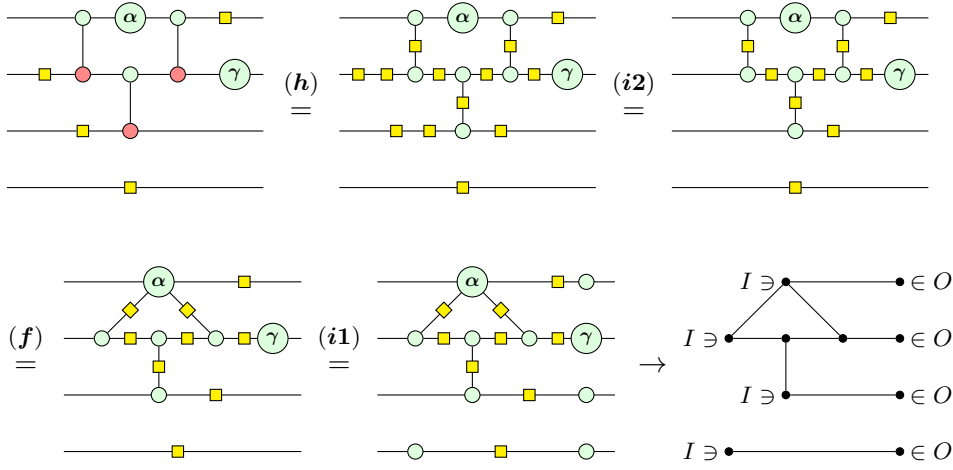}
	\caption{A circuit, which is transformed into an equivalent graph-like \zxdiagram, and its underlying open graph.}
    \label{fig:underlying-graph}
\end{figure}




A graph-like \zxdiagram can be seen as an open graph with an assignment of angles to each of its vertices. Note that, because of Definition~\ref{def:graph-form}, the sets $I$ and $O$ in $\UG D$ will always be disjoint. When discussing open graphs we will use the notation $\comp I:=V\setminus I$ and $\comp O:=V\setminus O$ for the non-inputs and non-outputs respectively. The set of neighbours of a vertex $v\in V$ will be denoted by $N(v)$.

We are now ready to define a graph-theoretic condition that will be instrumental in our circuit extraction procedure.


\begin{definition}{\cite{mhalla2011graph}}\label{def:gflow}
Given an open graph $G$, a \emph{focused gFlow} $(g,\prec)$ on $G$
consists of a function $g:\comp O \to 2^{\comp I}$ and a partial order $\prec$
on the vertices $V$ of $G$ such that for all $u\in \comp O$, 
\begin{enumerate} 
\item $\odd G {g(u)}\cap \comp O = \{u\}$
\item  $\forall v\in g(u), u\prec v$
\end{enumerate}
where $2^{\comp I}$ is the powerset of $\comp I$ and $\odd G A:=\{v\in V(G)~|~|N(v)\cap A| \equiv 1\bmod 2\}$ is the \emph{odd neighbourhood} of $A$.
\end{definition}

\begin{example}
While not strictly necessary for the following, we briefly give some intuition for the
conditions in Definition~\ref{def:gflow}. They can be understood in a more operational way with the following game. Suppose we consider an open graph $G$ whose vertices are labelled with $0$'s and $1$'s, where a $1$ indicates the presence of an error. Define an operation $\textbf{flip}_v$, which flips all of the bits on the \textit{neighbours} of a given vertex $v$. Our goal is to propagate all of the errors present in $G$ to the outputs using only applications of the operation $\textbf{flip}_v$. For example:
\begin{equation}\label{eq:has-gflow}
\begin{tikzpicture}
	\begin{pgfonlayer}{nodelayer}
		\node [style=vertex] (220) at (-1.5, -1) {};
		\node [style=vertex] (222) at (-1.5, 1) {};
		\node [style=vertex] (223) at (0, -1) {};
		\node [style=vertex] (225) at (1.5, -1) {};
		\node [style=none, anchor=east, font={\footnotesize}] (230) at (-1.75, 1) {$I \ni$};
		\node [style=none, anchor=east, font={\footnotesize}] (231) at (-1.75, -1) {$I \ni$};
		\node [style=none, anchor=west, font={\footnotesize}] (233) at (1.75, -1) {$\in O$};
		\node [style=none] (234) at (-1.5, 1.5) {\footnotesize\color{gray} $0$};
		\node [style=none] (235) at (0, -0.5) {\footnotesize\color{gray} $0$};
		\node [style=none] (236) at (1.5, -0.5) {\footnotesize\color{gray} $0$};
		\node [style=none] (237) at (-1.5, -0.5) {\footnotesize\color{gray} $1$};
		\node [style=none] (238) at (0, -1.5) {\color{blue} $v$};
		\node [style=vertex] (239) at (1.5, 1) {};
		\node [style=none, anchor=west, font={\footnotesize}] (240) at (1.75, 1) {$\in O$};
		\node [style=none] (241) at (1.5, 1.5) {\footnotesize\color{gray} $1$};
		\node [style=none] (242) at (1.5, 0.5) {\color{blue} $w$};
	\end{pgfonlayer}
	\begin{pgfonlayer}{edgelayer}
		\draw (223) to (225);
		\draw (220) to (223);
		\draw (222) to (223);
		\draw (222) to (239);
	\end{pgfonlayer}
\end{tikzpicture}
\ \  \xrightarrow{\textbf{flip}_v} \ \ 
\begin{tikzpicture}
	\begin{pgfonlayer}{nodelayer}
		\node [style=vertex] (220) at (-1.5, -1) {};
		\node [style=vertex] (222) at (-1.5, 1) {};
		\node [style=vertex] (223) at (0, -1) {};
		\node [style=vertex] (225) at (1.5, -1) {};
		\node [style=none, anchor=east, font={\footnotesize}] (230) at (-1.75, 1) {$I \ni$};
		\node [style=none, anchor=east, font={\footnotesize}] (231) at (-1.75, -1) {$I \ni$};
		\node [style=none, anchor=west, font={\footnotesize}] (233) at (1.75, -1) {$\in O$};
		\node [style=none] (234) at (-1.5, 1.5) {\footnotesize\color{gray} $1$};
		\node [style=none] (235) at (0, -0.5) {\footnotesize\color{gray} $0$};
		\node [style=none] (236) at (1.5, -0.5) {\footnotesize\color{gray} $1$};
		\node [style=none] (237) at (-1.5, -0.5) {\footnotesize\color{gray} $0$};
		\node [style=none] (238) at (0, -1.5) {\color{blue} $v$};
		\node [style=vertex] (239) at (1.5, 1) {};
		\node [style=none, anchor=west, font={\footnotesize}] (240) at (1.75, 1) {$\in O$};
		\node [style=none] (241) at (1.5, 1.5) {\footnotesize\color{gray} $1$};
		\node [style=none] (242) at (1.5, 0.5) {\color{blue} $w$};
	\end{pgfonlayer}
	\begin{pgfonlayer}{edgelayer}
		\draw (223) to (225);
		\draw (220) to (223);
		\draw (222) to (223);
		\draw (222) to (239);
	\end{pgfonlayer}
\end{tikzpicture}
\ \  \xrightarrow{\textbf{flip}_w} \ \ 
\begin{tikzpicture}
	\begin{pgfonlayer}{nodelayer}
		\node [style=vertex] (220) at (-1.5, -1) {};
		\node [style=vertex] (222) at (-1.5, 1) {};
		\node [style=vertex] (223) at (0, -1) {};
		\node [style=vertex] (225) at (1.5, -1) {};
		\node [style=none, anchor=east, font={\footnotesize}] (230) at (-1.75, 1) {$I \ni$};
		\node [style=none, anchor=east, font={\footnotesize}] (231) at (-1.75, -1) {$I \ni$};
		\node [style=none, anchor=west, font={\footnotesize}] (233) at (1.75, -1) {$\in O$};
		\node [style=none] (234) at (-1.5, 1.5) {\footnotesize\color{gray} $0$};
		\node [style=none] (235) at (0, -0.5) {\footnotesize\color{gray} $0$};
		\node [style=none] (236) at (1.5, -0.5) {\footnotesize\color{gray} $1$};
		\node [style=none] (237) at (-1.5, -0.5) {\footnotesize\color{gray} $0$};
		\node [style=none] (238) at (0, -1.5) {\color{blue} $v$};
		\node [style=vertex] (239) at (1.5, 1) {};
		\node [style=none, anchor=west, font={\footnotesize}] (240) at (1.75, 1) {$\in O$};
		\node [style=none] (241) at (1.5, 1.5) {\footnotesize\color{gray} $1$};
		\node [style=none] (242) at (1.5, 0.5) {\color{blue} $w$};
	\end{pgfonlayer}
	\begin{pgfonlayer}{edgelayer}
		\draw (223) to (225);
		\draw (220) to (223);
		\draw (222) to (223);
		\draw (222) to (239);
	\end{pgfonlayer}
\end{tikzpicture}
\end{equation}
For some open graphs and configurations of errors, this task might be impossible. For example, there is no solution for the following graph:
\ctikzfig{no-gflow}
However, we can always succeed if we are given the following data: an ordering $\prec$ of vertices which give a direction of `time' going from inputs to outputs, and, for each vertex, a \textit{correction set} $g(v)$ of vertices in the future of $v$ (w.r.t. $\prec$) such that applying $\textbf{flip}_w$ for all $w \in g(v)$ flips the bit on $v$ without affecting any other bits, except for those labelling outputs and $g(v)$ itself. By repeatedly finding the minimal vertex $v$ (w.r.t. $\prec$) with an error and applying $\textbf{flip}_w$ to all $w \in g(v)$, the procedure will eventually propagate all of the $1$'s to the outputs of $G$.
\end{example}

\begin{remark}\label{rem:focused-gflow-v-gflow}
  Generalised flow techniques were introduced in the context of measurement-based quantum computing \cite{Patterns,DKPP09}. \emph{Focused gFlow} is a special case of the standard notion of gFlow~\cite{GFlow}, where the latter allows other vertices $v$ in the set $\odd G {g(u)}\cap \comp O$, provided they are in the future of $u$ (i.e. $u \prec v$). However, as a property of graphs, the two notions are equivalent: an open graph has a gFlow if and only if it has a focused gFlow. We will rely on this equivalence in Appendix~\ref{sec:circuits-causal-flow} to prove the following Lemma.
\end{remark}

\begin{lemma}\label{lem:circuits-have-gflow}
  If $D$ is a graph-like \zxdiagram obtained from a circuit by the procedure of
  Lemma~\ref{lem:all-zx-are-graph-like}, then $\UG D$ admits a focused gFlow.
  \begin{proof}
    See Appendix~\ref{sec:circuits-causal-flow}.
  \end{proof}
\end{lemma}

\section{Local complementation and pivoting}\label{sec:lcomppivot}


Local complementation is a graph transformation introduced by
Kotzig \cite{kotzig}.
\begin{definition}
  Let $G$ be a 
  graph and let $u$ be a vertex of $G$. 
  The \emph{local complementation} of $G$ according to $u$, 
  written as $G\star u$, is a graph which has the same vertices as $G$, 
  but all the neighbours $v,w$ of $u$ are connected in $G\star u$ if and only if
  they are not connected in $G$. 
  All other edges are unchanged.
\end{definition}

\begin{example} Local complementations according to the vertices $a$ and $b$ are given by:
\begin{equation*}
G\quad\begin{tikzpicture}
	\begin{pgfonlayer}{nodelayer}
		\node [style=vertex] (0) at (-1, 1) {};
		\node [style=vertex] (1) at (-1, -1) {};
		\node [style=vertex] (2) at (1, -1) {};
		\node [style=vertex] (3) at (1, 1) {};
		\node [style=none] (4) at (-1.5, 1) {$a$};
		\node [style=none] (5) at (1.5, 1) {$b$};
		\node [style=none] (6) at (1.5, -1) {$d$};
		\node [style=none] (7) at (-1.5, -1) {$c$};
	\end{pgfonlayer}
	\begin{pgfonlayer}{edgelayer}
		\draw (0) to (3);
		\draw (0) to (2);
		\draw (3) to (2);
		\draw (0) to (1);
	\end{pgfonlayer}
\end{tikzpicture}\qquad\qquad G\star a\quad\begin{tikzpicture}
	\begin{pgfonlayer}{nodelayer}
		\node [style=vertex] (0) at (-1, 1) {};
		\node [style=vertex] (1) at (-1, -1) {};
		\node [style=vertex] (2) at (1, -1) {};
		\node [style=vertex] (3) at (1, 1) {};
		\node [style=none] (4) at (-1.5, 1) {$a$};
		\node [style=none] (5) at (1.5, 1) {$b$};
		\node [style=none] (6) at (1.5, -1) {$d$};
		\node [style=none] (7) at (-1.5, -1) {$c$};
	\end{pgfonlayer}
	\begin{pgfonlayer}{edgelayer}
		\draw (0) to (3);
		\draw (0) to (2);
		\draw (0) to (1);
		\draw (3) to (1);
		\draw (1) to (2);
	\end{pgfonlayer}
\end{tikzpicture}\qquad\qquad (G\star a) \star b\quad\begin{tikzpicture}
	\begin{pgfonlayer}{nodelayer}
		\node [style=vertex] (0) at (-1, 1) {};
		\node [style=vertex] (1) at (-1, -1) {};
		\node [style=vertex] (2) at (1, -1) {};
		\node [style=vertex] (3) at (1, 1) {};
		\node [style=none] (4) at (-1.5, 1) {$a$};
		\node [style=none] (5) at (1.5, 1) {$b$};
		\node [style=none] (6) at (1.5, -1) {$d$};
		\node [style=none] (7) at (-1.5, -1) {$c$};
	\end{pgfonlayer}
	\begin{pgfonlayer}{edgelayer}
		\draw (0) to (3);
		\draw (0) to (2);
		\draw (3) to (1);
		\draw (1) to (2);
	\end{pgfonlayer}
\end{tikzpicture}
\end{equation*}
\end{example}

A related graph-theoretic operation is pivoting, which takes place at a pair of adjacent vertices.
\begin{definition}
  Let $G$ be a 
  graph, and let $u$ and $v$ be a pair of connected vertices in $G$. 
  The \emph{pivot} of $G$ along the edge $uv$ is the graph $G\wedge uv := G\star u \star v \star u$.
\end{definition}
On a graph, pivoting consists in exchanging $u$ and $v$, and complementing the edges
between three particular subsets of the vertices: the common neighbourhood of $u$ and $v$ (i.e.~$N_G(u)\cap N_G(v)$), the exclusive neighbourhood of $u$ (i.e.~$N_G(u)\setminus (N_G(v)\cup \{v\})$), and exclusive neighbourhood of $v$ (i.e.~$N_G(v)\setminus (N_G(u)\cup \{u\})$):
\[G \quad\begin{tikzpicture}
	\begin{pgfonlayer}{nodelayer}
		\node [style=vertex] (0) at (-2, 1.5) {};
		\node [style=vertex] (1) at (2, 1.5) {};
		\node [style=vertex set] (2) at (0, 0.5) {$~~A~~$};
		\node [style=vertex set] (3) at (-2.5, -0.75) {$~~B~~$};
		\node [style=vertex set] (4) at (2.5, -0.75) {$~~C~~$};
		\node [style=none] (5) at (2.75, 1.5) {$v$};
		\node [style=none] (6) at (-2.75, 1.5) {$u$};
		\node [style=none] (7) at (-0.25, 0.5) {};
		\node [style=none] (8) at (0, 0.25) {};
		\node [style=none] (9) at (0.25, 0.25) {};
		\node [style=none] (10) at (0.5, 0.5) {};
		\node [style=none] (11) at (-2.5, -0.5) {};
		\node [style=none] (12) at (-2.25, -0.75) {};
		\node [style=none] (13) at (-2.25, -0.75) {};
		\node [style=none] (14) at (-2.25, -1) {};
		\node [style=none] (15) at (2.25, -0.75) {};
		\node [style=none] (16) at (2.25, -1) {};
		\node [style=none] (17) at (2.25, -0.75) {};
		\node [style=none] (18) at (2.5, -0.5) {};
	\end{pgfonlayer}
	\begin{pgfonlayer}{edgelayer}
		\draw (0) to (1);
		\draw (0) to (3);
		\draw (2) to (1);
		\draw (0) to (2);
		\draw (1) to (4);
		\draw (7.center) to (11.center);
		\draw (13.center) to (8.center);
		\draw (9.center) to (15.center);
		\draw (10.center) to (18.center);
		\draw (17.center) to (12.center);
		\draw (14.center) to (16.center);
	\end{pgfonlayer}
\end{tikzpicture}\qquad\qquad \quad G\wedge uv \quad\begin{tikzpicture}
	\begin{pgfonlayer}{nodelayer}
		\node [style=vertex] (0) at (-2, 1.5) {};
		\node [style=vertex] (1) at (2, 1.5) {};
		\node [style=vertex set] (2) at (0, 0.5) {$~~A~~$};
		\node [style=vertex set] (3) at (-2.5, -1) {$~~B~~$};
		\node [style=vertex set] (4) at (2.25, -1) {$~~C~~$};
		\node [style=none] (5) at (-3, 1.5) {$v$};
		\node [style=none] (6) at (3, 1.5) {$u$};
		\node [style=none] (7) at (-2.75, -0.75) {};
		\node [style=none] (8) at (-2.25, -1.25) {};
		\node [style=none] (9) at (-0.25, 0.75) {};
		\node [style=none] (10) at (0.25, 0.25) {};
		\node [style=none] (11) at (-0.25, 0.25) {};
		\node [style=none] (12) at (0.25, 0.75) {};
		\node [style=none] (13) at (2.5, -0.75) {};
		\node [style=none] (14) at (2, -1.25) {};
		\node [style=none] (15) at (2.25, -1.25) {};
		\node [style=none] (16) at (2.25, -0.75) {};
		\node [style=none] (17) at (-2.5, -1.25) {};
		\node [style=none] (18) at (-2.5, -0.75) {};
	\end{pgfonlayer}
	\begin{pgfonlayer}{edgelayer}
		\draw (0) to (1);
		\draw (0) to (3);
		\draw (2) to (1);
		\draw (0) to (2);
		\draw (1) to (4);
		\draw (7.center) to (10.center);
		\draw (9.center) to (8.center);
		\draw (18.center) to (15.center);
		\draw (17.center) to (16.center);
		\draw (12.center) to (14.center);
		\draw (11.center) to (13.center);
	\end{pgfonlayer}
\end{tikzpicture}
\]

For a more concrete illustration of pivoting see Example~\ref{ex:pivot} or Equation~\eqref{eq:gs-pivot} below.
\begin{example}\label{ex:pivot}
In the graph $G$ below, $\{a, b\}$ is in the neighbourhood of $u$ alone, $\{d\}$ is in the neighbourhood of $v$ alone, and $\{c\}$ is in the the neighbourhood of both. To perform the pivot along $uv$, we complement the edges connecting $\{a, b\}$ to $\{d\}$, $\{d\}$ to $\{c\}$ and $\{a, b\}$ to $\{c\}$. We then swap $u$ and $v$:
\[
G\quad\begin{tikzpicture}[scale=1.5]
	\begin{pgfonlayer}{nodelayer}
		\node [style=vertex] (0) at (-0.5, 1.25) {};
		\node [style=vertex] (1) at (1.5, 1.25) {};
		\node [style=vertex] (2) at (0.5, 0) {};
		\node [style=vertex] (3) at (-1.5, -0.75) {};
		\node [style=vertex] (4) at (-0.5, 0) {};
		\node [style=vertex] (5) at (1.5, -0.75) {};
		\node [style=vertex] (6) at (0, -1.25) {};
		\node [style=none] (7) at (-1, 1.25) {$u$};
		\node [style=none] (8) at (2, 1.25) {$v$};
		\node [style=none] (9) at (2, -0.75) {$d$};
		\node [style=none] (10) at (-0.75, 0) {$b$};
		\node [style=none] (11) at (0.5, 0.5) {$c$};
		\node [style=none] (12) at (-2, -0.75) {$a$};
		\node [style=none] (13) at (0, -1) {$e$};
	\end{pgfonlayer}
	\begin{pgfonlayer}{edgelayer}
		\draw (0) to (1);
		\draw (0) to (2);
		\draw (2) to (1);
		\draw (0) to (3);
		\draw (0) to (4);
		\draw (1) to (5);
		\draw (5) to (6);
		\draw (6) to (3);
		\draw (4) to (2);
		\draw (2) to (5);
	\end{pgfonlayer}
\end{tikzpicture} \qquad \quad G\wedge uv\quad\begin{tikzpicture}[scale=1.5]
	\begin{pgfonlayer}{nodelayer}
		\node [style=vertex] (0) at (-0.5, 1.25) {};
		\node [style=vertex] (1) at (1.5, 1.25) {};
		\node [style=vertex] (2) at (0.5, 0) {};
		\node [style=vertex] (3) at (-1.5, -0.75) {};
		\node [style=vertex] (4) at (-0.5, 0) {};
		\node [style=vertex] (5) at (1.5, -0.75) {};
		\node [style=vertex] (6) at (0, -1.25) {};
		\node [style=none] (7) at (-1.25, 1.25) {$v$};
		\node [style=none] (8) at (2, 1.25) {$u$};
		\node [style=none] (9) at (2, -0.75) {$d$};
		\node [style=none] (10) at (-0.75, 0.25) {$b$};
		\node [style=none] (11) at (0.5, 0.25) {$c$};
		\node [style=none] (12) at (-2, -0.75) {$a$};
		\node [style=none] (13) at (0, -1) {$e$};
	\end{pgfonlayer}
	\begin{pgfonlayer}{edgelayer}
		\draw (0) to (1);
		\draw (0) to (2);
		\draw (2) to (1);
		\draw (0) to (3);
		\draw (0) to (4);
		\draw (1) to (5);
		\draw (5) to (6);
		\draw (6) to (3);
		\draw (3) to (2);
		\draw (4) to (5);
		\draw [in=180, out=0] (3) to (5);
	\end{pgfonlayer}
\end{tikzpicture}
\]
\end{example}

Our primary interest in local complementation and pivoting is that each corresponds to a transformation of \zxdiagrams. In the special case where a \zxdiagram $D$ is a graph-state, it is possible to obtain a diagram $D'$ where $\UG{D'} = \UG{D} \star u$ by applying an $X_{-\pi/2}$ gate on the spider corresponding to $u$ and $Z_{\pi/2}$ on all of the spiders in $N(v)$:
\begin{equation}\label{eq:gs-local-comp}
\begin{tikzpicture}
	\begin{pgfonlayer}{nodelayer}
		\node [style=Z dot] (0) at (-5.75, 0.5) {};
		\node [style=X phase dot] (1) at (-5.75, 1.5) {-$\frac\pi2$};
		\node [style=Z phase dot] (2) at (-1, 1.5) {$\frac\pi2$};
		\node [style=Z phase dot] (3) at (-2.5, 1.5) {$\frac\pi2$};
		\node [style=Z phase dot] (4) at (-3.5, 1.5) {$\frac\pi2$};
		\node [style=none] (5) at (-1, 2.25) {};
		\node [style=none] (6) at (-2.5, 2.25) {};
		\node [style=none] (7) at (-3.5, 2.25) {};
		\node [style=none] (8) at (-5.75, 2.25) {};
		\node [style=Z dot] (12) at (-2.5, 0.5) {};
		\node [style=Z dot] (13) at (-1, -0.5) {};
		\node [style=Z dot] (14) at (-3.5, -0.5) {};
		\node [style=Z dot] (15) at (-1, -0.5) {};
		\node [style=none] (16) at (-0.5, -1) {};
		\node [style=none] (17) at (-4, -1) {};
		\node [style=none] (18) at (-2.25, -1.5) {$N(u)$};
		\node [style=none] (19) at (-5.75, -0.25) {$u$};
		\node [style=none] (20) at (0.5, 0.5) {$=$};
		\node [style=Z dot] (21) at (2.25, 0.5) {};
		\node [style=Z dot] (22) at (5.5, 0.5) {};
		\node [style=Z dot] (23) at (7, -0.5) {};
		\node [style=Z dot] (24) at (4.5, -0.5) {};
		\node [style=Z dot] (25) at (7, -0.5) {};
		\node [style=none] (30) at (-1.75, 0.75) {...};
		\node [style=none] (31) at (7, 2.25) {};
		\node [style=none] (32) at (5.5, 2.25) {};
		\node [style=none] (33) at (4.5, 2.25) {};
		\node [style=none] (34) at (2.25, 2.25) {};
		\node [style=none] (35) at (6.25, 0.75) {...};
	\end{pgfonlayer}
	\begin{pgfonlayer}{edgelayer}
		\draw (1) to (0);
		\draw (8.center) to (1);
		\draw (2) to (5.center);
		\draw (6.center) to (3);
		\draw (7.center) to (4);
		\draw [style=hadamard edge] (0) to (15);
		\draw (15) to (2);
		\draw (12) to (3);
		\draw (14) to (4);
		\draw [style=hadamard edge] (0) to (14);
		\draw [style=hadamard edge] (0) to (12);
		\draw [style=hadamard edge] (12) to (15);
		\draw [style=brace edge] (16.center) to (17.center);
		\draw [style=hadamard edge] (21) to (25);
		\draw [style=hadamard edge] (21) to (24);
		\draw [style=hadamard edge] (21) to (22);
		\draw (21) to (34.center);
		\draw (24) to (33.center);
		\draw (22) to (32.center);
		\draw (25) to (31.center);
		\draw [style=hadamard edge] (24) to (22);
		\draw [style=hadamard edge] (24) to (25);
	\end{pgfonlayer}
\end{tikzpicture}
\end{equation}
This is a well-known property of graph states, and a derivation using the rules of the \zxcalculus was given in Ref.~\cite{DP1}. Similarly, it was shown in Ref.~\cite{DP3} that a pivot $\UG{D'} = \UG{D} \wedge uv$ can be introduced by applying Hadamard gates on $u$ and $v$ and $Z_{\pi}$ gates on $N(u) \cap N(v)$:
\begin{equation}\label{eq:gs-pivot}
\begin{tikzpicture}
	\begin{pgfonlayer}{nodelayer}
		\node [style=Z dot] (0) at (-8.75, 1.25) {};
		\node [style=Z dot] (1) at (-3, 1.25) {};
		\node [style=Z dot] (2) at (-6.5, 0.5) {};
		\node [style=Z dot] (3) at (-5.5, -0.75) {};
		\node [style=Z dot] (4) at (-10.25, -1.5) {};
		\node [style=Z dot] (5) at (-9.25, -2.5) {};
		\node [style=Z dot] (6) at (-2.5, -2.5) {};
		\node [style=Z dot] (7) at (-1.5, -1.5) {};
		\node [style=none] (8) at (-8.75, 2.75) {};
		\node [style=none] (9) at (-3, 2.75) {};
		\node [style=none] (10) at (-5.5, 2.75) {};
		\node [style=none] (11) at (-6.5, 2.75) {};
		\node [style=none] (12) at (-10.5, -1) {};
		\node [style=none] (13) at (-9.5, -2) {};
		\node [style=none] (14) at (-2.25, -2) {};
		\node [style=none] (15) at (-1.25, -1) {};
		\node [style=none] (17) at (-1.75, -2.25) {...};
		\node [style=none] (18) at (-10, -2.25) {...};
		\node [style=hadamard] (19) at (-8.75, 2) {};
		\node [style=hadamard] (20) at (-3, 2) {};
		\node [style=Z dot] (21) at (3.25, 1.75) {};
		\node [style=Z dot] (22) at (8.25, 1.75) {};
		\node [style=Z dot] (23) at (5.25, 0.75) {};
		\node [style=Z dot] (24) at (6.25, -0.75) {};
		\node [style=Z dot] (25) at (1.75, -1.25) {};
		\node [style=Z dot] (26) at (2.75, -2.25) {};
		\node [style=Z dot] (27) at (8.75, -2.25) {};
		\node [style=Z dot] (28) at (9.75, -1.25) {};
		\node [style=none] (29) at (3.25, 3) {};
		\node [style=none] (30) at (8.25, 3) {};
		\node [style=none] (31) at (6.25, 3) {};
		\node [style=none] (32) at (5.25, 3) {};
		\node [style=none] (33) at (1.5, -0.75) {};
		\node [style=none] (34) at (2.5, -1.75) {};
		\node [style=none] (35) at (9, -1.75) {};
		\node [style=none] (36) at (10, -0.75) {};
		\node [style=Z phase dot] (40) at (-5.5, 2) {$\pi$};
		\node [style=Z phase dot] (41) at (-6.5, 2) {$\pi$};
		\node [style=none] (42) at (-6, 2.5) {...};
		\node [style=none] (43) at (-6, 0) {...};
		\node [style=none] (44) at (5.75, 0) {...};
		\node [style=none] (45) at (9.75, -2) {...};
		\node [style=none] (46) at (2, -2) {...};
		\node [style=none] (47) at (5.75, 3) {...};
		\node [style=none] (49) at (-9.5, 1.5) {$u$};
		\node [style=none] (50) at (-2.25, 1.5) {$v$};
		\node [style=none] (51) at (-8.5, -3) {};
		\node [style=none] (52) at (-11, -3) {};
		\node [style=none] (53) at (-9.75, -3.5) {$N(u) \backslash (N(v) \cup \{v\})$};
		\node [style=none] (54) at (-0.75, -3) {};
		\node [style=none] (55) at (-3.25, -3) {};
		\node [style=none] (56) at (-2, -3.5) {$N(v) \backslash (N(u) \cup \{u\})$};
		\node [style=none] (57) at (-6.75, 3) {};
		\node [style=none] (58) at (-5.25, 3) {};
		\node [style=none] (59) at (-6, 3.5) {$N(u) \cap N(v)$};
		\node [style=none] (60) at (0, 0) {$=$};
	\end{pgfonlayer}
	\begin{pgfonlayer}{edgelayer}
		\draw (8.center) to (0);
		\draw (9.center) to (1);
		\draw (2) to (11.center);
		\draw (3) to (10.center);
		\draw (6) to (14.center);
		\draw (7) to (15.center);
		\draw (5) to (13.center);
		\draw (4) to (12.center);
		\draw [style=hadamard edge] (0) to (2);
		\draw [style=hadamard edge] (0) to (3);
		\draw [style=hadamard edge] (3) to (1);
		\draw [style=hadamard edge] (2) to (1);
		\draw [style=hadamard edge] (0) to (4);
		\draw [style=hadamard edge] (0) to (5);
		\draw [style=hadamard edge] (1) to (6);
		\draw [style=hadamard edge] (1) to (7);
		\draw [style=hadamard edge] (0) to (1);
		\draw (23) to (32.center);
		\draw (24) to (31.center);
		\draw (27) to (35.center);
		\draw (28) to (36.center);
		\draw (26) to (34.center);
		\draw (25) to (33.center);
		\draw [style=hadamard edge] (21) to (23);
		\draw [style=hadamard edge] (21) to (24);
		\draw [style=hadamard edge] (24) to (22);
		\draw [style=hadamard edge] (23) to (22);
		\draw [style=hadamard edge] (21) to (25);
		\draw [style=hadamard edge] (21) to (26);
		\draw [style=hadamard edge] (22) to (27);
		\draw [style=hadamard edge] (22) to (28);
		\draw [style=hadamard edge] (25) to (24);
		\draw [style=hadamard edge] (26) to (23);
		\draw [style=hadamard edge] (24) to (28);
		\draw [style=hadamard edge] (25) to (27);
		\draw [style=hadamard edge] (26) to (28);
		\draw [in=90, out=-90, looseness=0.50] (29.center) to (22);
		\draw [in=90, out=-90, looseness=0.50] (30.center) to (21);
		\draw [style=hadamard edge] (21) to (22);
		\draw [style=brace edge] (51.center) to (52.center);
		\draw [style=brace edge] (54.center) to (55.center);
		\draw [style=brace edge] (57.center) to (58.center);
		\draw [style=hadamard edge] (5) to (6);
		\draw [style=hadamard edge] (4) to (7);
		\draw [style=hadamard edge] (3) to (5);
		\draw [style=hadamard edge] (2) to (4);
		\draw [style=hadamard edge] (3) to (6);
		\draw [style=hadamard edge] (2) to (7);
		\draw [style=hadamard edge] (27) to (23);
	\end{pgfonlayer}
\end{tikzpicture}
\end{equation}
Note that the swap on the RHS comes from the vertices $u$ and $v$ being interchanged by the pivot.

The following theorem will be crucial for our extraction routine. It shows that the existence of a focused gFlow is preserved by local complementation (resp.~pivoting) followed by the deletion of the vertex (resp.~the pair of vertices) on which the transformation is applied: 


\begin{theorem}\label{thm:gflow-preserve}
Let $(G,I,O)$ be an open graph that admits a focused gFlow, then $(G', I, O)$ also admits a gFlow in the following two cases: 
\begin{enumerate}
\item for $u \notin I \cup O$, setting $G' := (G \star u) \backslash \{ u \}$
\item for adjacent $u,v \notin I \cup O$, setting $G' := (G \wedge uv) \backslash \{ u, v \}$
\end{enumerate}
where $G\setminus W$ is the graph obtained by deleting the vertices in $W$ and any incident edges.
\end{theorem}

\begin{proof}
The two cases are proved in Appendix~\ref{sec:pres-focus-gflow}, Lemmas \ref{lem:gflow-lcomp} and \ref{lem:gflow-pivot} respectively.
\end{proof}


\section{A simplification strategy for circuits}\label{sec:simp}
We now have all the necessary machinery to introduce our simplification routine. The general idea is to use local complementation and pivoting based rewrite rules to remove as many \emph{interior} spiders as possible. A spider is called interior when it is not connected to an input or an output, otherwise it is called a \textit{boundary} spider.

\begin{definition}
  We call a spider \emph{Pauli} if its phase is a multiple of $\pi$, and \emph{Clifford} if its phase is a multiple of $\frac\pi2$. If the phase of a Clifford spider is an odd multiple of $\frac\pi2$ (and hence non-Pauli), we call this a \emph{proper Clifford} spider.
\end{definition}

\noindent The graph-like \zxdiagram resulting from the translation of a Clifford circuit contains only Clifford spiders, since the only time the phase changes on a spider is during a spider-fusion step, in which case the phases are added together.

We will show that our procedure is able to eliminate \emph{all} interior proper Clifford spiders and all Pauli spiders adjacent either to a boundary spider or any interior Pauli spider. In particular, for Clifford circuits, we will obtain a pseudo-normal form which contains no interior spiders (cf. Section~\ref{sec:cliffordT}).

The main rewrite rules we use are based on local complementation and pivoting. The first one allows us to remove any interior proper Clifford spider while the second one removes any two connected interior Pauli spiders.
\begin{lemma}\label{lem:lc-simp}
The following rule holds in the \zxcalculus:
\begin{equation}\label{eq:lc-simp}
  \hfill \begin{tikzpicture}
	\begin{pgfonlayer}{nodelayer}
		\node [style=Z phase dot] (0) at (-5, 1) {$\ \pm\frac\pi2\ $};
		\node [style=Z phase dot] (2) at (-7.5, 0.5) {$\alpha_1$};
		\node [style=Z phase dot] (4) at (-2.25, 0.5) {$\alpha_n$};
		\node [style=none] (5) at (-8.5, -0.5) {};
		\node [style=none] (6) at (-7.5, -0.5) {};
		\node [style=none] (7) at (-2.25, -0.5) {};
		\node [style=none] (8) at (-1.25, -0.5) {};
		\node [style=none] (9) at (-5, -0.25) {...};
		\node [style=none] (10) at (-8, -0.5) {...};
		\node [style=none] (11) at (-1.75, -0.5) {...};
		\node [style=none] (16) at (0.25, -0.75) {$=$};
		\node [style=none] (17) at (10, -0.5) {};
		\node [style=none] (18) at (1.75, -0.5) {};
		\node [style=none] (19) at (9.5, -0.5) {...};
		\node [style=none] (20) at (9, -0.5) {};
		\node [style=Z phase dot] (24) at (2.75, 0.5) {$\ \alpha_1\!\mp\!\frac\pi2\ $};
		\node [style=none] (26) at (2.75, -0.5) {};
		\node [style=none] (27) at (2.25, -0.5) {...};
		\node [style=Z phase dot] (30) at (9, 0.5) {$\ \alpha_n \!\mp\! \frac\pi2\ $};
		\node [style=none] (31) at (-6.5, -2) {};
		\node [style=Z phase dot] (32) at (-6.5, -1) {$\alpha_2$};
		\node [style=none] (33) at (-7, -2) {...};
		\node [style=none] (34) at (-7.5, -2) {};
		\node [style=Z phase dot] (35) at (-3.25, -1) {$\ \alpha_{n\!-\!1}\ $};
		\node [style=none] (36) at (-3.25, -2) {};
		\node [style=none] (37) at (-2.25, -2) {};
		\node [style=none] (38) at (-2.75, -2) {...};
		\node [style=Z phase dot] (39) at (4, -1) {$\ \alpha_2\!\mp\!\frac\pi2\ $};
		\node [style=none] (40) at (3.5, -2) {...};
		\node [style=none] (41) at (3, -2) {};
		\node [style=none] (42) at (4, -2) {};
		\node [style=Z phase dot] (43) at (7.75, -1) {$\ \alpha_{n\!-\!1} \!\mp\! \frac\pi2\ $};
		\node [style=none] (44) at (8.75, -2) {};
		\node [style=none] (45) at (8.25, -2) {...};
		\node [style=none] (46) at (7.75, -2) {};
		\node [style=none] (47) at (6, 0) {...};
		\node [style=none] (48) at (-5, 1.75) {\color{red} $*$};
	\end{pgfonlayer}
	\begin{pgfonlayer}{edgelayer}
		\draw (5.center) to (2);
		\draw (6.center) to (2);
		\draw (7.center) to (4);
		\draw (8.center) to (4);
		\draw [style=hadamard edge] (2) to (0);
		\draw [style=hadamard edge] (4) to (0);
		\draw (18.center) to (24);
		\draw (26.center) to (24);
		\draw (20.center) to (30);
		\draw (17.center) to (30);
		\draw (34.center) to (32);
		\draw (31.center) to (32);
		\draw (36.center) to (35);
		\draw (37.center) to (35);
		\draw (41.center) to (39);
		\draw (42.center) to (39);
		\draw (46.center) to (43);
		\draw (44.center) to (43);
		\draw [style=hadamard edge] (24) to (30);
		\draw [style=hadamard edge] (30) to (43);
		\draw [style=hadamard edge] (43) to (39);
		\draw [style=hadamard edge] (39) to (24);
		\draw [style=hadamard edge] (24) to (43);
		\draw [style=hadamard edge] (39) to (30);
		\draw [style=hadamard edge] (0) to (32);
		\draw [style=hadamard edge] (0) to (35);
	\end{pgfonlayer}
\end{tikzpicture} \hfill
\end{equation}
where the RHS is obtained from the LHS by performing a local complementation at the marked vertex, removing it, and updating the phases as shown.
\end{lemma}
\begin{lemma}\label{lem:pivot-simp}
The following rule holds in the \zxcalculus:
\begin{equation}\label{eq:pivot-simp}
  \hfill \begin{tikzpicture}
	\begin{pgfonlayer}{nodelayer}
		\node [style=Z phase dot] (0) at (-11.5, 1.5) {$j\pi$};
		\node [style=Z phase dot] (2) at (-13.25, 1) {$\alpha_1$};
		\node [style=none] (16) at (-3.5, 0) {$=$};
		\node [style=Z phase dot] (32) at (-13.25, -0.75) {$\alpha_n$};
		\node [style=Z phase dot] (37) at (-9, -2.25) {$\beta_1$};
		\node [style=Z phase dot] (41) at (-9, -0.5) {$\beta_n$};
		\node [style=Z phase dot] (45) at (-5, 1) {$\gamma_1$};
		\node [style=Z phase dot] (49) at (-5, -0.75) {$\gamma_n$};
		\node [style=Z phase dot] (51) at (-6.5, 1.5) {$k\pi$};
		\node [style=none, rotate=90] (52) at (-13.25, 0.125) {...};
		\node [style=none, rotate=90] (53) at (-9, -1.375) {...};
		\node [style=none, rotate=90] (54) at (-5, 0.125) {...};
		\node [style=Z phase dot] (55) at (-1.25, -0.25) {$\ \alpha_n+k\pi\ $};
		\node [style=Z phase dot] (56) at (3.5, -1.5) {$\ \beta_n+(j+k+1)\pi\ $};
		\node [style=none, rotate=90] (57) at (3.5, -2.375) {...};
		\node [style=Z phase dot] (58) at (3.5, -3.25) {$\ \beta_1+(j+k+1)\pi\ $};
		\node [style=Z phase dot] (59) at (8.5, 1.5) {$\ \gamma_1+j\pi\ $};
		\node [style=Z phase dot] (60) at (-1.25, 1.5) {$\ \alpha_1+k\pi\ $};
		\node [style=none, rotate=90] (61) at (8.5, 0.625) {...};
		\node [style=none, rotate=90] (62) at (-1.25, 0.625) {...};
		\node [style=Z phase dot] (63) at (8.5, -0.25) {$\ \gamma_n+j\pi\ $};
		\node [style=none] (64) at (-12.75, 1.75) {};
		\node [style=none] (65) at (-13.75, 1.75) {};
		\node [style=none] (66) at (-13.25, 1.75) {...};
		\node [style=none] (67) at (-13.75, -1.5) {};
		\node [style=none] (68) at (-13.25, -1.5) {...};
		\node [style=none] (69) at (-12.75, -1.5) {};
		\node [style=none] (70) at (-9.5, -3) {};
		\node [style=none] (71) at (-9, -3) {...};
		\node [style=none] (72) at (-8.5, -3) {};
		\node [style=none] (73) at (-9.5, 0.25) {};
		\node [style=none] (74) at (-9, 0.25) {...};
		\node [style=none] (75) at (-8.5, 0.25) {};
		\node [style=none] (76) at (-5.5, 1.75) {};
		\node [style=none] (77) at (-5, 1.75) {...};
		\node [style=none] (78) at (-4.5, 1.75) {};
		\node [style=none] (79) at (-5, -1.5) {...};
		\node [style=none] (80) at (-5.5, -1.5) {};
		\node [style=none] (81) at (-4.5, -1.5) {};
		\node [style=none] (82) at (-1.75, 2.25) {};
		\node [style=none] (83) at (-1.25, 2.25) {...};
		\node [style=none] (84) at (-0.75, 2.25) {};
		\node [style=none] (85) at (-1.25, -1) {...};
		\node [style=none] (86) at (-1.75, -1) {};
		\node [style=none] (87) at (-0.75, -1) {};
		\node [style=none] (88) at (3, -0.75) {};
		\node [style=none] (89) at (3.5, -0.75) {...};
		\node [style=none] (90) at (4, -0.75) {};
		\node [style=none] (91) at (3.5, -4) {...};
		\node [style=none] (92) at (3, -4) {};
		\node [style=none] (93) at (4, -4) {};
		\node [style=none] (94) at (8, 2.25) {};
		\node [style=none] (95) at (8.5, 2.25) {...};
		\node [style=none] (96) at (9, 2.25) {};
		\node [style=none] (97) at (8.5, -1) {...};
		\node [style=none] (98) at (8, -1) {};
		\node [style=none] (99) at (9, -1) {};
		\node [style=none] (100) at (-11.5, 2.25) {\color{red} $*$};
		\node [style=none] (101) at (-6.5, 2.25) {\color{red} $*$};
	\end{pgfonlayer}
	\begin{pgfonlayer}{edgelayer}
		\draw [style=hadamard edge] (2) to (0);
		\draw [style=hadamard edge] (0) to (32);
		\draw [style=hadamard edge] (0) to (51);
		\draw [style=hadamard edge] (0) to (37);
		\draw [style=hadamard edge] (0) to (41);
		\draw [style=hadamard edge] (51) to (37);
		\draw [style=hadamard edge] (51) to (41);
		\draw [style=hadamard edge] (51) to (45);
		\draw [style=hadamard edge] (51) to (49);
		\draw [style=hadamard edge] (60) to (56);
		\draw [style=hadamard edge] (60) to (58);
		\draw [style=hadamard edge] (55) to (56);
		\draw [style=hadamard edge] (55) to (58);
		\draw [style=hadamard edge] (59) to (56);
		\draw [style=hadamard edge] (59) to (58);
		\draw [style=hadamard edge] (56) to (63);
		\draw [style=hadamard edge] (63) to (58);
		\draw [style=hadamard edge] (60) to (59);
		\draw [style=hadamard edge] (60) to (63);
		\draw [style=hadamard edge] (55) to (63);
		\draw [style=hadamard edge] (55) to (59);
		\draw (64.center) to (2);
		\draw (2) to (65.center);
		\draw (67.center) to (32);
		\draw (69.center) to (32);
		\draw (70.center) to (37);
		\draw (72.center) to (37);
		\draw (41) to (75.center);
		\draw (41) to (73.center);
		\draw (45) to (78.center);
		\draw (45) to (76.center);
		\draw (49) to (80.center);
		\draw (49) to (81.center);
		\draw (60) to (82.center);
		\draw (60) to (84.center);
		\draw (59) to (94.center);
		\draw (59) to (96.center);
		\draw (98.center) to (63);
		\draw (99.center) to (63);
		\draw (86.center) to (55);
		\draw (87.center) to (55);
		\draw (56) to (88.center);
		\draw (56) to (90.center);
		\draw (92.center) to (58);
		\draw (93.center) to (58);
	\end{pgfonlayer}
\end{tikzpicture} \hfill
\end{equation}
where the RHS is obtained from the LHS by performing a pivot at the marked vertices, removing them, and updating the phases as shown.
\end{lemma}
The proofs of these lemmas can be found in Appendix~\ref{sec:zx-reduction-rules}.

We can additionally apply~\eqref{eq:pivot-simp} to remove an interior Pauli spider that is adjacent to a boundary spider. To do this, we first turn the boundary spider into a (phaseless) interior spider via the following transformation, which follows from the rules \SpiderRule, \IdentityRule, and \HHRule:
\begin{equation}\label{eq:pivot-unfuse}
  \begin{tikzpicture}
	\begin{pgfonlayer}{nodelayer}
		\node [style=Z phase dot] (0) at (-8.25, 1) {$j\pi$};
		\node [style=Z phase dot] (2) at (-9.75, 0) {$\alpha_1$};
		\node [style=none] (16) at (-0.25, 0) {$=$};
		\node [style=Z phase dot] (32) at (-6.75, 0) {$\alpha_n$};
		\node [style=Z phase dot] (51) at (-3.5, 1) {$\alpha$};
		\node [style=none] (52) at (-8.25, 0) {...};
		\node [style=none] (54) at (-3.5, -0.75) {...};
		\node [style=none] (64) at (-9.25, -0.75) {};
		\node [style=none] (65) at (-10.25, -0.75) {};
		\node [style=none] (66) at (-9.75, -0.75) {...};
		\node [style=none] (67) at (-7.25, -0.75) {};
		\node [style=none] (68) at (-6.75, -0.75) {...};
		\node [style=none] (69) at (-6.25, -0.75) {};
		\node [style=none] (100) at (-1.25, 1) {};
		\node [style=none] (104) at (-9.25, 1) {\color{red}$u$};
		\node [style=none] (139) at (-2.75, -0.75) {};
		\node [style=none] (140) at (-4.25, -0.75) {};
		\node [style=Z phase dot] (141) at (3, 1) {$j\pi$};
		\node [style=Z phase dot] (142) at (1.5, 0) {$\alpha_1$};
		\node [style=Z phase dot] (143) at (4.25, 0) {$\alpha_n$};
		\node [style=Z dot] (144) at (7, 1) {};
		\node [style=none] (145) at (3, 0) {...};
		\node [style=none] (146) at (7, -0.75) {...};
		\node [style=none] (147) at (2, -0.75) {};
		\node [style=none] (148) at (1, -0.75) {};
		\node [style=none] (149) at (1.5, -0.75) {...};
		\node [style=none] (150) at (3.75, -0.75) {};
		\node [style=none] (151) at (4.25, -0.75) {...};
		\node [style=none] (152) at (4.75, -0.75) {};
		\node [style=none] (153) at (11.5, 1) {};
		\node [style=none] (155) at (7.75, -0.75) {};
		\node [style=none] (156) at (6.25, -0.75) {};
		\node [style=Z dot] (158) at (8.5, 1) {};
		\node [style=Z phase dot] (159) at (10.5, 1) {$\alpha$};
		\node [style=none] (160) at (-4.25, 0.5) {\color{red} $v$};
		\node [style=hadamard] (162) at (9.5, 1) {};
		\node [style=none] (163) at (8.5, 0.5) {\color{blue} $w$};
		\node [style=none] (164) at (2, 1) {\color{red} $u$};
		\node [style=none] (165) at (6.25, 0.5) {\color{red} $v$};
	\end{pgfonlayer}
	\begin{pgfonlayer}{edgelayer}
		\draw [style=hadamard edge] (2) to (0);
		\draw [style=hadamard edge] (0) to (32);
		\draw [style=hadamard edge] (0) to (51);
		\draw (64.center) to (2);
		\draw (2) to (65.center);
		\draw (67.center) to (32);
		\draw (69.center) to (32);
		\draw (51) to (100.center);
		\draw [style=hadamard edge] (142) to (141);
		\draw [style=hadamard edge] (141) to (143);
		\draw [style=hadamard edge] (141) to (144);
		\draw (147.center) to (142);
		\draw (142) to (148.center);
		\draw (150.center) to (143);
		\draw (152.center) to (143);
		\draw [style=hadamard edge] (144) to (158);
		\draw (158) to (159);
		\draw (159) to (153.center);
		\draw (51) to (140.center);
		\draw (51) to (139.center);
		\draw (144) to (156.center);
		\draw (144) to (155.center);
	\end{pgfonlayer}
\end{tikzpicture}
\end{equation}
After this transformation, we can perform~\eqref{eq:pivot-simp} on the two spiders labelled $\{u, v\}$ to remove them. For our purposes, we can then treat $w$ as a boundary spider, and save the single-qubit unitaries to the right of $w$ separately.

\begin{theorem}\label{thm:simp}
There exists a terminating procedure which turns any graph-like ZX-diagram $D$ into a graph-like ZX-diagram $D'$ (up to single-qubit unitaries on inputs/outputs) which does not contain
\begin{enumerate}
\item interior proper Clifford spiders,
\item adjacent pairs of interior Pauli spiders,
\item and interior Pauli spiders adjacent to a boundary spider.
\end{enumerate}
In particular, if $D$ only contains Clifford spiders, then $D'$ contains no interior spiders.
\end{theorem}

\begin{proof}
Starting with a graph-like \zxdiagram, mark every interior Clifford spider for simplification. Then repeat the steps below (followed by \eqref{eq:parallel-edges-loops} to remove parallel edges), until no rule matches:
\begin{itemize}
\item Apply \eqref{eq:lc-simp} to remove a marked proper Clifford spider.
\item Apply \eqref{eq:pivot-simp} to remove an adjacent pair of marked Pauli spiders.
\item Apply \eqref{eq:pivot-unfuse} to a boundary adjacent to a marked Pauli spider and immediately apply~\eqref{eq:pivot-simp}.
\end{itemize}
All 3 rules remove at least 1 marked node, so the procedure terminates. Since none of these rules will introduce non-Clifford spiders, if $D$ contains only Clifford spiders so too does $D'$. Since the only possibilities for interior spiders in $D'$ are covered by cases 1-3, $D'$ contains no interior spiders.
\end{proof}

\begin{corollary}\label{cor:simp-preserves-gflow}
For $D$ and $D'$ the \zxdiagrams from Theorem \ref{thm:simp}, if $\UG D$ admits a focused gFlow, then so too does $\UG{D'}$.
\end{corollary}

\begin{proof}
It was shown in Theorem~\ref{thm:gflow-preserve} that the transformations \eqref{eq:lc-simp} and \eqref{eq:pivot-simp} preserve focused gFlow. The transformation \eqref{eq:pivot-unfuse} acts as an \textit{input/output extension} on $\UG D$. That is, for $X \in \{ I, O \}$, $v \in X$, it adds a new vertex $w$ and an edge $vw$, and sets $X' := (X \backslash \{v\}) \cup \{w\}$. It was shown in \cite{mhalla2011graph} that these transformations preserve focused gFlow. 
\end{proof}

This simplification procedure is very effective, especially for circuits with few non-Clifford gates. We can see this by considering a randomly-generated circuit, where $\sim2\%$ of the gates are non-Clifford:
\[ \scalebox{0.85}{\input{big-example-4rows.tikz}} \]
This circuit has 195 gates, with 4 non-Clifford gates. If we apply the simplification procedure, we obtain a small `skeleton' of the circuit, containing 12 spiders:
\begin{equation}\label{eq:big-example-skeleton}
	\begin{tikzpicture}
	\begin{pgfonlayer}{nodelayer}
		\node [style=none] (0) at (-6.5, 1.5) {};
		\node [style=none] (1) at (-6.5, 0.5) {};
		\node [style=none] (2) at (-6.5, -0.5) {};
		\node [style=none] (3) at (-6.5, -1.5) {};
		\node [style=Z dot] (5) at (-4, 0.5) {};
		\node [style=Z phase dot] (10) at (-4, -0.5) {$\frac{\pi}{2}$};
		\node [style=Z phase dot] (24) at (-1.5, -1.5) {$\frac{\pi}{4}$};
		\node [style=Z phase dot] (49) at (-1.5, 0.75) {$\frac{5\pi}{4}$};
		\node [style=Z phase dot] (58) at (-4, -1.5) {$\frac{\pi}{2}$};
		\node [style=Z dot] (66) at (-4, 1.5) {};
		\node [style=Z phase dot] (96) at (0, 0) {$\frac{\pi}{4}$};
		\node [style=Z phase dot] (210) at (1.25, -1.5) {$\frac{\pi}{4}$};
		\node [style=Z dot] (272) at (4, 0.5) {};
		\node [style=Z dot] (273) at (4, -0.5) {};
		\node [style=Z phase dot] (283) at (4, -1.5) {$\pi$};
		\node [style=Z phase dot] (284) at (4, 1.5) {$\pi$};
		\node [style=none] (285) at (6, 1.5) {};
		\node [style=none] (286) at (6, 0.5) {};
		\node [style=none] (287) at (6, -0.5) {};
		\node [style=none] (288) at (6, -1.5) {};
		\node [style=hadamard] (289) at (-5.25, -0.5) {};
		\node [style=hadamard] (290) at (-5.25, -1.5) {};
		\node [style=hadamard] (291) at (5, -0.5) {};
		\node [style=hadamard] (292) at (5, -1.5) {};
	\end{pgfonlayer}
	\begin{pgfonlayer}{edgelayer}
		\draw (0.center) to (66);
		\draw (1.center) to (5);
		\draw (2.center) to (10);
		\draw (3.center) to (58);
		\draw [style=hadamard edge] (5) to (66);
		\draw [style=hadamard edge] (5) to (24);
		\draw [style=hadamard edge, bend right=60] (5) to (58);
		\draw [style=hadamard edge] (10) to (283);
		\draw [style=hadamard edge, bend left=60] (10) to (66);
		\draw [style=hadamard edge] (10) to (49);
		\draw [style=hadamard edge] (10) to (210);
		\draw [style=hadamard edge] (10) to (24);
		\draw [style=hadamard edge] (10) to (58);
		\draw [style=hadamard edge] (24) to (66);
		\draw [style=hadamard edge] (24) to (272);
		\draw [style=hadamard edge] (24) to (284);
		\draw [style=hadamard edge] (24) to (96);
		\draw [style=hadamard edge] (24) to (58);
		\draw [style=hadamard edge] (49) to (96);
		\draw [style=hadamard edge] (49) to (66);
		\draw [style=hadamard edge] (49) to (272);
		\draw [style=hadamard edge] (49) to (273);
		\draw [style=hadamard edge] (58) to (96);
		\draw [style=hadamard edge] (66) to (272);
		\draw [style=hadamard edge] (66) to (284);
		\draw [style=hadamard edge] (96) to (210);
		\draw [style=hadamard edge] (210) to (273);
		\draw [style=hadamard edge] (210) to (283);
		\draw [style=hadamard edge] (272) to (284);
		\draw (272) to (286.center);
		\draw (273) to (287.center);
		\draw (283) to (288.center);
		\draw (284) to (285.center);
	\end{pgfonlayer}
\end{tikzpicture}
\end{equation}
Note the non-Clifford phases now clearly appear on the 4 interior spiders. In general, the interior spiders will either be non-Clifford or be Pauli with only non-Clifford neighbours. Any other possibilities will be removed.

\begin{remark}\label{rem:simp-performance}
Let $n$ denote the amount of spiders in the diagram. As each step of Theorem~\ref{thm:simp} removes a spider, the amount of steps is upper-bounded by $n$. Each step toggles the connectivity of some subset of the vertices in the graph, and hence the elementary graph operations per step is upper-bounded by $n^2$. We see then that the full simplification procedure has worst-case complexity $O(n^3)$. However, since the graphs coming from circuits are far from being fully connected, we see a much better scaling in practice, with the complexity lying between $O(n)$ and $O(n^2)$ (see Ref.~\cite{pyzx} for more details). Our implementation simplifies the example above in $\sim$25ms on a laptop computer, scaling up to a few minutes for circuits with 10-20K gates.
\end{remark}

\section{Circuit extraction of Clifford circuits}\label{sec:clifford}

As described in the previous section, when a Clifford circuit is fed into our simplification routine, no interior spiders will be left, and hence the diagram is in \emph{GS-LC} form (\emph{graph-state with local Cliffords}). It is straightforward to extract a Clifford circuit from a GS-LC form. First, we unfuse spiders to create local Cliffords and CZ gates on the inputs and outputs, then apply the \HadamardRule rule to the spiders on the right side of the resulting bipartite graph:
\begin{equation}\label{eq:gs-lc}
\scalebox{1.0}{\begin{tikzpicture}
	\begin{pgfonlayer}{nodelayer}
		\node [style=box] (0) at (-3.5, 1.75) {LC};
		\node [style=none] (3) at (-5.25, 1.75) {};
		\node [style=none] (4) at (-5.25, 0.25) {};
		\node [style=none] (5) at (-5.25, -1.75) {};
		\node [style=box] (6) at (-3.5, 0.25) {LC};
		\node [style=box] (7) at (-3.5, -1.75) {LC};
		\node [style=Z dot] (8) at (-1.75, 1.75) {};
		\node [style=Z dot] (9) at (-1.75, 0.25) {};
		\node [style=Z dot] (10) at (-1.75, -1.75) {};
		\node [style=none] (11) at (-3.5, -0.75) {...};
		\node [style=none] (12) at (5.25, -1.75) {};
		\node [style=box] (13) at (3.5, 1.75) {LC};
		\node [style=none] (14) at (5.25, 0.25) {};
		\node [style=box] (15) at (3.5, 0.25) {LC};
		\node [style=none] (16) at (5.25, 1.75) {};
		\node [style=none] (17) at (3.5, -0.75) {...};
		\node [style=Z dot] (18) at (1.75, 1.75) {};
		\node [style=Z dot] (19) at (1.75, -1.75) {};
		\node [style=box] (20) at (3.5, -1.75) {LC};
		\node [style=Z dot] (21) at (1.75, 0.25) {};
		\node [style=none] (22) at (-0.75, 1.75) {};
		\node [style=none] (23) at (-0.75, 0.5) {};
		\node [style=none] (24) at (-0.75, -0.5) {};
		\node [style=none] (25) at (-0.75, -1.5) {};
		\node [style=none] (26) at (-0.75, -1.75) {};
		\node [style=none] (27) at (-0.75, 0) {};
		\node [style=none] (28) at (-0.75, 1.25) {};
		\node [style=none] (29) at (0.75, -1.75) {};
		\node [style=none] (30) at (0.75, 1.75) {};
		\node [style=none] (31) at (0.75, -1.75) {};
		\node [style=none] (32) at (0.75, 1.25) {};
		\node [style=none] (33) at (0.75, 0.75) {};
		\node [style=none] (34) at (0.75, -0.5) {};
		\node [style=none] (35) at (0.75, 0) {};
		\node [style=none] (36) at (0.75, -1.25) {};
		\node [style=none, rotate=90] (37) at (0, -0.25) {\Huge$\cdots$};
		\node [style=none] (38) at (6.5, 0) {$=$};
		\node [style=none] (39) at (15, 0.75) {};
		\node [style=none] (40) at (15, 1.75) {};
		\node [style=none] (41) at (15, -1.25) {};
		\node [style=box] (42) at (19.75, 0.25) {LC};
		\node [style=box] (43) at (19.75, 1.75) {LC};
		\node [style=none] (44) at (7.75, 1.75) {};
		\node [style=X dot] (45) at (16, -1.75) {};
		\node [style=none] (46) at (19.75, -0.75) {...};
		\node [style=none] (47) at (13.5, -1) {};
		\node [style=none] (48) at (7.75, 0.25) {};
		\node [style=box] (49) at (9.25, 0.25) {LC};
		\node [style=none] (50) at (21.5, -1.75) {};
		\node [style=box] (51) at (9.25, -1.75) {LC};
		\node [style=none] (52) at (15, 0) {};
		\node [style=none] (53) at (13.5, -0.5) {};
		\node [style=none] (54) at (13.5, 0.5) {};
		\node [style=box] (55) at (19.75, -1.75) {LC};
		\node [style=none] (56) at (13.5, 1.75) {};
		\node [style=none] (57) at (15, -0.5) {};
		\node [style=Z dot] (58) at (12.5, -1.75) {};
		\node [style=none] (59) at (13.5, 1.25) {};
		\node [style=Z dot] (60) at (12.5, 1.75) {};
		\node [style=none] (61) at (13.5, -1.75) {};
		\node [style=Z dot] (62) at (12.5, 0.25) {};
		\node [style=none] (63) at (15, 1.25) {};
		\node [style=none] (64) at (21.5, 0.25) {};
		\node [style=none] (65) at (9.25, -0.75) {...};
		\node [style=X dot] (66) at (16, 1.75) {};
		\node [style=X dot] (67) at (16, 0.25) {};
		\node [style=none] (68) at (21.5, 1.75) {};
		\node [style=none] (69) at (7.75, -1.75) {};
		\node [style=none] (70) at (13.5, 0) {};
		\node [style=none] (71) at (15, -1.25) {};
		\node [style=box] (72) at (9.25, 1.75) {LC};
		\node [style=none, rotate=90] (73) at (14.25, -0.25) {\Huge$\cdots$};
		\node [style=none] (74) at (15, -1.75) {};
		\node [style=Z dot] (75) at (11.5, 1.75) {};
		\node [style=Z dot] (76) at (11.5, 0.25) {};
		\node [style=Z dot] (77) at (10.5, 1.75) {};
		\node [style=Z dot] (78) at (10.5, -1.75) {};
		\node [style=Z dot] (79) at (18.25, 0.25) {};
		\node [style=Z dot] (80) at (18.25, -1.75) {};
		\node [style=hadamard] (81) at (10.5, -1) {};
		\node [style=hadamard] (82) at (11.5, 1) {};
		\node [style=hadamard] (83) at (18.25, -0.75) {};
		\node [style=hadamard] (84) at (17, 1.75) {};
		\node [style=hadamard] (85) at (17, 0.25) {};
		\node [style=hadamard] (86) at (17, -1.75) {};
		\node [style=none] (87) at (16.5, -2.375) {};
		\node [style=none] (89) at (12, -2.375) {};
		\node [style=none] (90) at (14.25, -3) {$\mathcal{P}$};
	\end{pgfonlayer}
	\begin{pgfonlayer}{edgelayer}
		\draw (3.center) to (0);
		\draw (4.center) to (6);
		\draw (5.center) to (7);
		\draw (0) to (8);
		\draw (6) to (9);
		\draw (7) to (10);
		\draw (16.center) to (13);
		\draw (14.center) to (15);
		\draw (12.center) to (20);
		\draw (13) to (18);
		\draw (15) to (21);
		\draw (20) to (19);
		\draw [style=hadamard edge] (8) to (22.center);
		\draw [style=hadamard edge] (9) to (23.center);
		\draw [style=hadamard edge] (9) to (24.center);
		\draw [style=hadamard edge] (10) to (25.center);
		\draw [style=hadamard edge] (10) to (26.center);
		\draw [style=hadamard edge] (9) to (27.center);
		\draw [style=hadamard edge] (8) to (28.center);
		\draw [style=hadamard edge] (18) to (30.center);
		\draw [style=hadamard edge] (18) to (32.center);
		\draw [style=hadamard edge] (18) to (33.center);
		\draw [style=hadamard edge] (21) to (35.center);
		\draw [style=hadamard edge] (21) to (34.center);
		\draw [style=hadamard edge] (19) to (36.center);
		\draw [style=hadamard edge] (19) to (31.center);
		\draw [style=hadamard edge] (19) to (29.center);
		\draw [style=hadamard edge] (8) to (9);
		\draw [style=hadamard edge, bend right=45, looseness=0.75] (8) to (10);
		\draw [style=hadamard edge] (21) to (19);
		\draw (44.center) to (72);
		\draw (48.center) to (49);
		\draw (69.center) to (51);
		\draw (72) to (60);
		\draw (49) to (62);
		\draw (51) to (58);
		\draw (68.center) to (43);
		\draw (64.center) to (42);
		\draw (50.center) to (55);
		\draw (43) to (66);
		\draw (42) to (67);
		\draw (55) to (45);
		\draw (60) to (56.center);
		\draw (62) to (54.center);
		\draw (62) to (53.center);
		\draw (58) to (47.center);
		\draw (58) to (61.center);
		\draw (62) to (70.center);
		\draw (60) to (59.center);
		\draw (66) to (40.center);
		\draw (66) to (63.center);
		\draw (66) to (39.center);
		\draw (67) to (52.center);
		\draw (67) to (57.center);
		\draw (45) to (41.center);
		\draw (45) to (71.center);
		\draw (45) to (74.center);
		\draw (77) to (78);
		\draw (75) to (76);
		\draw (79) to (80);
		\draw [style=brace edge] (87.center) to (89.center);
	\end{pgfonlayer}
\end{tikzpicture}}
\end{equation}
The part of the diagram marked $\mathcal P$ now has the form of a classical parity circuit. That is, it is a permutation of computational basis states where the output basis is given in terms of parities of the inputs, e.g.
$ \ket{x_1, x_2, x_3, x_4} \mapsto \ket{x_1 \oplus x_2, x_1 \oplus x_3, x_4, x_3}.$
Such a unitary can always be realised via CNOT gates, using standard techniques based on Gaussian elimination over $\mathbb F_2$ (see e.g.~\cite{markov2008optimal}). Hence, we obtain a Clifford circuit with 6 layers of gates:
\begin{equation*}
  \textrm{Local Clifford + CZ + CNOT + H + CZ + Local Clifford}
\end{equation*}

\noindent Using a strategy based on~\cite[Thm.~13]{Backens1}, we can apply the local complementation rule~\eqref{eq:gs-local-comp} to the LHS of \eqref{eq:gs-lc} to reduce all of the local Cliffords on the inputs to the set $\{S^n, H, ZH\}$ and the outputs to the set $\{S^n, H, HZ \}$ (cf. Appendix~\ref{sec:reduction-lc}). Hence, we can further refine the decomposition above into 8 layers:
\begin{equation}\label{eq:cliff-nf}
  \textrm{H + S + CZ + CNOT + H + CZ + S + H}
\end{equation}
where the S layer means powers of S gates (including $S^2 = Z$).

There are a variety of pseudo-normal forms for Clifford circuits in the literature, starting with the 11-layer form given by Aaronson and Gottesman~\cite{aaronsongottesman2004} and the courser-grained 5-layer form of Dehaene and De Moore~\cite{dehaene2003clifford}, which led to improved versions by Maslov and Roetteler~\cite{maslov2018shorter} and van den Nest~\cite{nest2010clifford}, respectively. While there are some superficial similarities between our normal forms and these earlier ones, there is at least one notable difference. All of the forms mentioned above require at least two distinct CNOT layers, but (with the exception of Ref.~\cite{aaronsongottesman2004}) require just a single later of Hadamard gates. On the other hand, our normal form has just a single CNOT layer, at the cost of multiple Hadamard layers. We will now see that this trade-off has some nice consequences.


In Ref.~\cite{maslov2018shorter} the authors argued that since there are $2^{2n^2 +O(n)}$ Clifford unitaries on $n$ qubits, one needs at least $2n^2 + O(n)$ Boolean degrees of freedom to specify all the $n$ qubit Clifford unitaries, and furthermore they found a normal form which has the same number of degrees of freedom and hence is asymptotically optimal in this sense. They also study the problem of finding a Clifford pseudo-normal form that has the lowest 2-qubit gate depth when restricted to a linear nearest neighbour architecture and find a different normal form that is bounded by a $14n-4$ 2-qubit gate depth. The decomposition~\eqref{eq:cliff-nf} improves on this bound, while at the same time also satisfying the $2n^2 + O(n)$ asymptotically optimal gate count.

\begin{proposition}\label{prop:cliff-gate-depth}
    The GS-LC pseudo-normal form on $n$ qubits has an asymptotically optimal number of degrees of freedom $2n^2 + O(n)$. Furthermore, any Clifford unitary in this normal form can be mapped to a linear nearest neighbour architecture with a 2-qubit gate depth of $9n-2$.
\end{proposition}

\begin{proof}
    The argument follows closely the one given in Ref.~\cite{maslov2018shorter}, but for a different normal form. Our normal form has 5 layers of single qubit Clifford gates. The Hadamard layers each add at most $n$ gates, while the Z-phase layers add at most $3n$ gates, hence these layers only add a linear amount of degrees of freedom. Each CZ layer adds $n^2/2$ degrees of freedom, while a CNOT layer adds $n^2$ degrees of freedom~\cite[Section I]{maslov2018shorter}. Hence the total degrees of freedom is given by $3n+2 \cdot 3n + 2 \cdot n^2/2 + n^2 = 2n^2 + O(n)$.

    For the 2-qubit gate depth, we note that any CNOT circuit can be implemented on a linear nearest neighbour architecture in depth $5n$~\cite{kutin2007computation}. A CZ circuit followed or preceded by a series of SWAP gates that reverses the qubit order can be implemented in depth $2n+2$ on a linear nearest neighbour architecture~\cite[Thm.~6]{maslov2018shorter}. But by~\cite[Cor. 7]{maslov2018shorter}, when we have two of these CZ circuits, possibly separated by some other gates, then this pair of CZ circuits can be implemented in 2-qubit gate depth $4n-2$. As the only layers in our pseudo-normal form that contribute to the 2-qubit gate depth are two CZ layers and a CNOT layer we then indeed have a total depth of $5n + 4n-2 = 9n - 2$.
\end{proof}


\section{Circuit extraction for general circuits}\label{sec:cliffordT}

If a Clifford+T circuit -- or a more general circuit containing gates $Z_\alpha$ or $X_\alpha$ for $\alpha \neq k\frac\pi2$ -- is fed into our simplification routine, there might still be interior spiders left, and hence the diagram will not be in GS-LC form. 
Our general strategy for extracting a circuit from a diagram produced by our simplification routine is to progress through the diagram from right-to-left, consuming vertices to the left and introducing quantum gates on the right. For an $n$-qubit circuit, we maintain a layer of $n$ vertices called the \textit{frontier} between these two parts. The goal is therefore to progress the frontier from the outputs all the way to the inputs, at which point we are left with an extracted circuit. We will see that the existence of a focused gFlow guarantees we can always progress the frontier leftward in this manner.

Much like in the Clifford case, our extraction procedure relies on applying (a variant of) Gaussian elimination on the adjacency matrix of a graph. Hence, we first establish a correspondence between primitive row operations and CNOT gates, which is proven in Appendix~\ref{sec:zx-reduction-rules}.

\begin{proposition}\label{prop:cnotgflow}
  For any \zxdiagram $D$, the following equation holds:
  \begin{equation*}
  \scalebox{1.0}{\begin{tikzpicture}
	\begin{pgfonlayer}{nodelayer}
		\node [style=none] (0) at (4.5, 1.5) {};
		\node [style=none] (1) at (4.5, 0.75) {};
		\node [style=none] (2) at (4.5, -1.75) {};
		\node [style=Z dot] (3) at (2.75, 1.5) {};
		\node [style=Z dot] (4) at (2.75, 0.75) {};
		\node [style=Z dot] (5) at (2.75, -1.75) {};
		\node [style=Z dot] (6) at (0, 1.5) {};
		\node [style=Z dot] (7) at (0, 0.75) {};
		\node [style=Z dot] (8) at (0, -1.5) {};
		\node [style=Z dot] (13) at (-4.5, 1.5) {};
		\node [style=Z dot] (14) at (-4.5, 0.5) {};
		\node [style=Z dot] (15) at (-4.5, -1.75) {};
		\node [style=none] (16) at (-6.25, 1.5) {};
		\node [style=none] (17) at (-6.25, 0.5) {};
		\node [style=none] (18) at (-6.25, -1.75) {};
		\node [style=none] (19) at (-3, 1.5) {};
		\node [style=none] (20) at (-3, 1) {};
		\node [style=none] (21) at (-3, 0) {};
		\node [style=none] (22) at (-3, -1) {};
		\node [style=none] (23) at (-1, -1) {};
		\node [style=none] (24) at (-1, 0.25) {};
		\node [style=none] (25) at (-1, 1.25) {};
		\node [style=none] (26) at (-1, 1.75) {};
		\node [style=none] (27) at (1, 1.25) {};
		\node [style=none] (28) at (1, 0.5) {};
		\node [style=none] (29) at (1, -0.75) {};
		\node [style=none] (30) at (1.75, -0.5) {};
		\node [style=none] (31) at (1.75, 0.5) {};
		\node [style=none] (32) at (1.75, 1.25) {};
		\node [style=none, rotate=90] (33) at (1.375, 0.5) {\ldots};
		\node [style=none] (34) at (3.25, -2.5) {};
		\node [style=none] (35) at (-0.5, -2.5) {};
		\node [style=none] (36) at (1.5, -3) {$M$};
		\node [style=none] (38) at (1, 0.75) {};
		\node [style=none] (39) at (1, -0.25) {};
		\node [style=none] (40) at (1.75, 0) {};
		\node [style=none] (41) at (1.75, 1) {};
		\node [style=none, rotate=90] (43) at (3.5, -0.25) {\ldots};
		\node [style=none, rotate=90] (44) at (-5.5, -0.5) {\ldots};
		\node [style=box, minimum height=2cm, minimum width=1cm] (45) at (-2, 0.25) {$D$};
		\node [style=none] (46) at (6.25, 0) {$=$};
		\node [style=none] (47) at (18.5, 1.5) {};
		\node [style=none] (48) at (18.5, 0.5) {};
		\node [style=none] (49) at (18.5, -1.75) {};
		\node [style=Z dot] (50) at (16.5, 1.5) {};
		\node [style=Z dot] (51) at (16.5, 0.5) {};
		\node [style=Z dot] (52) at (16.5, -1.75) {};
		\node [style=Z dot] (53) at (13.75, 1.5) {};
		\node [style=Z dot] (54) at (13.75, 0.5) {};
		\node [style=Z dot] (55) at (13.75, -1.5) {};
		\node [style=Z dot] (56) at (9.25, 1.5) {};
		\node [style=Z dot] (57) at (9.25, 0.5) {};
		\node [style=Z dot] (58) at (9.25, -1.5) {};
		\node [style=none] (59) at (7.75, 1.5) {};
		\node [style=none] (60) at (7.75, 0.5) {};
		\node [style=none] (61) at (7.75, -1.5) {};
		\node [style=none] (62) at (10.75, 1.5) {};
		\node [style=none] (63) at (10.75, 1) {};
		\node [style=none] (64) at (10.75, 0) {};
		\node [style=none] (65) at (10.75, -1) {};
		\node [style=none] (66) at (12.75, -1) {};
		\node [style=none] (67) at (12.75, 0) {};
		\node [style=none] (68) at (12.75, 1) {};
		\node [style=none] (69) at (12.75, 1.75) {};
		\node [style=none] (70) at (14.75, 1.25) {};
		\node [style=none] (71) at (14.75, 0) {};
		\node [style=none] (72) at (14.75, -1) {};
		\node [style=none] (73) at (15.5, -0.75) {};
		\node [style=none] (74) at (15.5, 0.25) {};
		\node [style=none] (75) at (15.5, 1.25) {};
		\node [style=none, rotate=90] (76) at (15.125, 0.25) {\ldots};
		\node [style=none] (77) at (17, -2.5) {};
		\node [style=none] (78) at (13.25, -2.5) {};
		\node [style=none] (79) at (15.25, -3) {$M^\prime$};
		\node [style=none] (80) at (14.75, 0.75) {};
		\node [style=none] (81) at (14.75, -0.5) {};
		\node [style=none] (82) at (15.5, -0.25) {};
		\node [style=none] (83) at (15.5, 0.75) {};
		\node [style=none, rotate=90] (84) at (17.5, -0.5) {\ldots};
		\node [style=none, rotate=90] (85) at (8.5, -0.5) {\ldots};
		\node [style=X dot] (86) at (17.5, 1.5) {};
		\node [style=Z dot] (87) at (17.5, 0.5) {};
		\node [style=box, minimum height=2cm, minimum width=1cm] (88) at (11.75, 0.25) {$D$};
	\end{pgfonlayer}
	\begin{pgfonlayer}{edgelayer}
		\draw (16.center) to (13);
		\draw (14) to (17.center);
		\draw (18.center) to (15);
		\draw (2.center) to (5);
		\draw (1.center) to (4);
		\draw (0.center) to (3);
		\draw [style=hadamard edge] (6) to (26.center);
		\draw [style=hadamard edge] (6) to (25.center);
		\draw [style=hadamard edge] (6) to (24.center);
		\draw [style=hadamard edge] (7) to (24.center);
		\draw [style=hadamard edge] (7) to (25.center);
		\draw [style=hadamard edge] (8) to (23.center);
		\draw [style=hadamard edge] (15) to (22.center);
		\draw [style=hadamard edge] (15) to (20.center);
		\draw [style=hadamard edge] (14) to (21.center);
		\draw [style=hadamard edge] (14) to (20.center);
		\draw [style=hadamard edge] (13) to (19.center);
		\draw [style=hadamard edge] (13) to (20.center);
		\draw [style=hadamard edge] (6) to (27.center);
		\draw [style=hadamard edge] (8) to (29.center);
		\draw [style=hadamard edge] (5) to (30.center);
		\draw [style=hadamard edge] (4) to (31.center);
		\draw [style=hadamard edge] (4) to (32.center);
		\draw [style=brace edge] (34.center) to (35.center);
		\draw [style=hadamard edge] (6) to (38.center);
		\draw [style=hadamard edge] (7) to (39.center);
		\draw [style=hadamard edge] (3) to (40.center);
		\draw [style=hadamard edge] (5) to (41.center);
		\draw [style=hadamard edge] (8) to (28.center);
		\draw (59.center) to (56);
		\draw (57) to (60.center);
		\draw (61.center) to (58);
		\draw (49.center) to (52);
		\draw (48.center) to (51);
		\draw (47.center) to (50);
		\draw [style=hadamard edge] (53) to (69.center);
		\draw [style=hadamard edge] (53) to (68.center);
		\draw [style=hadamard edge] (53) to (67.center);
		\draw [style=hadamard edge] (54) to (67.center);
		\draw [style=hadamard edge] (54) to (68.center);
		\draw [style=hadamard edge] (55) to (66.center);
		\draw [style=hadamard edge] (58) to (65.center);
		\draw [style=hadamard edge] (58) to (63.center);
		\draw [style=hadamard edge] (57) to (64.center);
		\draw [style=hadamard edge] (57) to (63.center);
		\draw [style=hadamard edge] (56) to (62.center);
		\draw [style=hadamard edge] (56) to (63.center);
		\draw [style=hadamard edge] (53) to (70.center);
		\draw [style=hadamard edge] (55) to (72.center);
		\draw [style=hadamard edge] (52) to (73.center);
		\draw [style=hadamard edge] (51) to (74.center);
		\draw [style=hadamard edge] (51) to (75.center);
		\draw [style=brace edge] (77.center) to (78.center);
		\draw [style=hadamard edge] (53) to (80.center);
		\draw [style=hadamard edge] (54) to (81.center);
		\draw [style=hadamard edge] (50) to (82.center);
		\draw [style=hadamard edge] (52) to (83.center);
		\draw [style=hadamard edge] (55) to (71.center);
		\draw (86) to (87);
	\end{pgfonlayer}
\end{tikzpicture}}
  \end{equation*}
  where $M$ describes the biadjacency matrix of the relevant vertices, and $M^\prime$ is the matrix produced by starting with $M$ and then adding row 1 to row 2, taking sums modulo 2. Furthermore, if the diagram on the LHS has a focused gFlow, then so does the RHS.
\end{proposition}


We now describe our extraction procedure at a high level. To aid in understanding this procedure, we provide a fully worked-out example in Appendix~\ref{sec:example-derivation} and a pseudocode description in Appendix~\ref{sec:algorithms}. To start, let the frontier $F$ consist of all the output vertices. If a vertex of the frontier has a non-zero phase, unfuse this phase in the direction of the output. Whenever two vertices of the frontier are connected, unfuse this connection into a CZ gate. If the diagram started out having a gFlow, then the resulting diagram still has a gFlow since we only affected output vertices.
Let $g$ be the gFlow of the diagram. We push the frontier onto new vertices in the following way:
\begin{enumerate}
  \item Find an unextracted vertex $v$ such that $g(v)\subseteq F$. Such a vertex always exists: take for instance any non-frontier vertex that is maximal in the gFlow order.
  \item We know that Odd$(g(v)) = \{v\}$ hence we can pick a designated frontier vertex $w\in g(v)$ such that $w$ is connected to $v$.
  \item Write the biadjency matrix of $g(v)\subseteq F$ and their neighbours to the left of the frontier (`the past'). Because Odd$(g(v)) = \{v\}$ we know that the sum of all the rows in this matrix modulo 2 is an all-zero vector except for a 1 corresponding to $v$. Hence, if we apply row operations with $w$ as the target from every other vertex in $g(v)$ we know that $w$ will have only a single vertex connected to it in its past, namely $v$. Do these row operations using Proposition~\ref{prop:cnotgflow}. As a result, \CNOT gates will be pushed onto the extracted circuit. By that same proposition, the diagram will still have a gFlow.
  \item At this point it might be that $v$ is still connected to vertices of $F$ that are not $w$. For each of these vertices $w^\prime \in F$, do a row operation from $w$ to $w^\prime$ by adding the corresponding \CNOT to the extracted diagram in order to disconnect $w^\prime$ from $v$. At this point $v$'s only connection to $F$ is via $w$. We can therefore conclude that the gFlow of the resulting diagram is such that the only vertex $u$ that has $w\in g(u)$ is $v$.
  \item The vertex $w$ has arity 2 and since it is phaseless we can remove it from the diagram by~\IdentityRule. Remove $w$ from the frontier, and replace it by $v$. By the remarks in the previous point the resulting diagram still has a gFlow. Unfuse the phase of $v$ as a phase gate to the extracted part.
  \item If there are still unextracted vertices, go back to step 1, otherwise we are done.
\end{enumerate}
After all vertices have been extracted by the above procedure, we will have some permutation of the frontier vertices to the inputs. The extraction is finished by realising this permutation as a series of SWAP gates.


\begin{remark}
	An interesting thing to note is we do not actually need a pre-calculated focused gFlow for this algorithm: knowing that one exists is enough. In step 3, it suffices to apply primitive row operations until a row with a single 1 appears in the biadjacency matrix. If this is possible by \textit{some} set of row operations (i.e. if there exists a focused gFlow), such a row will always appear in the reduced echelon form, so we can simply do Gauss-Jordan elimination. This is how we describe this step in the pseudocode presented in Appendix~\ref{sec:algorithms}.
\end{remark}

\begin{remark}\label{rem:extract-performance}
    The cost of the implementation of the extraction algorithm is dominated by the Gauss-Jordan elimination steps. The amount of rows in the matrices involved is equal to the qubit count $q$, while the amount of columns is equal to the amount of neighbours of the frontier, and is hence upper-bounded by the amount of spiders in the diagram $n$. As in the worst case $n$ elimination steps are needed and each elimination step takes $O(q^2n)$, this brings the complexity to $O(q^2n^2)$. As with simplification, we see much more favourable scaling in practice, due to sparsity of the graphs involved. For circuit extraction, our implementation gives similar performance to that described in Remark~\ref{rem:simp-performance}.
\end{remark}

Picking up our example from the previous section, we can apply this procedure to the skeleton of the circuit in equation~\eqref{eq:big-example-skeleton} and we obtain a new, extracted circuit:
\[
\scalebox{1.0}{\begin{tikzpicture}
	\begin{pgfonlayer}{nodelayer}
		\node [style=none] (0) at (-13.5, 1.5) {};
		\node [style=none] (1) at (-13.5, 0.5) {};
		\node [style=none] (2) at (-13.5, -0.5) {};
		\node [style=none] (3) at (-13.5, -1.5) {};
		\node [style=Z phase dot] (5) at (-11.75, -0.5) {$\frac{\pi}{2}$};
		\node [style=Z phase dot] (7) at (-11.75, -1.5) {$\frac{\pi}{2}$};
		\node [style=Z dot] (8) at (-10.75, -1.5) {};
		\node [style=Z dot] (9) at (-10.75, 0.5) {};
		\node [style=Z dot] (10) at (-9.75, -1.5) {};
		\node [style=Z dot] (11) at (-9.75, -0.5) {};
		\node [style=Z dot] (12) at (-9.75, 1.5) {};
		\node [style=Z dot] (13) at (-9.75, 0.5) {};
		\node [style=Z dot] (14) at (-8.25, 1.5) {};
		\node [style=Z dot] (15) at (-8.25, -0.5) {};
		\node [style=Z phase dot] (17) at (-7.25, 0.5) {$\frac{\pi}{4}$};
		\node [style=Z dot] (18) at (-6.25, 0.5) {};
		\node [style=Z dot] (19) at (-6.25, -0.5) {};
		\node [style=Z dot] (20) at (-5.25, -1.5) {};
		\node [style=Z dot] (21) at (-5.25, 0.5) {};
		\node [style=Z dot] (22) at (-4.25, 1.5) {};
		\node [style=Z dot] (23) at (-4.25, 0.5) {};
		\node [style=Z phase dot] (25) at (-3.25, -1.5) {$\frac{\pi}{4}$};
		\node [style=Z dot] (26) at (-2.25, -1.5) {};
		\node [style=Z dot] (27) at (-2.25, 0.5) {};
		\node [style=X dot] (28) at (-1, -1.5) {};
		\node [style=Z dot] (29) at (-1, -0.5) {};
		\node [style=Z phase dot] (31) at (1, -0.5) {$\pi$};
		\node [style=X dot] (32) at (-1, 0.5) {};
		\node [style=Z dot] (33) at (-1, 1.5) {};
		\node [style=Z phase dot] (35) at (1, 1.5) {$\frac{5\pi}{4}$};
		\node [style=Z dot] (36) at (2, -1.5) {};
		\node [style=Z dot] (37) at (2, 1.5) {};
		\node [style=Z phase dot] (39) at (3.5, -1.5) {$\frac{\pi}{4}$};
		\node [style=Z dot] (40) at (4.5, -0.5) {};
		\node [style=Z dot] (41) at (4.5, -1.5) {};
		\node [style=Z dot] (43) at (6, -1.5) {};
		\node [style=Z dot] (44) at (6, 1.5) {};
		\node [style=Z dot] (46) at (7.5, 1.5) {};
		\node [style=Z dot] (47) at (7.5, 0.5) {};
		\node [style=Z phase dot] (49) at (9.5, 0.5) {$\pi$};
		\node [style=Z dot] (50) at (10.5, 0.5) {};
		\node [style=Z dot] (51) at (10.5, 1.5) {};
		\node [style=X dot] (54) at (11.5, 0.5) {};
		\node [style=Z dot] (55) at (11.5, 1.5) {};
		\node [style=X dot] (56) at (12.5, 1.5) {};
		\node [style=Z dot] (57) at (12.5, 0.5) {};
		\node [style=X dot] (58) at (13.5, 0.5) {};
		\node [style=Z dot] (59) at (13.5, 1.5) {};
		\node [style=X dot] (60) at (11.5, -1.5) {};
		\node [style=Z dot] (61) at (11.5, -0.5) {};
		\node [style=X dot] (62) at (12.5, -0.5) {};
		\node [style=Z dot] (63) at (12.5, -1.5) {};
		\node [style=X dot] (64) at (13.5, -1.5) {};
		\node [style=Z dot] (65) at (13.5, -0.5) {};
		\node [style=none] (66) at (14.5, 1.5) {};
		\node [style=none] (67) at (14.5, 0.5) {};
		\node [style=none] (68) at (14.5, -0.5) {};
		\node [style=none] (69) at (14.5, -1.5) {};
		\node [style=hadamard] (70) at (-12.75, -0.5) {};
		\node [style=hadamard] (71) at (-12.75, -1.5) {};
		\node [style=hadamard] (72) at (0, -0.5) {};
		\node [style=hadamard] (73) at (7.5, -0.5) {};
		\node [style=hadamard] (74) at (7.5, -1.5) {};
		\node [style=hadamard] (75) at (8.5, 0.5) {};
		\node [style=hadamard] (76) at (-4.25, -1.5) {};
		\node [style=hadamard] (77) at (5.25, -1.5) {};
		\node [style=hadamard] (78) at (6.75, 1.5) {};
		\node [style=hadamard] (79) at (0, 1.5) {};
		\node [style=hadamard] (80) at (-9, 0.5) {};
		\node [style=hadamard] (81) at (2.75, -1.5) {};
	\end{pgfonlayer}
	\begin{pgfonlayer}{edgelayer}
		\draw (0.center) to (12);
		\draw (1.center) to (9);
		\draw (2.center) to (5);
		\draw (3.center) to (7);
		\draw (5) to (11);
		\draw (7) to (8);
		\draw [style=hadamard edge] (8) to (9);
		\draw (8) to (10);
		\draw (9) to (13);
		\draw [style=hadamard edge] (10) to (11);
		\draw (10) to (20);
		\draw (11) to (15);
		\draw [style=hadamard edge] (12) to (13);
		\draw (12) to (14);
		\draw (13) to (17);
		\draw (14) to (22);
		\draw [style=hadamard edge] (14) to (15);
		\draw (15) to (19);
		\draw (17) to (18);
		\draw [style=hadamard edge] (18) to (19);
		\draw (18) to (21);
		\draw (19) to (29);
		\draw (20) to (25);
		\draw [style=hadamard edge] (20) to (21);
		\draw (21) to (23);
		\draw (22) to (33);
		\draw [style=hadamard edge] (22) to (23);
		\draw (23) to (27);
		\draw (25) to (26);
		\draw [style=hadamard edge] (26) to (27);
		\draw (26) to (28);
		\draw (27) to (32);
		\draw (28) to (36);
		\draw (28) to (29);
		\draw (29) to (31);
		\draw (31) to (40);
		\draw (32) to (33);
		\draw (32) to (47);
		\draw (33) to (35);
		\draw (35) to (37);
		\draw [style=hadamard edge] (36) to (37);
		\draw (36) to (39);
		\draw (37) to (44);
		\draw (39) to (41);
		\draw [style=hadamard edge] (40) to (41);
		\draw (40) to (61);
		\draw (41) to (43);
		\draw [style=hadamard edge] (43) to (44);
		\draw (43) to (60);
		\draw (44) to (46);
		\draw (46) to (51);
		\draw [style=hadamard edge] (46) to (47);
		\draw (47) to (49);
		\draw (49) to (50);
		\draw [style=hadamard edge] (50) to (51);
		\draw (50) to (54);
		\draw (51) to (55);
		\draw (54) to (57);
		\draw (54) to (55);
		\draw (55) to (56);
		\draw (56) to (57);
		\draw (56) to (59);
		\draw (57) to (58);
		\draw (58) to (67.center);
		\draw (58) to (59);
		\draw (59) to (66.center);
		\draw (60) to (61);
		\draw (60) to (63);
		\draw (61) to (62);
		\draw (62) to (65);
		\draw (62) to (63);
		\draw (63) to (64);
		\draw (64) to (65);
		\draw (64) to (69.center);
		\draw (65) to (68.center);
	\end{pgfonlayer}
\end{tikzpicture}}
\]
The circuit has been reduced from 195 gates to 41 gates, going via an intermediate representation as a \zxdiagram with 12 spiders. As mentioned in the introduction, this example is available as a Jupyter notebook `\href{https://nbviewer.jupyter.org/github/Quantomatic/pyzx/blob/906f6db3/demos/example-gtsimp.ipynb}{\color{blue!80!black} \texttt{demos/example-gtsimp.ipynb}}' in the PyZX~\cite{pyzx} repository on GitHub.


\begin{figure}
    \centering
    \includegraphics[width=0.45\textwidth]{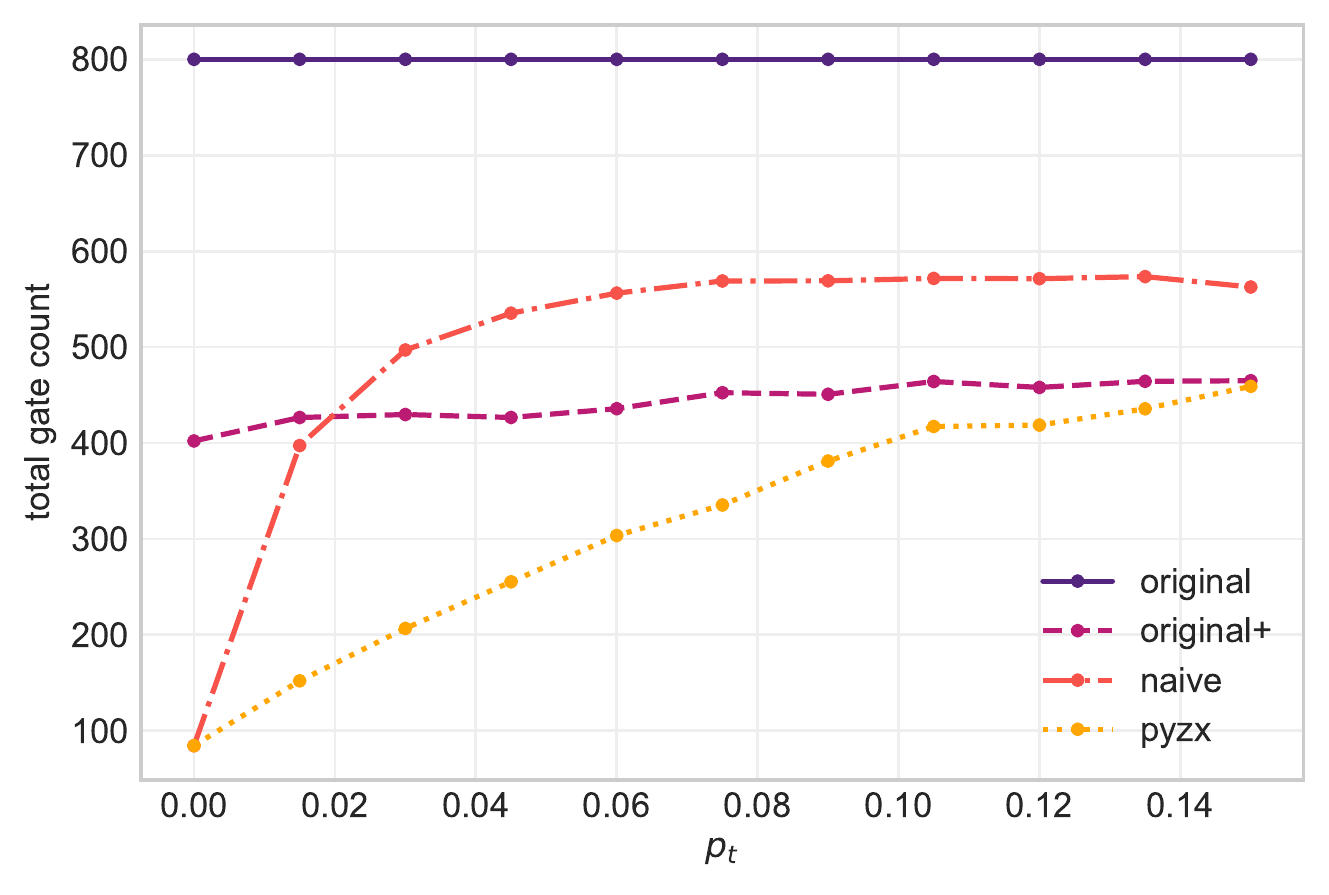}\ \ 
    \includegraphics[width=0.45\textwidth]{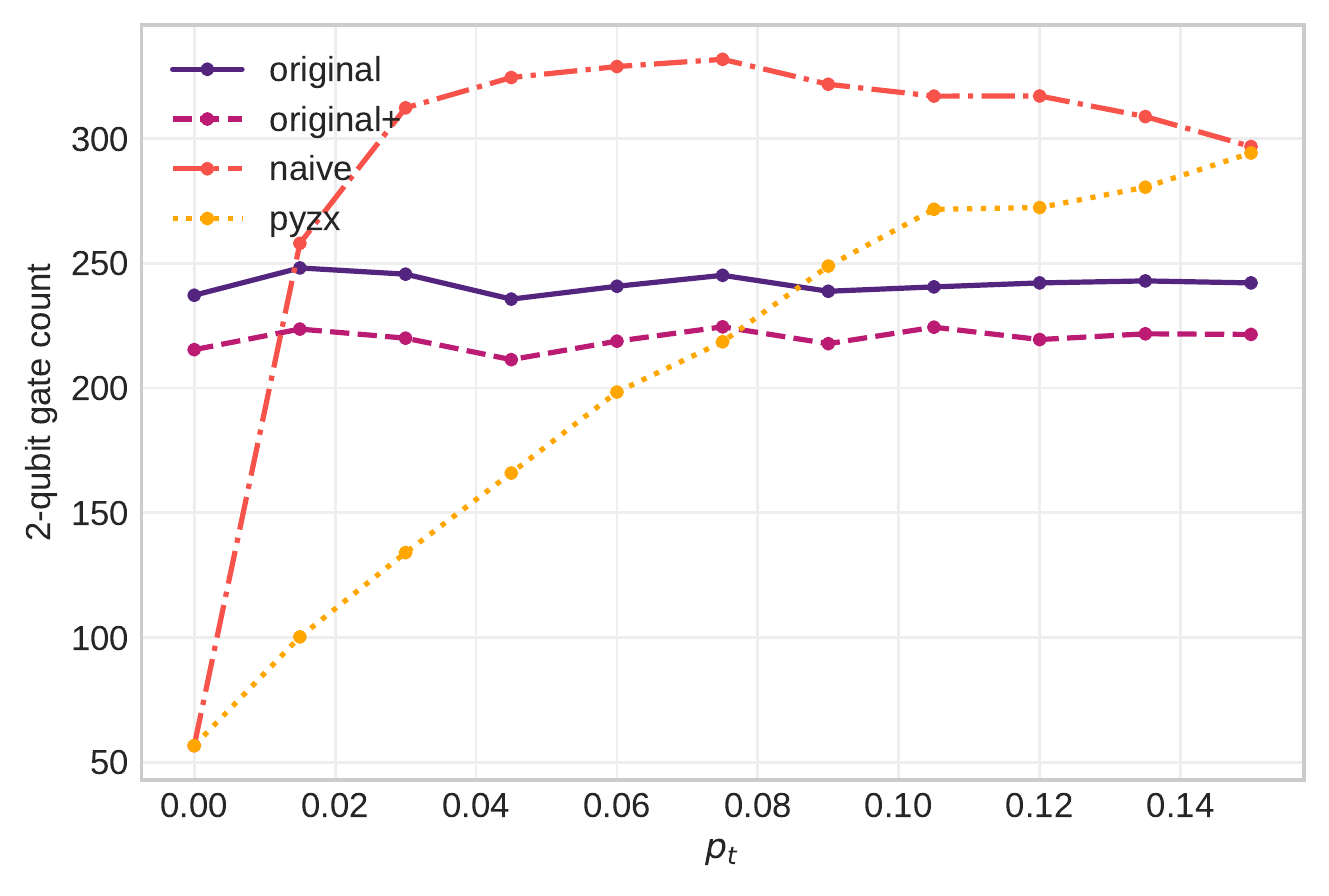}
    \caption{Comparison of different circuit optimization methods for random 8-qubit circuits with a varying proportion $p_t$ of $T$ gates. `original+' applies simple peephole optimisations (e.g. gate cancellations), `na\"ive' reduces each block of Clifford circuits to GS-LC form and re-synthesizes it, whereas `pyzx' is our full optimise-and-extract procedure.}\label{fig:experiment-gateopt}
\end{figure}

In Figure~\ref{fig:experiment-gateopt}, we give an empirical comparison between the full optimise-and-extract technique and the na\"ive approach of optimising all of the Clifford sub-circuits of a general Clifford+T circuit. This data was produced by generating random 8 qubit Clifford+T circuits with 800 gates where the probability that each gate is a CNOT gate is $0.3$ and we vary the probability $p_t$ of $T$ gates from $0$ to $0.15$. For each value we generate 20 circuits and we report the average total and 2-qubit gate counts for the original circuit (labelled \textsf{original} in Figure~\ref{fig:experiment-gateopt}) and 3 different reduced versions.

The first reduced version (\textsf{original+}) is the original circuit after some basic post-processing. This post-processing tries to cancel and combine as many adjacent gates as possible. It does this by applying a forward pass on the circuit, where simple 2-gate commutations such as $\textrm{CZ}(1 \otimes H) = (1 \otimes H)\textrm{CNOT}$ are used to delay the placement of Hadamard gates in the circuit as long as possible, while at the same time keeping track of a stack of commuting gates behind the delayed Hadamard gates that are available for combination with new incoming gates. After a pass, the circuit is reversed and the process is repeated until no more gate reductions occur. This post-processing is intended to eliminate obvious redundancies, rather than providing significant gate-based optimisations such as those in e.g.~\cite{nam2018automated}.

The interesting two cases in Figure~\ref{fig:experiment-gateopt} are \textsf{na\"ive} and \textsf{pyzx}. In the na\"ive case, we cut the circuit into alternating layers of Clifford circuits and T-gates. We then apply the Clifford normalisation and extraction described in Section~\ref{sec:clifford} to each Clifford chunk. For \textsf{pyzx}, the full Clifford+T simplification and extraction described in Section~\ref{sec:cliffordT} is applied. In both cases, the resulting circuit is post-processed as in \textsf{original+}. For the steps requiring Gaussian elimination, we use the asymptotically optimal algorithm proposed in Ref.~\cite{markov2008optimal}.

When the circuits are very close to Clifford, both of the second two optimisations perform very well, and the \textsf{pyzx} method outperforms the na\"ive method in every case. For the na\"ive method, as the probability of T gates increases, each Clifford chunk will have far fewer than quadratically many 2-qubit gates, in which case 
resynthesizing is likely actually increase the gate count. At high T gate probability, Figure~\ref{fig:experiment-gateopt} shows that the gate and 2-qubit gate count of the na\"ive method saturate. This is because the Clifford chunks are so small that no Gaussian elimination is used at all.
The \textsf{pyzx} case performs much better than the \textsf{na\"ive} case, as it can use more `non-local' structure that looks beyond the boundary of the Clifford subcircuits. Nevertheless, as the non-Clifford density increases, the \textsf{pyzx} routine also begins to increase the gate count. This is because we are still re-synthesising parts of the circuit using the procedure from Ref.~\cite{markov2008optimal}. While this is \textit{asymptotically} optimal, it is less beneficial for very small pieces of the circuit, and can even increase the gate count.


There are two general ways in which this method can be improved: better simplification, and better extraction. The simplification can be improved by including a larger selection of simplification steps. This is done in Ref.~\cite{zxtcount}, and allows the size of the graph to be reduced significantly in many cases. The circuit extraction can be improved in several ways. For instance, when doing Gaussian elimination, the order of the rows in the matrix is immaterial as it corresponds to an arbitrary labelling of interior spiders. A suitable heuristic, such as for instance the genetic algorithm of Ref.~\cite{kissinger2019cnot}, for choosing an order of columns can greatly improve the performance of the Gaussian elimination.

These experiments are available as a Jupyter Notebook in the PyZX repository\footnote{\href{https://nbviewer.jupyter.org/github/Quantomatic/pyzx/blob/671da79/demos/Optimising\%20almost-Clifford\%20circuits.ipynb}{\color{blue!80!black}\texttt{nbviewer.jupyter.org/github/Quantomatic/pyzx/blob/671da79/demos/Optimising almost-Clifford circuits.ipynb}}}.

\section{Conclusions and Future Work}\label{sec:conclusion}
We have introduced a terminating rewrite strategy for \zxdiagrams that simplifies Clifford circuits to GS-LC form using the graph-theoretic notions of local complementation and pivoting. We have shown how the diagrams produced by the rewrite strategy applied to non-Clifford circuits can be turned back into a circuit using the fact that our rewrites preserve focused gFlow. 

The question of how a circuit can be extracted from a general ZX-diagram is still open. We speculate that this problem is not tractable in the general case, as it seems to be related to the problem of finding an ancilla-free unitary from a circuit containing ancillae. This observation does however not preclude finding larger classes of rewrites and ZX-diagrams such that we can efficiently extract a circuit. 

The next step for optimising quantum circuits using \zxdiagrams is to find additional rewrite rules that allow one to produce smaller diagrams, while still being able to extract the diagrams into a circuit. As mentioned in the introduction, two of the authors have recently proposed an approach for T-count reduction based on the methods described here which has been very successful in the ancilla-free case~\cite{zxtcount}. As mentioned in Ref.~\cite{zxtcount}, even more dramatic simplifications of \zxdiagrams are possible, but re-extracting a circuit using gFlow becomes problematic. Developing more general circuit extraction procedures, possibly with the help of ancillae and classical control is therefore a topic of ongoing research.


A second avenue of future work is to adapt the circuit extraction procedure to work in constrained qubit topologies. Many near-term quantum devices only allow interactions between certain pairs of qubits (e.g. nearest neighbours in some fixed planar graph). A notable feature of the circuit extraction procedure described in Section~\ref{sec:cliffordT} is that we have some freedom in choosing which CNOTs to apply in reducing the bi-adjacency matrix of the ZX-diagram. Indeed it effectively amounts to performing Gaussian elimination using only a constrained set of primitive row operations. In Ref.~\cite{kissinger2019cnot} (and independently in \cite{nash2019quantum}) a strategy based on Steiner trees has been proposed for doing exactly that for CNOT or CNOT+Phase circuits. In principle, these methods are directly applicable to the extraction procedure from Section~\ref{sec:cliffordT}. However, unlike the simpler cases considered by \cite{kissinger2019cnot} and \cite{nash2019quantum}, our extraction procedure relies on many rounds of Gaussian elimination, so it will likely be necessary to use some sort of lookahead to minimise global overhead.



\medskip

\noindent {\small \textbf{Acknowledgements.} 
SP acknowledges support from the projects ANR-17-CE25-0009 SoftQPro, ANR-17-CE24-0035 VanQuTe, PIA-GDN/Quantex, and  LUE / UOQ. AK and JvdW are supported in part by AFOSR grant FA2386-18-1-4028. We would like to thank Quanlong Wang and Niel de Beaudrap for useful conversations about circuit optimisation and extraction with the \zxcalculus.}

\bibliographystyle{plainnat}
\bibliography{main}

\appendix

\section{Example derivation}
\label{sec:example-derivation}

In this section we will work through a complete example of our simplification algorithm.








\noindent Let us start with the following circuit:
\ctikzfig{compl-example1}
This has 5 two-qubit gates and 19 single-qubit gates. Using the procedure of Lemma~\ref{lem:all-zx-are-graph-like} to get it in graph-like form we get:
\ctikzfig{compl-example2}

We see that we have multiple interior Pauli and proper Clifford phases and hence we must do local complementations and pivots. We start by doing a sequence of local complementations. 
The locations of these operations are marked by a red star.
\ctikzfig{compl-example3}
We now have two internal Pauli vertices left, but they are not connected. Hence we must use equation \eqref{eq:pivot-unfuse}, to be able to get rid of them:
\ctikzfig{compl-example4}
Applying the pivots to the designated vertices gives us:
\ctikzfig{compl-example5}
Note that in the bottom qubit we end up where we started before doing the pivots, but now we have marked the $\pi/2$ vertex as a local Clifford, and hence the $\pi$ vertex now counts as a boundary vertex. This means the simplification procedure has ended. The final step is extracting a circuit:
\ctikzfig{compl-example6a}
Here the frontier is marked by the box. We first unfuse the phases onto gates. Then we want to extract new vertices. These are marked by a red star. Since the three marked vertices have no other connections to the frontier, and the frontier vertices have only these single neighbours, we can extract these by simply putting them onto the frontier. The next step is to extract the bottomleft vertex, but in this case there is a connection too many which we must extract as a \CNOT:
\ctikzfig{compl-example6b}
We now simply repeat this procedure until the entire circuit is extracted.
\ctikzfig{compl-example6c}

The extracted circuit has 4 two-qubit gates and 18 single-qubit gates. Applying some final gate simplifications on the inputs and outputs gets the number of single-qubit gates down to 11:
\ctikzfig{compl-example7b}
\vspace{0.3cm}





\section{Proofs}
\label{sec:proofs}

\subsection{Circuits and focused gFlow}
\label{sec:circuits-causal-flow}


To show that graph-like \zxdiagrams arising from circuits admit a focused gFlow, we introduce a simpler sufficient condition.

\begin{definition}{\cite{Danos2006Determinism-in-}}\label{def:causal-flow}
Given an open graph $G$, a
\emph{causal flow} $(f,\prec)$ on $G$  consists of a
function $f: \comp O \to \comp I$ and a partial order $\prec$ on the set
$V$ satisfying the following properties:
\begin{enumerate}
\item $f(v) \sim v$ \label{flowi}
\item $v \prec f(v)$ \label{flowii}
\item if $u \sim f(v)$ then $v \prec u$ \label{flowiii}
\end{enumerate}
where $v \sim u$ means that $u$ and $v$ are adjacent in the graph.
\end{definition}

Like focused gFlow, causal flow is a special case of a more general property called gFlow:

\begin{definition}
Given an open graph $G$, a \textit{generalised flow}, or gFlow consists of a pair $(g : \comp O \to 2^{\comp I},\prec)$ such that for any $u \in \comp O$:
\begin{enumerate}
\item $u \in \odd{G}{g(u)}$
\item if $v \in g(u)$, then $u \prec v$
\item if $v \in \odd{G}{g(u)}$ and $v \neq u$, then $u \prec v$
\end{enumerate}
\end{definition}

In the case where all of the sets $g(u)$ are singletons, the conditions above are equivalent to the causal flow conditions. Hence, if a graph admits a causal flow, it also admits a gFlow by letting $g(u) := \{ f(u) \}$. It was furthermore shown in \cite{mhalla2011graph} that any gFlow can be systematically transformed into a focused gFlow. Combining these two facts, we obtain the following theorem.

\begin{theorem}\label{thm:causal-implies-focused}
If an open graph admits a causal flow $(f,\prec)$, then it also admits a focused gFlow.
\end{theorem}

The following proof of Lemma~\ref{lem:circuits-have-gflow} shows that the diagrams produced by making a circuit graph-like according to the procedure outlined in Lemma~\ref{lem:all-zx-are-graph-like} have causal flow and hence a focused gFlow.

\begin{proof}[Proof of Lemma \ref{lem:circuits-have-gflow}]
  Let $D$ denote the circuit, and let $D^\prime$ denote the diagram produced 
  by applying the procedure of Lemma~\ref{lem:all-zx-are-graph-like} to $D$.
  
  In order to prove that $D^\prime$ has a focused gFlow, it suffices to show that it admits a causal flow, by Theorem~\ref{thm:causal-implies-focused}.

  With every spider $v$ of the circuit $D$ we associate a number $q_v$ specifying on which `qubit-line' it appears.
  We also associate a `row-number' $r_v\geq 0$ specifying how `deep' in the circuit it appears.
  Suppose now that $v\sim w$ in $D$. If they are on the same qubit, so $q_v=q_w$, then necessarily $r_v\neq r_w$.
  Conversely, if they are on different qubits, $q_v\neq q_w$, then they must be part of a CZ or CNOT gate,
  and hence $r_v=r_w$.

  In the diagram resulting from Lemma~\ref{lem:all-zx-are-graph-like}, every spider arises from fusing together adjacent spiders on the same qubit line from the original diagram.
  For a spider $v$ in $D^\prime$ we can thus associate two numbers $s_v$, and $t_v$, 
  where $s_v$ is the lowest row-number of a spider fused into $v$, and $t_v$ is the highest.
  Spider fusion from $D$ only happens for spiders on the same qubit-line, and hence $v$ also inherits a $q_v$ from
  all the spiders that got fused into it. Any two spiders $v$ and $w$ in $D^\prime$ with $q_v=q_w$
  must have been produced by fusing together adjacent spiders on the same qubit-line,
  and hence we must have $t_v<s_w$ or $t_w<s_v$, depending on which of the spiders is most to the left.
  If instead $v\sim w$ and $q_v\neq q_w$, then their connection must have arisen from some CNOT or CZ gate in $D$, 
  and hence the intervals $[s_v,t_v]$ and $[s_w,t_w]$ must have overlap,
  so that necessarily $s_w\leq t_v$ and $s_v\leq t_w$.


  Now we let $O$ be the set of spiders connected to an output, 
  and $I$ the set of spiders connected to an input in $D^\prime$.
  For all $v\in \comp O$ we let $f(v)$ be the unique connected spider to the right of $v$ on the same qubit-line.
  We define the partial order as follows: $v\prec w$ if and only if $v=w$ or $t_v<t_w$. 
  It is straightforward to check that this is indeed a partial order.

  By construction $f(v)\sim v$ and the property $v \prec f(v)$, follows from $t_v < s_{f(v)}$ discussed above,
  so let us look at the third property of causal flow.

  Suppose $w\sim f(v)$. We need to show that $v\prec w$. If $v=w$ this is trivial so suppose $v\neq w$.
  First suppose that $q_w = q_{f(v)}$ (which is also equal to $q_v$). $f(v)$ has a maximum of two neighbours
  on the same qubit-line, and since one of them is $v$, this can only be if $w=f(f(v))$ and hence $v\prec f(v)\prec w$,
  and we are done.
  So suppose that $q_w\neq q_{f(v)}$. 
  By the discussion above, we must then have $s_w\leq t_{f(v)}$ and $s_{f(v)} \leq t_w$.
  Since we also have $t_v < s_{f(v)}$ we get $t_v < s_{f(v)} \leq t_w$ so that indeed $v\prec w$.
\end{proof}

\subsection{Preservation of focused gFlow}
\label{sec:pres-focus-gflow}

Throughout this section, we will rely extensively on the symmetric difference of sets: ${A\symd B} := (A\cup B)\setminus (A\cap B)$. Note that $\symd$ is associative and commutative, so it extends to an $n$-ary operation in the obvious way. For $I:={1,\ldots,n}$ we have:
\[\Symdi{i \in I} A_i := A_1 \symd A_2 \symd \ldots \symd A_n \]
In particular, we have $a \in \Symdi{i \in I} A_i$ if and only if $a$ appears in an odd number of sets $A_i$. By convention, we assume $\Symdi{...}$ binds as far to the right as possible, i.e.
\[ \left( \Symdi{i \in I} A_i \symd B \right) := \left( \Symdi{i \in I} (A_i \symd B) \right) \]

\subsubsection{Local complementation}
\label{sec:proofLC}

The following lemma will be needed in the other proofs, it shows how the odd neighbourhood of a set evolves under local complementation:
\begin{lemma}\label{lem:oddneighbours}
Given a graph $G=(V,E)$, $A\subseteq V$ and $u \in V$, $$\odd {G\star u} A=\begin{cases}\odd G A \symd (N_G(u)\cap A) & \text{if $u\notin \odd G A$}\\ \odd G {A} \symd (N_G(u)\setminus A)&\text{if $u\in \odd G A$}\end{cases}$$
\end{lemma}
\begin{proof}
First notice that $\odd G .$ is linear: 
$\forall A, B$, $\odd G {A\symd B} = \odd G A \symd \odd G B$. 
Moreover $\forall v, \odd G {\{v\}} = N_G(v)$, the neighbourhood of $v$ in $G$. 
As a consequence, $\forall A, \odd G A  =\Symdi{v\in A} N_G (v)$. 

The local complementation is acting as follows on the neighbourhoods: 
$$\forall v, N_{G\star u} (v) = \begin{cases} N_G(v) \symd N_G(u)\symd \{v\} &\text{if $v\in N_G(u)$}\\ N_G(v) &\text{otherwise} \end{cases}$$
As a consequence,
\begin{eqnarray*}
  \odd {G\star u} A&=&\Symdi{v\in A} N_{G\star u} (v)\\
  &=& \left(\Symdi{v\in A\cap N_G(u)} N_{G} (v)\symd N_G(u)\symd \{v\}\right) \symd  \left(\Symdi{v\in A\setminus  N_G(u)} N_{G} (v)\right)\\
  &=&\left(\Symdi{v\in A} N_{G} (v)\right) \symd \left(\Symdi{v\in A\cap N_G(u)}  N_G(u) \right) \symd  \left(\Symdi{v\in A\cap N_G(u)}  \{v\}\right)\\
  &=&\odd G A \symd \left(\Symdi{v\in A\cap N_G(u)}  N_G(u) \right) \symd (A\cap N_G(u))
\end{eqnarray*}
Notice that $|A\cap N_G(u)| \equiv 1 \bmod 2$ iff $u\in \odd G A$. Hence, if $u\notin \odd G A$, 
$
   \odd {G\star u} A  =\odd G A  \symd (A\cap N_G(u))
$.  Otherwise, if $u \in \odd G A$, 
$    \odd {G\star u} A  =\odd G A \symd N_G(u)\symd (A\cap N_G(u)) =  \odd G A\symd (N_G(u)\setminus A)
$.%
\end{proof}

\begin{lemma}\label{lem:gflow-lcomp} 
If $(g,\prec)$ is a focused gFlow for $(G,I,O)$ then $(g',\prec)$ is a focused gFlow for $(G\star u\setminus \{u\},I,O)$ where $g':\comp O\setminus \{u\} \to 2^{\comp I\setminus \{u\}}$ is recursively defined as 
\[
g'(w):=
\begin{cases}
  g(w) \symd \left( \Symdi{t\in R_w} g'(t) \right) & \text{if $u \notin  g(w)$}\\
  g(w) \symd \{u\}\symd g(u) \symd \left(\Symdi{t\in R_u\symd R_w} g'(t) \right) & \text{if $u \in g(w)$}
\end{cases}
\]
where $R_w:=N_G(u)\cap g(w)\cap \comp O$.
\end{lemma}

\begin{proof}

 
First of all, $g'$ can indeed be defined inductively by starting at the maximal elements in the partial order, since for any $t\in R_w$, $t\in g(w)$ so that $w\prec t$, and similarly when $u\in g(w)$, $w\prec u$ so $t\in R_u$ implies $u\prec t$ hence $w\prec t$. That $g'$ preserves the order, i.e.\ $v\in g'(w)\implies w\prec v$, follows similarly. Moreover, it is easy to show that $\forall w\in \comp O\setminus \{u\}$, $g'(w)\subseteq \comp I\setminus \{u\}$. 

It remains to show that $\forall w\in \comp O\backslash \{u\}, \odd {G\star u \backslash u} {g'(w)} \cap \comp O = \{w\}$, or equivalently, letting $O':=O\cup \{u\}$, $\forall w\in \comp{O'}, \odd {G\star u} {g'(w)} \cap \comp{O'} = \{w\}$.

First notice that  for any $w\in \comp{O'}$, $u\notin \odd G{g(w)}$ (because $\odd G{g(w)}\subseteq \{w\} \cup \comp{O}$), and hence using Lemma \ref{lem:oddneighbours}, the linearity of $\textup{Odd}(.)$ 
and the distributivity of $\cap$ over $\symd$, 
\begin{eqnarray*} \odd {G\star u}{g(w)} \cap \comp{O'}&=& \odd {G}{g(w)}\cap \comp{O'}\symd (N_G(u)\cap g(w) \cap \comp{O'})\\
&=&\{w\}\symd R_w\\
\end{eqnarray*}
Moreover, $u\in \odd G{g(u)}$, so
\begin{eqnarray*}
\odd {G\star u}{ \{u\}\symd g(u)}\cap \comp{O'} &=&(N_{G\star u}(u)\symd \odd {G\star u}{ g(u)})\cap \comp{O'}\\
\text{(Lemma \ref{lem:oddneighbours})} &=& (N_G(u)\symd \odd G {g(u)} \symd (N_G(u)\setminus g(u)))\cap \comp{O'}\\
 &=& (\odd G {g(u)}\cap \comp{O'}) \symd (N_G(u)\cap g(u) \cap \comp{O'})\\
&=&N_G(u)\cap g(u) \cap \comp{O'}\\
&=&R_u
\end{eqnarray*}

As a consequence, again using distributivity of $\symd$ over the odd neighbourhood and the definition of $g'$:
$$\odd{G\star u}{g'(w)}\cap \comp{O'}=\begin{cases}\{w\}\symd R_w \symd \left(\Symdi{t\in R_w}\odd {G\star u}{g'(t)} \cap \comp{O'} \right) & \text{if $u\notin  g(w)$}\\ 
\{w\}\symd R_w\symd R_u \symd \left(\Symdi{t\in R_u\symd R_w} \odd {G\star u}{g'(t)} \cap \comp{O'} \right) &\text{if $u\in  g(w)$}\end{cases}$$


By induction on $w\in \comp{O'}$, starting at maximal elements in the order, we will now show that $\odd{G\star u}{g'(w)}\cap \comp{O'} =\{w\}$. 
Let $w_0\in \comp{O'}$ be maximal for $\prec$ i.e. $\forall w\in \comp{O'}$, $\neg (w_0\prec w)$. By maximality we must have $g(w_0)\subseteq O'$ and hence $R_{w_0} = \emptyset$. 
If $u\notin g(w_0)$, we then have $\odd{G\star u}{g'(w_0)}\cap \comp{O'} = \{w_0\}$ and we are done. Otherwise, if $u\in g(w_0)$, then $w_0\prec u$, so that $u$ is also maximal and hence $g(u)\subset O$ which also gives $R_u=\emptyset$. As a result $\odd{G\star u}{g'(w)}\cap \comp{O'}$ is also equal to $\{w_0\}$ in this case.

%
%
%
%
%
%

Now let $w\in \comp{O'}$ be arbitary. By induction we can now assume that for any $t\in R_w$ $\odd {G\star u}{g'(t)} \cap \comp{O'} = \{t\}$. Hence:
\begin{itemize}
\item If $u\notin g(w)$, 
\begin{eqnarray*}
\odd {G\star u}{g'(w)}\cap \comp{O'}& =& \{w\}\symd R_w \symd \left(\Symdi{t\in R_w}\odd {G\star u}{g'(t)} \cap \comp{O'} \right) \\
(IH)& = & \{w\}\symd R_w \symd \left( \Symdi{t\in R_w} \{t\} \right) \\
&=& \{w\}\symd R_w\symd R_w\\
&=&\{w\}
\end{eqnarray*} 

\item If $u\in g(w)$, then $w\prec u$ and hence also $\odd {G\star u}{g'(t)} \cap \comp{O'} = \{t\}$ for any $t\in R_u$. We calculate:
\begin{eqnarray*}
\odd {G\star u}{g'(w)}\cap \comp{O'}& =& \{w\}\symd R_w\symd R_u \symd \left(\Symdi{t\in R_u\symd R_w}\odd {G\star u}{g'(t)} \cap \comp{O'} \right) \\
(IH)& = & \{w\}\symd R_w\symd R_u \symd \left( \Symdi{t\in R_u\symd R_w} \{t\} \right) \\
&=& \{w\}\symd R_w\symd R_u\symd R_u\symd R_w\\
&=&\{w\}
\end{eqnarray*} 
\end{itemize}
\end{proof}

\begin{remark}
If $(g,\prec)$ is a focused gFlow for $(G,I,O)$ then $(g',\prec)$ is a gFlow for $(G\star u\setminus u,I,O)$ where $g':\comp O\setminus \{u\} \to 2^{\comp I\setminus \{u\}}$ is  defined as $$g'(w):=\begin{cases}g(w)& \text{if $u\notin  g(w)$}\\ g(w)\symd \{u\}\symd g(u)   &\text{otherwise}\end{cases}$$
Notice that $(g',\prec)$ is not necessarily a \emph{focused} gFlow. 
\end{remark}

\subsubsection{Pivoting}
\label{sec:proofPivot}

We will show in this section that pivoting preserves focused gFlow. For this, it will be useful to introduce some special sets. Let $N_G[u] := N_G(u)\symd \{u\}$ and $\codd G A := \odd  G A \symd A$ denote the \emph{closed neighbourhood} and the \emph{closed odd neighbourhood} respectively. Now, we prove a technical lemma about the action of pivoting on the odd-neighbourhood in terms of these sets.

\begin{lemma}\label{lem:oddneighbourspivot}
Given a graph $G=(V,E)$, $A\subseteq V$, $u \in V$, and $v\in N_G(u)$,
\begin{equation}\label{eq:odd-after-pivot}
\odd {G\wedge uv} A = 
\begin{cases}
  \odd G A & \text{if $u,v\notin \codd G A$}\\
  \odd G {A} \symd N_G[v]&\text{if $u\in \codd G A, v\notin \codd G A$}\\
  \odd G {A} \symd N_G[u]&\text{if $u\notin \codd G A, v\in \codd G A$}\\
  \odd G A \symd N_G[u]\symd N_G[v]& \text{if $u,v\in \codd G A$}
\end{cases}
\end{equation} 
\end{lemma}
\begin{proof}

Note that pivoting acts as follows on neighbourhoods: 
$$\forall w, N_{G\wedge uv} (w) = \begin{cases} 
N_G(w) \symd N_G[u]\symd N_G[v] &\text{if $w\in N_G[u]\cap N_G[v]$}\\
N_G(w) \symd N_G[v] &\text{if $w\in N_G[u]\setminus  N_G[v]$}\\
N_G(w) \symd N_G[u] &\text{if $w\in N_G[v]\setminus  N_G[u]$}\\
 N_G(w) &\text{otherwise} \end{cases}$$
As a consequence, 
\begin{eqnarray*}
  \odd {G\wedge uv} A&=&\Symdi{w\in A} N_{G\wedge uv} (w)\\
  &=& \left(\symdi{w\in A\cap N_G[u]\cap N_G[v]} N_{G\wedge uv} (w) \right) \symd
      \left(\symdi{w\in A\cap N_G[u]^c\cap N_G[v]^c} N_{G\wedge uv} (w) \right) \symd \\
&&    \left(\symdi{w\in A\cap N_G[u]^c\cap N_G[v]} N_{G\wedge uv} (w) \right) \symd
      \left(\symdi{w\in A\cap N_G[u]\cap N_G[v]^c} N_{G\wedge uv} (w) \right) \\
 &=& \left(\symdi{w\in A\cap N_G[u]\cap N_G[v]} N_{G} (w) \symd N_G[u]\symd N_G[v] \right) \symd
    \left(\symdi{w\in A\cap N_G[u]^c\cap N_G[v]^c} N_{G} (w) \right) \symd \\
&&    \left(\Symdi{w\in A\cap N_G[u]^c\cap N_G[v]} N_{G} (w)  \symd N_G[u] \right) \symd
      \left(\Symdi{w\in A\cap N_G[u]\cap N_G[v]^c} N_{G} (w)  \symd N_G[v] \right) \\
  &=& \odd {G} A \symd \left(\Symdi{w\in A\cap N_G[v]} N_G[u]\right) \symd
      \left(\Symdi{w\in A\cap N_G[u]} N_G[v] \right)
\end{eqnarray*}
Notice that $\lvert A\cap N_G[u]\rvert \equiv 1 \bmod 2$ iff $u\in \codd G A$, and similarly for $v$. Hence we obtain \eqref{eq:odd-after-pivot} by case distinction.
\end{proof}

\begin{lemma} \label{lem:gflow-pivot}
If $(g,\prec)$ is a focused gFlow for $(G,I,O)$ then $(g',\prec)$ is a focused gFlow for ${(G\wedge uv \setminus \{u,v\},I,O)}$ where $\forall w \in \comp O \setminus \{u,v\}$, $g'(w):= g(w) \setminus  \{u,v\}$.
\end{lemma}

\begin{proof}
Every condition needed for $g'$ to be a focused gFlow is obvious except that $\forall w\in \comp O\backslash \{u,v\}, \odd {G\wedge uv \backslash u\backslash v} {g'(w)} \cap \comp O = \{w\}$, or equivalently, defining $O':=O\cup \{u,v\}$, $\forall w\in \comp{O'}, \odd {G\wedge uv} {g'(w)} \cap \comp{O'} = \{w\}$  

Let $w\in \comp{O'}$. Note that $\codd G {g(w)} = \odd G {g(w)} \symd g(w) = \{w\}\symd g(w) = g(w)\cup \{w\}$, and thus $u\in \codd G {g(w)} \iff u \in g(w)$ and similarly for $v$. Hence, using Lemma~\ref{lem:oddneighbourspivot}:
\begin{itemize}
\item If $u,v\in g(w)$, 
\begin{eqnarray*}
\odd {G\wedge uv}{g'(w)}\cap \comp{O'}&=&\odd {G\wedge uv}{g(w)\symd \{u,v\}}\cap \comp{O'}\\
&=&\odd {G\wedge uv}{g(w)}\cap \comp{O'}\symd \odd {G\wedge uv}{\{u,v\}}\cap \comp{O'}\\
(\textrm{Lemma \ref{lem:oddneighbourspivot}})
&=&\odd {G}{g(w)}\cap \comp{O'}\symd (N_G[u] \symd N_G[v])\cap \comp{O'}\symd \odd {G}{\{u,v\}}\cap \comp{O'}\\
&=&\odd {G}{g(w)}\cap \comp{O'}\symd \odd {G}{\{u,v\}}\cap \comp{O'} \symd \odd {G}{\{u,v\}}\cap \comp{O'}\\
&=&\{w\}
\end{eqnarray*}
\item If $u\in g(w)$ and $v\notin g(w)$,
\begin{eqnarray*}
\odd {G\wedge uv}{g'(w)}\cap \comp{O'}&=&\odd {G\wedge uv}{g(w)\symd \{u\}}\cap \comp{O'}\\
&=&\odd {G\wedge uv}{g(w)}\cap \comp{O'}\symd \odd {G\wedge uv}{\{u\}}\cap \comp{O'}\\
(\textrm{Lemma \ref{lem:oddneighbourspivot}})
&=&\odd {G}{g(w)}\cap \comp{O'} \symd N_G[v]\cap \comp{O'}\symd \odd {G}{\{v\}}\cap \comp{O'}\\
&=&\{w\}
\end{eqnarray*}
Here we have used that $\odd {G\wedge uv} {\{u\}} = N_G(v)$ and that $N_G[v]\cap \comp{O'} = N_G(v)\cap \comp{O'}$.
\item If $u\notin g(w)$ and $v\in g(w)$ we prove it similarly to the previous case.
\item If $u,v\notin g(w)$, 
\begin{eqnarray*}
\odd {G\wedge uv}{g'(w)}\cap \comp{O'}&=&\odd {G\wedge uv}{g(w)}\cap \comp{O'}\\
&=&\odd {G}{g(w)}\cap \comp{O'}\\ 
&=&\{w\}
\end{eqnarray*}
\end{itemize}
\end{proof}

%
%

\subsection{ZX reduction rules}
\label{sec:zx-reduction-rules}

\begin{lemma*}[\ref{lem:lc-simp}]~
  \ctikzfig{lc-simp}
\end{lemma*}
\begin{proof}
First of all, we will need the following equation:
  \begin{equation}\label{eq:s-state-eq}
  \hfill\begin{tikzpicture}
	\begin{pgfonlayer}{nodelayer}
		\node [style=Z phase dot] (0) at (-10, 0) {$\frac\pi2$};
		\node [style=none] (1) at (-8.75, 0) {};
		\node [style=none] (2) at (-8, 0) {=};
		\node [style=none] (3) at (-4.75, 0) {};
		\node [style=hadamard] (4) at (-5.75, 0) {};
		\node [style=X phase dot] (5) at (-6.75, 0) {$\frac\pi2$};
		\node [style=none] (6) at (1.5, 0) {};
		\node [style=X phase dot] (7) at (-3, 0) {$\frac\pi2$};
		\node [style=none] (8) at (-4, 0) {=};
		\node [style=X phase dot] (9) at (-1.75, 0) {-$\frac\pi2$};
		\node [style=Z phase dot] (10) at (-0.5, 0) {-$\frac\pi2$};
		\node [style=X phase dot] (11) at (0.75, 0) {-$\frac\pi2$};
		\node [style=none] (12) at (2.25, 0) {=};
		\node [style=X dot] (13) at (3, 0) {};
		\node [style=X phase dot] (14) at (5.25, 0) {-$\frac\pi2$};
		\node [style=Z phase dot] (15) at (4, 0) {-$\frac\pi2$};
		\node [style=none] (16) at (6, 0) {};
		\node [style=none] (17) at (6.75, 0) {$=$};
		\node [style=X dot] (18) at (7.5, 0) {};
		\node [style=X phase dot] (19) at (8.25, 0) {-$\frac\pi2$};
		\node [style=none] (20) at (9, 0) {};
		\node [style=none] (21) at (9.75, 0) {=};
		\node [style=X phase dot] (22) at (10.75, 0) {-$\frac\pi2$};
		\node [style=none] (23) at (11.75, 0) {};
	\end{pgfonlayer}
	\begin{pgfonlayer}{edgelayer}
		\draw (0) to (1.center);
		\draw (5) to (4);
		\draw (4) to (3.center);
		\draw (7) to (9);
		\draw (9) to (10);
		\draw (10) to (11);
		\draw (11) to (6.center);
		\draw (15) to (14);
		\draw (14) to (16.center);
		\draw (13) to (15);
		\draw (19) to (20.center);
		\draw (18) to (19);
		\draw (22) to (23.center);
	\end{pgfonlayer}
\end{tikzpicture}\hfill
  \end{equation}

  We pull out all of the phases via \SpiderRule then apply the local complementation rule \eqref{eq:gs-local-comp}:
  \ctikzfig{lc-simp-proof}
  Thanks to equation \eqref{eq:s-state-eq}, the topmost spider in the RHS above becomes an X-spider, with phase $\mp \pi/2$, which is cancelled out by $X_{\pm \pi/2}$ gate directly below it. The resulting (phaseless) X-spider copies and fuses with the neighbours:
  \ctikzfig{lc-simp-proof-2}
\end{proof}

\begin{lemma*}[\ref{lem:pivot-simp}]~
  \ctikzfig{pivot-simp}
\end{lemma*}
\begin{proof}
   We pull out all of the phases via \SpiderRule then apply the pivoting rule \eqref{eq:gs-pivot}:
   \ctikzfig{pivot-simp-proof}
   We then apply the colour-change rule to turn the Z-spiders with phases $j\pi$ and $k\pi$ into X-spiders, which copy and fuse with the neighbours of the original two vertices:
   \ctikzfig{pivot-simp-proof-2}
\end{proof}

\begin{proposition*}[\ref{prop:cnotgflow}]
  The following equation holds.
  \begin{equation}
  \begin{tikzpicture}
	\begin{pgfonlayer}{nodelayer}
		\node [style=none] (0) at (4.5, 1.5) {};
		\node [style=none] (1) at (4.5, 0.75) {};
		\node [style=none] (2) at (4.5, -1.75) {};
		\node [style=Z dot] (3) at (2.75, 1.5) {};
		\node [style=Z dot] (4) at (2.75, 0.75) {};
		\node [style=Z dot] (5) at (2.75, -1.75) {};
		\node [style=Z dot] (6) at (0, 1.5) {};
		\node [style=Z dot] (7) at (0, 0.75) {};
		\node [style=Z dot] (8) at (0, -1.5) {};
		\node [style=Z dot] (13) at (-4.5, 1.5) {};
		\node [style=Z dot] (14) at (-4.5, 0.5) {};
		\node [style=Z dot] (15) at (-4.5, -1.75) {};
		\node [style=none] (16) at (-6.25, 1.5) {};
		\node [style=none] (17) at (-6.25, 0.5) {};
		\node [style=none] (18) at (-6.25, -1.75) {};
		\node [style=none] (19) at (-3, 1.5) {};
		\node [style=none] (20) at (-3, 1) {};
		\node [style=none] (21) at (-3, 0) {};
		\node [style=none] (22) at (-3, -1) {};
		\node [style=none] (23) at (-1, -1) {};
		\node [style=none] (24) at (-1, 0.25) {};
		\node [style=none] (25) at (-1, 1.25) {};
		\node [style=none] (26) at (-1, 1.75) {};
		\node [style=none] (27) at (1, 1.25) {};
		\node [style=none] (28) at (1, 0.5) {};
		\node [style=none] (29) at (1, -0.75) {};
		\node [style=none] (30) at (1.75, -0.5) {};
		\node [style=none] (31) at (1.75, 0.5) {};
		\node [style=none] (32) at (1.75, 1.25) {};
		\node [style=none, rotate=90] (33) at (1.375, 0.5) {\ldots};
		\node [style=none] (34) at (3.25, -2.5) {};
		\node [style=none] (35) at (-0.5, -2.5) {};
		\node [style=none] (36) at (1.5, -3) {$M$};
		\node [style=none] (38) at (1, 0.75) {};
		\node [style=none] (39) at (1, -0.25) {};
		\node [style=none] (40) at (1.75, 0) {};
		\node [style=none] (41) at (1.75, 1) {};
		\node [style=none, rotate=90] (43) at (3.5, -0.25) {\ldots};
		\node [style=none, rotate=90] (44) at (-5.5, -0.5) {\ldots};
		\node [style=box, minimum height=2cm, minimum width=1cm] (45) at (-2, 0.25) {$D$};
		\node [style=none] (46) at (6.25, 0) {$=$};
		\node [style=none] (47) at (18.5, 1.5) {};
		\node [style=none] (48) at (18.5, 0.5) {};
		\node [style=none] (49) at (18.5, -1.75) {};
		\node [style=Z dot] (50) at (16.5, 1.5) {};
		\node [style=Z dot] (51) at (16.5, 0.5) {};
		\node [style=Z dot] (52) at (16.5, -1.75) {};
		\node [style=Z dot] (53) at (13.75, 1.5) {};
		\node [style=Z dot] (54) at (13.75, 0.5) {};
		\node [style=Z dot] (55) at (13.75, -1.5) {};
		\node [style=Z dot] (56) at (9.25, 1.5) {};
		\node [style=Z dot] (57) at (9.25, 0.5) {};
		\node [style=Z dot] (58) at (9.25, -1.5) {};
		\node [style=none] (59) at (7.75, 1.5) {};
		\node [style=none] (60) at (7.75, 0.5) {};
		\node [style=none] (61) at (7.75, -1.5) {};
		\node [style=none] (62) at (10.75, 1.5) {};
		\node [style=none] (63) at (10.75, 1) {};
		\node [style=none] (64) at (10.75, 0) {};
		\node [style=none] (65) at (10.75, -1) {};
		\node [style=none] (66) at (12.75, -1) {};
		\node [style=none] (67) at (12.75, 0) {};
		\node [style=none] (68) at (12.75, 1) {};
		\node [style=none] (69) at (12.75, 1.75) {};
		\node [style=none] (70) at (14.75, 1.25) {};
		\node [style=none] (71) at (14.75, 0) {};
		\node [style=none] (72) at (14.75, -1) {};
		\node [style=none] (73) at (15.5, -0.75) {};
		\node [style=none] (74) at (15.5, 0.25) {};
		\node [style=none] (75) at (15.5, 1.25) {};
		\node [style=none, rotate=90] (76) at (15.125, 0.25) {\ldots};
		\node [style=none] (77) at (17, -2.5) {};
		\node [style=none] (78) at (13.25, -2.5) {};
		\node [style=none] (79) at (15.25, -3) {$M^\prime$};
		\node [style=none] (80) at (14.75, 0.75) {};
		\node [style=none] (81) at (14.75, -0.5) {};
		\node [style=none] (82) at (15.5, -0.25) {};
		\node [style=none] (83) at (15.5, 0.75) {};
		\node [style=none, rotate=90] (84) at (17.5, -0.5) {\ldots};
		\node [style=none, rotate=90] (85) at (8.5, -0.5) {\ldots};
		\node [style=X dot] (86) at (17.5, 1.5) {};
		\node [style=Z dot] (87) at (17.5, 0.5) {};
		\node [style=box, minimum height=2cm, minimum width=1cm] (88) at (11.75, 0.25) {$D$};
	\end{pgfonlayer}
	\begin{pgfonlayer}{edgelayer}
		\draw (16.center) to (13);
		\draw (14) to (17.center);
		\draw (18.center) to (15);
		\draw (2.center) to (5);
		\draw (1.center) to (4);
		\draw (0.center) to (3);
		\draw [style=hadamard edge] (6) to (26.center);
		\draw [style=hadamard edge] (6) to (25.center);
		\draw [style=hadamard edge] (6) to (24.center);
		\draw [style=hadamard edge] (7) to (24.center);
		\draw [style=hadamard edge] (7) to (25.center);
		\draw [style=hadamard edge] (8) to (23.center);
		\draw [style=hadamard edge] (15) to (22.center);
		\draw [style=hadamard edge] (15) to (20.center);
		\draw [style=hadamard edge] (14) to (21.center);
		\draw [style=hadamard edge] (14) to (20.center);
		\draw [style=hadamard edge] (13) to (19.center);
		\draw [style=hadamard edge] (13) to (20.center);
		\draw [style=hadamard edge] (6) to (27.center);
		\draw [style=hadamard edge] (8) to (29.center);
		\draw [style=hadamard edge] (5) to (30.center);
		\draw [style=hadamard edge] (4) to (31.center);
		\draw [style=hadamard edge] (4) to (32.center);
		\draw [style=brace edge] (34.center) to (35.center);
		\draw [style=hadamard edge] (6) to (38.center);
		\draw [style=hadamard edge] (7) to (39.center);
		\draw [style=hadamard edge] (3) to (40.center);
		\draw [style=hadamard edge] (5) to (41.center);
		\draw [style=hadamard edge] (8) to (28.center);
		\draw (59.center) to (56);
		\draw (57) to (60.center);
		\draw (61.center) to (58);
		\draw (49.center) to (52);
		\draw (48.center) to (51);
		\draw (47.center) to (50);
		\draw [style=hadamard edge] (53) to (69.center);
		\draw [style=hadamard edge] (53) to (68.center);
		\draw [style=hadamard edge] (53) to (67.center);
		\draw [style=hadamard edge] (54) to (67.center);
		\draw [style=hadamard edge] (54) to (68.center);
		\draw [style=hadamard edge] (55) to (66.center);
		\draw [style=hadamard edge] (58) to (65.center);
		\draw [style=hadamard edge] (58) to (63.center);
		\draw [style=hadamard edge] (57) to (64.center);
		\draw [style=hadamard edge] (57) to (63.center);
		\draw [style=hadamard edge] (56) to (62.center);
		\draw [style=hadamard edge] (56) to (63.center);
		\draw [style=hadamard edge] (53) to (70.center);
		\draw [style=hadamard edge] (55) to (72.center);
		\draw [style=hadamard edge] (52) to (73.center);
		\draw [style=hadamard edge] (51) to (74.center);
		\draw [style=hadamard edge] (51) to (75.center);
		\draw [style=brace edge] (77.center) to (78.center);
		\draw [style=hadamard edge] (53) to (80.center);
		\draw [style=hadamard edge] (54) to (81.center);
		\draw [style=hadamard edge] (50) to (82.center);
		\draw [style=hadamard edge] (52) to (83.center);
		\draw [style=hadamard edge] (55) to (71.center);
		\draw (86) to (87);
	\end{pgfonlayer}
\end{tikzpicture} 
  \end{equation}
  Here $M$ describes the biadjacency matrix of the relevant vertices, and $M^\prime$ is the matrix produced by starting with $M$ and then adding row 2 to row 1, taking sums modulo 2. Furthermore, if the diagram on the LHS has a focused gFlow, then so does the RHS.
\end{proposition*}
\begin{proof}
  We will show how to transform the first diagram into the second in such a way that gFlow and equality is preserved at every step. For clarity we will not draw the entire diagram, but instead we focus on the relevant part. First of all we note that we can add CNOTs in the following way while preserving equality:
  \ctikzfig{cnot-pivot3}
  As we are only adding vertices at the outputs, it should be clear how the gFlow can be extended to incorporate these new vertices.

  Now let $A$ denote the set of vertices connected to the top vertex, but not to the vertex beneath it, $B$ the set of vertices connected to both, and $C$ the vertices connected only to the bottom one. Further restricting our view of the diagram to just these two lines, we see that we can apply a pivot rewrite as in \eqref{eq:pivot-simp}:
  \ctikzfig{cnot-pivot4}
  By Corollary~\ref{cor:simp-preserves-gflow} this rewrite preserves focused gFlow. Looking at the connectivity matrix, it is straightforward to see that the matrix $M$ has now been changed in exactly the way described.
\end{proof}

\subsection{Reduction of Local Cliffords}
\label{sec:reduction-lc}

This section will rely on the local complementation rule:
\begin{equation}\label{eq:gs-local-comp-app}
\begin{tikzpicture}
	\begin{pgfonlayer}{nodelayer}
		\node [style=Z dot] (0) at (-5.75, 0.5) {};
		\node [style=X phase dot] (1) at (-5.75, 1.5) {-$\frac\pi2$};
		\node [style=Z phase dot] (2) at (-1, 1.5) {$\frac\pi2$};
		\node [style=Z phase dot] (3) at (-2.5, 1.5) {$\frac\pi2$};
		\node [style=Z phase dot] (4) at (-3.5, 1.5) {$\frac\pi2$};
		\node [style=none] (5) at (-1, 2.25) {};
		\node [style=none] (6) at (-2.5, 2.25) {};
		\node [style=none] (7) at (-3.5, 2.25) {};
		\node [style=none] (8) at (-5.75, 2.25) {};
		\node [style=Z dot] (12) at (-2.5, 0.5) {};
		\node [style=Z dot] (13) at (-1, -0.5) {};
		\node [style=Z dot] (14) at (-3.5, -0.5) {};
		\node [style=Z dot] (15) at (-1, -0.5) {};
		\node [style=none] (16) at (-0.5, -1) {};
		\node [style=none] (17) at (-4, -1) {};
		\node [style=none] (18) at (-2.25, -1.5) {$N(u)$};
		\node [style=none] (19) at (-5.75, -0.25) {$u$};
		\node [style=none] (20) at (0.5, 0.5) {$=$};
		\node [style=Z dot] (21) at (2.25, 0.5) {};
		\node [style=Z dot] (22) at (5.5, 0.5) {};
		\node [style=Z dot] (23) at (7, -0.5) {};
		\node [style=Z dot] (24) at (4.5, -0.5) {};
		\node [style=Z dot] (25) at (7, -0.5) {};
		\node [style=none] (30) at (-1.75, 0.75) {...};
		\node [style=none] (31) at (7, 2.25) {};
		\node [style=none] (32) at (5.5, 2.25) {};
		\node [style=none] (33) at (4.5, 2.25) {};
		\node [style=none] (34) at (2.25, 2.25) {};
		\node [style=none] (35) at (6.25, 0.75) {...};
	\end{pgfonlayer}
	\begin{pgfonlayer}{edgelayer}
		\draw (1) to (0);
		\draw (8.center) to (1);
		\draw (2) to (5.center);
		\draw (6.center) to (3);
		\draw (7.center) to (4);
		\draw [style=hadamard edge] (0) to (15);
		\draw (15) to (2);
		\draw (12) to (3);
		\draw (14) to (4);
		\draw [style=hadamard edge] (0) to (14);
		\draw [style=hadamard edge] (0) to (12);
		\draw [style=hadamard edge] (12) to (15);
		\draw [style=brace edge] (16.center) to (17.center);
		\draw [style=hadamard edge] (21) to (25);
		\draw [style=hadamard edge] (21) to (24);
		\draw [style=hadamard edge] (21) to (22);
		\draw (21) to (34.center);
		\draw (24) to (33.center);
		\draw (22) to (32.center);
		\draw (25) to (31.center);
		\draw [style=hadamard edge] (24) to (22);
		\draw [style=hadamard edge] (24) to (25);
	\end{pgfonlayer}
\end{tikzpicture}
\end{equation}
to reduce the local Cliffords in a GS-LC form to a fixed set.

Let $\widetilde S := HSH$, i.e. an X-phase gate with phase $\pi/2$. Then, \eqref{eq:gs-local-comp-app} introduces a $\widetilde S^\dagger$ on a single output while introducing a $S$ on each of its neighbours.

\begin{theorem}
From the GS-LC form:
\begin{equation}\label{eq:gslc-app}
\begin{tikzpicture}
	\begin{pgfonlayer}{nodelayer}
		\node [style=box] (0) at (-3.5, 1.75) {LC};
		\node [style=none] (3) at (-5.25, 1.75) {};
		\node [style=none] (4) at (-5.25, 0.25) {};
		\node [style=none] (5) at (-5.25, -1.75) {};
		\node [style=box] (6) at (-3.5, 0.25) {LC};
		\node [style=box] (7) at (-3.5, -1.75) {LC};
		\node [style=Z dot] (8) at (-1.75, 1.75) {};
		\node [style=Z dot] (9) at (-1.75, 0.25) {};
		\node [style=Z dot] (10) at (-1.75, -1.75) {};
		\node [style=none] (11) at (-3.5, -0.75) {...};
		\node [style=none] (12) at (5.25, -1.75) {};
		\node [style=box] (13) at (3.5, 1.75) {LC};
		\node [style=none] (14) at (5.25, 0.25) {};
		\node [style=box] (15) at (3.5, 0.25) {LC};
		\node [style=none] (16) at (5.25, 1.75) {};
		\node [style=none] (17) at (3.5, -0.75) {...};
		\node [style=Z dot] (18) at (1.75, 1.75) {};
		\node [style=Z dot] (19) at (1.75, -1.75) {};
		\node [style=box] (20) at (3.5, -1.75) {LC};
		\node [style=Z dot] (21) at (1.75, 0.25) {};
		\node [style=none] (22) at (-0.75, 1.75) {};
		\node [style=none] (23) at (-0.75, 0.5) {};
		\node [style=none] (24) at (-0.75, -0.5) {};
		\node [style=none] (25) at (-0.75, -1.5) {};
		\node [style=none] (26) at (-0.75, -1.75) {};
		\node [style=none] (27) at (-0.75, 0) {};
		\node [style=none] (28) at (-0.75, 1.25) {};
		\node [style=none] (29) at (0.75, -1.75) {};
		\node [style=none] (30) at (0.75, 1.75) {};
		\node [style=none] (31) at (0.75, -1.75) {};
		\node [style=none] (32) at (0.75, 1.25) {};
		\node [style=none] (33) at (0.75, 0.75) {};
		\node [style=none] (34) at (0.75, -0.5) {};
		\node [style=none] (35) at (0.75, 0) {};
		\node [style=none] (36) at (0.75, -1.25) {};
		\node [style=none, rotate=90] (37) at (0, -0.25) {\Huge$\cdots$};
	\end{pgfonlayer}
	\begin{pgfonlayer}{edgelayer}
		\draw (3.center) to (0);
		\draw (4.center) to (6);
		\draw (5.center) to (7);
		\draw (0) to (8);
		\draw (6) to (9);
		\draw (7) to (10);
		\draw (16.center) to (13);
		\draw (14.center) to (15);
		\draw (12.center) to (20);
		\draw (13) to (18);
		\draw (15) to (21);
		\draw (20) to (19);
		\draw [style=hadamard edge] (8) to (22.center);
		\draw [style=hadamard edge] (9) to (23.center);
		\draw [style=hadamard edge] (9) to (24.center);
		\draw [style=hadamard edge] (10) to (25.center);
		\draw [style=hadamard edge] (10) to (26.center);
		\draw [style=hadamard edge] (9) to (27.center);
		\draw [style=hadamard edge] (8) to (28.center);
		\draw [style=hadamard edge] (18) to (30.center);
		\draw [style=hadamard edge] (18) to (32.center);
		\draw [style=hadamard edge] (18) to (33.center);
		\draw [style=hadamard edge] (21) to (35.center);
		\draw [style=hadamard edge] (21) to (34.center);
		\draw [style=hadamard edge] (19) to (36.center);
		\draw [style=hadamard edge] (19) to (31.center);
		\draw [style=hadamard edge] (19) to (29.center);
		\draw [style=hadamard edge] (8) to (9);
		\draw [style=hadamard edge, bend right=45, looseness=0.75] (8) to (10);
		\draw [style=hadamard edge] (21) to (19);
	\end{pgfonlayer}
\end{tikzpicture}
\end{equation}
it is possible to apply local complementation rule~\eqref{eq:gs-local-comp} until all of the local Cliffords on the inputs are in the set $\{S^n, H, ZH\}$ and the outputs to the set $\{S^n, H, HZ \}$.
\end{theorem}

\begin{proof}
In Ref.~\cite[Thm.~13]{Backens1}, Backens showed that, from a state (i.e. a diagram with only outputs) in GS-LC form, the rule~\eqref{eq:gs-local-comp-app} could be applied until all local Cliffords are in the set $\{ S^n, \widetilde S S, \widetilde S S^\dagger \}$ and no two qubits with an $\widetilde S$ gate are adjacent. If an output has $\widetilde S S$ as its local Clifford, we can apply \eqref{eq:gs-local-comp-app} from right to left 3 times to transform it into $\widetilde S S \widetilde S \propto H$. As the vertex is only connected to vertices that have a $S^n$ gate, this doesn't change the type of Clifford those neighbours have, and it also doesn't create new connections between vertices that both have a $\widetilde S$ gate.

Similarly, if the output has local Clifford $\widetilde S S^\dagger$, we can apply \eqref{eq:gs-local-comp-app} 1 time to transform it into $\widetilde S S^\dagger \widetilde S^\dagger \propto HZ$. Hence, the local Cliffords on all outputs for a GS-LC state can be taken from the set $\{ S^n, H, HZ \}$.

To go from GS-LC states considered by Backens to GS-LC maps of the form \eqref{eq:gslc-app}, we simply change some of the outputs into inputs via (partial) transpose. Hence input local Cliffords can be taken from the set $\{ (S^n)^T, H^T, (HZ)^T \} = \{ S^n, H, ZH \}$.
\end{proof}

\newpage

\section{Pseudo-code for extraction algorithm}
\label{sec:algorithms}

\makeatletter
\algrenewcommand\ALG@beginalgorithmic{\small}
\makeatother

\renewcommand\algorithmiccomment[1]{\hfill{\color{gray!80!green}$\triangleright$\,\textit{#1}}}

\begin{algorithm}[!htb]
\caption{Extraction algorithm}\label{alg:extraction}
\begin{algorithmic}[1]
\Procedure{Extract}{$D$}\Comment{input is graph-like diagram $D$}
  \State Init empty circuit $C$
  \State $G,I,O\gets $ Graph$(D)$\Comment{get the underlying graph of $D$}
  \State $F \gets O$ \Comment{initialise the frontier}
  \For{$v\in F$}
    \If{$v$ connected to output by Hadamard}
      \State $C\gets $ Hadamard(Qubit($v$))
      \State Remove Hadamard from $v$ to output
    \EndIf
    \If{$v$ has non-zero phase}
      \State $C\gets $ Phase-gate$(Phase(v),Qubit($v$))$
    \EndIf
  \EndFor
  \For{edge between $v$ and $w$ in $F$}
    \State $C\gets $ CZ(Qubit($v$), Qubit($w$))
    \State Remove edge between $v$ and $w$
  \EndFor

  \While{$\exists v\in D\backslash F$}\Comment{there are still vertices to be processed}
    \State $D,F,C \gets $ UpdateFrontier$(D,F,C)$
  \EndWhile
  \For{$v\in F$} \Comment{the only vertices still in $D$ are in $F$}
    \If{$v$ connected to input via Hadamard}
      \State $C\gets$ Hadamard(Qubit($v$))
    \EndIf
  \EndFor
  \State Perm $\gets$ Permutation of Qubits of $F$ to qubits of inputs
  \For{swap$(q_1,q_2)$ in PermutationAsSwaps(Perm)}
    \State $C\gets$ Swap$(q_1,q_2)$
  \EndFor
  \State \textbf{return} $C$
\EndProcedure

\Procedure{UpdateFrontier}{$D,F,C$}
  \State $N\gets $ Neighbours$(F)$
  \State $M \gets$ Biadjacency$(F,N)$ 
  \State $M^\prime \gets $ GaussReduce$(M)$ 
  \State Init $ws$ \Comment{initialise empty set $ws$}
  \For{row $r$ in $M^\prime$}
    \If{sum$(r) == 1$}\Comment{there is a single 1 on row $r$}
      \State Set $w$ to vertex corresponding to nonzero column of $r$
      \State Add $w$ to $ws$ \Comment{$w$ will be part of the new frontier}
    \EndIf
  \EndFor
  \State $M \gets $  Biadjacency($F,ws$) \Comment{smaller biadjacency matrix}
  \For{$(r_1,r_2) \in $ GaussRowOperations$(M)$}
    \State $C\gets $ CNOT(Qubit$(r_1)$, Qubit$(r_2)$)
    \State Update $D$ based on row operation
  \EndFor
  \For{$w\in ws$} \Comment{all $w$ now have a unique neighbour in $F$}
    \State $v \gets $ Unique neighbour of $w$ in $F$
    \State $C\gets $ Hadamard(Qubit$(v)$)
    \State $C\gets $ Phase-gate(Phase$(w)$,Qubit$(v)$)
    \State Remove phase from $w$
    \State Remove $v$ from $F$
    \State Add $w$ to $F$
  \EndFor
  \For{edge between $w_1$ and $w_2$ of $ws$}
    \State $C\gets $ CZ(Qubit($w_1$),Qubit$(w_2)$)
    \State Remove edge between $w_1$ and $w_2$
  \EndFor
  \State \textbf{return} $D,F,C$
\EndProcedure
\end{algorithmic}
\end{algorithm}

\end{document}